\documentclass{article}
\usepackage[utf8]{inputenc}
\usepackage[T1]{fontenc}
\usepackage[english]{babel}
\usepackage{amsmath,amsthm,amssymb,amsfonts}
\usepackage{xspace}
\usepackage{hyperref}
\usepackage{cleveref}
\usepackage{verbatim}
\usepackage{thm-restate}
\usepackage{tikz}
\usetikzlibrary{positioning, 
                quotes}

\usepackage{multirow,hhline}
\usepackage{enumerate}
\usepackage{cite}
\usepackage{caption}
\usepackage{subcaption}
\usepackage{float}
\usepackage{graphicx}
\usepackage{paralist}

\usepackage[textsize=footnotesize,color=green!40]{todonotes}

\usepackage[a4paper,            bindingoffset=0.2in,
           left=1.0in,            right=1.0in,
           top=1.0in,
           bottom=1.0in,
           footskip=.25in]{geometry}

\newcommand{\tw}{{\rm tw}}

\usepackage{bm}

\usepackage[vlined, ruled,linesnumbered]{algorithm2e}
\usepackage{algpseudocode}
\usepackage{thm-restate}

\usepackage{todonotes}

\newcommand{\hconstant}{3}

\newtheorem{theorem}{Theorem}
\crefname{theorem}{theorem}{theorems}
\newtheorem{lemma}[theorem]{Lemma}
\crefname{lemma}{lemma}{lemmas}
\newtheorem{proposition}[theorem]{Proposition}

\crefname{proposition}{proposition}{propositions}
\crefname{result}{result}{results}
\newtheorem{corollary}[theorem]{Corollary}
\crefname{corollary}{corollary}{corollaries}

\crefname{fact}{fact}{facts}
\newtheorem{observation}[theorem]{Observation}
\crefname{observation}{observation}{observations}

\crefname{question}{question}{questions}
\newtheorem{claim}[theorem]{Claim}
\crefname{claim}{claim}{claims}

\crefname{note}{note}{notes}

\crefname{conj}{conjecture}{conjectures}
\newtheorem{definition}[theorem]{Definition}
\crefname{definition}{definition}{definitions}
\newtheorem{remark}[theorem]{Remark}
\crefname{remark}{remark}{remarks}

\newcounter{claimcounter}
\numberwithin{claimcounter}{lemma}

\newenvironment{proofofclaim}{%
    
  \proof}{\endproof}

\tikzstyle{noeud}=[circle,inner sep=2, minimum size =3 pt, line width = 1pt, draw=black, fill=white]

\usepackage{hyperref}
\hypersetup{
  pdftitle = {(Almost-)Optimal FPT Algorithm and Kernel for {\sc $T$-Cycle} on Planar Graphs},
  pdfauthor = {H. Gahlawat, A. Rathod, M. Zehavi},
  colorlinks = true,
  linkcolor = black!30!red,
  citecolor = black!30!green
}

\title{(Almost-)Optimal FPT Algorithm and Kernel for {\sc $T$-Cycle} on Planar Graphs}
\author{Harmender Gahlawat\thanks{LIMOS, Université Clermont Auvergne, Clermont-Ferrand, France} \thanks{Laboratoire G-SCOP, Grenoble-INP, France} \\  \texttt{harmendergahlawat@gmail.com}
\and \and Abhishek Rathod\thanks{Ben-Gurion University of the Negev, Beersheba, Israel} \\  \texttt{arathod@post.bgu.ac.il}
\and \and Meirav Zehavi\footnotemark[3] \\  \texttt{meiravze@bgu.ac.il} }

\usepackage{xspace}

\usepackage{xcolor}

\newcommand{\tcycle}{$T$-\textsc{Cycle}\xspace}

\newcommand{\dispaths}{\textsc{Disjoint Paths}\xspace}

\newcounter{problemcounter}
\newcommand{\Pb}[4]{%
\begin{center}
  \begin{tabular}{|l|}%
  \hline
    \begin{minipage}[c]{0.95\textwidth}
      \smallskip%
      \par\noindent%
      #1%
      \par\noindent%
      \textbf{\textsf{Input}}: #2% 
      \par\noindent%
      \textbf{\textsf{#3}}: #4 
      \smallskip%
      \par\noindent%
    \end{minipage}
  \\\hline
  \end{tabular}%
\end{center}
}%

\begin{document}

\maketitle

\begin{abstract}
Research of cycles through specific vertices is a central topic in graph theory. In this context, we focus on a well-studied computational problem, {\sc $T$-Cycle}: given an undirected $n$-vertex graph $G$ and a set of  $k$ vertices $T\subseteq V(G)$ termed \textit{terminals}, the objective is to determine whether $G$ contains a simple cycle $C$ through all the terminals.
Our contribution is twofold: (i) We provide a $2^{O(\sqrt{k}\log k)}\cdot n$-time fixed-parameter deterministic algorithm for {\sc $T$-Cycle} on planar graphs; (ii) We provide a $k^{O(1)}\cdot n$-time deterministic kernelization algorithm for {\sc $T$-Cycle} on planar graphs where the produced instance is of size $k\log^{O(1)}k$. 

Both of our algorithms are optimal in terms of both $k$ and $n$ up to (poly)logarithmic factors in $k$ under the ETH.  In fact, our algorithms are the first subexponential-time fixed-parameter algorithm for {\sc $T$-Cycle} on planar graphs, as well as the first polynomial kernel for {\sc $T$-Cycle} on planar graphs. This substantially improves upon/expands the known literature on the parameterized complexity of the problem.
\end{abstract}

\section{Introduction}\label{S:intro}
Combinatorial properties of (simple) cycles passing through a specific set of vertices (along with, possibly, other vertices) in a graph---and, in particular, combinatorial properties of the graphs containing them---have been intensively explored for over six decades. As observed by Bj\"{o}rklund et al.~\cite{DBLP:conf/soda/BjorklundHT12}, the public interest in this topic can even be dated back to 1898, when Lewis Carroll challenged the readers of Vanity Fair with a riddle linked to this problem. However, the flurry of scientific results on this topic has, roughly, begun with Dirac's result~\cite{dirac1960abstrakten} from 1960, stating that given $k$ vertices in a $k$-connected (undirected) graph $G$, $G$ has a cycle through all of them. For a few illustrative examples of works in Graph Theory that followed up on this result, we refer to~\cite{bondy1981cycles,holton1982nine,kawarabayashi2004cycles,woodall1977circuits,erdHos1985any,sanders1996circuits,thomassen1977note,haggkvist1982circuits,kawarabayashi2002one}. Indeed, as stated by Kawarabayashi~\cite{kawarabayashi2008improved}: ``Since 1960’s, cycles through a vertex set or an edge set are one of central topics in all of graph theory.''

Computationally, given an undirected $n$-vertex graph $G$ and a set of  $k$ vertices $T\subseteq V(G)$ referred to as {\em terminals} (for any $k\in\{1,2,\ldots,n\}$), the objective of the {\sc $T$-Cycle} problem is to determine whether $G$ contains a (simple) cycle $C$ that passes through all (but not necessarily only) the terminals. 
Henceforth, such a cycle is referred to as a \emph{$T$-loop}.
Due to the immediate relation between the \tcycle and \dispaths problems (formalized later in this introduction), the {\sc $T$-Cycle} problem has been long known (since the seminal work of Robertson and Seymour~\cite{DBLP:journals/jct/RobertsonS95b}) to be {\em fixed-parameter tractable}, that is, solvable in time $f(k)\cdot n^c$ for some computable function $f$ of $k$ and some fixed constant $c$. However, here, the best known $f$ is galactic, satisfying $f(k)\geq 2^{2^{2^{2^{2^k}}}}$, and $c=2$~\cite{DBLP:journals/jct/KawarabayashiKR12}.  Still, already in 1980, LaPaugh and Rivest~\cite{DBLP:journals/jcss/LapaughR80} provided a linear time algorithm for {\sc $T$-Cycle} when $k = 3$, which involved no large hidden constants. 

The first breakthrough for the general case---having $k$ as part of the input---was established in 2008 by Kawarabayashi~\cite{kawarabayashi2008improved}, who proved that:\footnote{Kawarabayashi~\cite{kawarabayashi2008improved} states that when $k$ is  a fixed constant, the dependency on $n$ in the runtimes of his algorithms can be made  linear, but no details of a proof to support this statement are given.}
\begin{itemize}
\item The {\sc $T$-Cycle} problem is solvable in time $2^{2^{O(k^{10})}}\cdot n^2$. We remark that using the developments in~\cite{DBLP:journals/jacm/ChekuriC16} on the Grid-Minor Theorem that were first published shortly after the work of Kawarabayashi, it can be easily seen that his algorithm can be made to run in time $2^{O(k^c)}\cdot n^2$ for some large constant $c$.
\item Restricted to planar graphs, the {\sc $T$- Cycle} problem is solvable in time $2^{O(k\log^4 k)}\cdot n^2$.
\end{itemize}

We note that all algorithms mentioned so far are deterministic. 
In 2010, by making novel use of the idea behind the breakthrough technique of algebraic fingerprints~\cite{DBLP:journals/siamcomp/Bjorklund14,DBLP:journals/jcss/BjorklundHKK17}, Bj\"{o}rklund et al.~\cite{DBLP:conf/soda/BjorklundHT12} developed a randomized algorithm for {\sc $T$-Cycle} that runs in time $O(2^k\cdot n)$, which, prior to this manuscript, has remained the fastest known algorithm for this problem even if restricted to planar graphs. We note that 
however, prior to this manuscript, both of  Kawarabayashi's algorithms have remained the best-known deterministic ones for the problem even if restricted to planar graphs.

With the above context in mind, the first of the two  main contributions of our manuscript can be summarized in the following theorem. Specifically, we present the first subexponential-time fixed-parameter algorithm for the {\sc $T$-Cycle} problem on planar graphs.

\begin{restatable}{theorem}{mainFPT}\label{thm:main1}
Restricted to planar graphs, the {\sc $T$-Cycle} problem is solvable in time $2^{O(\sqrt{k}\log k)}\cdot n$.
\end{restatable}

Notably, under the Exponential Time Hypothesis (ETH), the time complexity in Theorem~\ref{thm:main1} is tight up to the logarithmic factor in the exponent, and, on general graphs, no sub-exponential-time (fixed-parameter or not) algorithm exists for {\sc $T$-Cycle}. This can be seen from the fact that {\sc Hamiltonicity} is the special case of {\sc $T$-Cycle} when $T$ is the entire vertex set of the graph, and {\sc Hamiltonicity} is not solvable in times $2^{o(\sqrt{n})}$ and $2^{o(n)}$ on planar and general graphs, respectively, under the ETH~\cite{DBLP:books/sp/CyganFKLMPPS15}.

The existence of a subexponential-time fixed-parameter algorithm for the {\sc $T$-Cycle} problem on planar graphs might seem surprising when considered in the context of two other well-known ``terminals-based'' problems on planar graphs: Specifically, under the ETH, the {\sc Steiner Tree} problem on planar graphs is not solvable in time $2^{o(k)}\cdot n^{O(1)}$~\cite{DBLP:conf/focs/MarxPP18}, and the \dispaths problem was only recently shown to be solvable in time $2^{k^{O(1)}}\cdot n$~\cite{cho2023parameterized}, where, in both cases, $k$ denotes the number of terminals.

Following-up on the work of Bj\"{o}rklund et al.~\cite{DBLP:conf/soda/BjorklundHT12}, Wahlstr\"{o}m~\cite{DBLP:conf/stacs/Wahlstrom13} studied the compressability of {\sc $T$-Cycle}. To discuss Wahlstr\"{o}m's contribution, let us first present the definitions of compression and kernelization. Formally, we say that a (decision) parameterized problem $\Pi$ admits a {\em compression} into a (decision, not necessarily parameterized) problem $\Pi'$ if there exists a polynomial-time algorithm (called a {\em compression algorithm}) that, given an instance $(I,k)$ of $\Pi$, produces an equivalent instance $J$ of $\Pi'$ such that $|J|\leq f(k)$ for some computable function $f$ of $k$. When $f$ is polynomial, then the compression is said to be {\em polynomial-sized} (or, for short, {\em polynomial}), and when $\Pi'=\Pi$, then the compression is termed {\em kernelization}. Kernelization is, perhaps, the most well-studied research subarea of Parameterized Complexity after that of fixed-parameter tractability~\cite{kernelbook}. In fact, kernelization was termed ``the lost continent of polynomial time''~\cite{DBLP:conf/iwpec/Fellows06} by one of the two fathers of the field of parameterized complexity. Indeed, kernelization is, essentially, the only mathematical framework to rigorously reason about the effectiveness of preprocessing procedures, which are ubiquitous in computer science in general.  Here, the focus is on polynomial-sized kernels (and compressions), since the existence of a (not necessarily polynomial) kernel for a problem is equivalent to that problem being fixed-parameter tractable~\cite{cai1997advice}, while, on the other hand, there exist numerous problems known to be fixed-parameter tractable but which do not admit a polynomial kernel unless the polynomial hierarchy collapses~\cite{kernelbook} (with the first examples having been discovered already in 2008~\cite{DBLP:journals/jcss/BodlaenderDFH09}).

Wahlstr\"{o}m~\cite{DBLP:conf/stacs/Wahlstrom13} proved that {\sc $T$-Cycle} admits a compression of size $O(k^3)$. The target problem is a somewhat artificial algebraic problem, which, in particular, is not known to be in NP, and hence this compression does not imply the existence of a polynomial-sized kernelization. Over the years, Wahlstr\"{o}m's result on the compressability of {\sc $T$-Cycle}  has become an extremely intriguing finding in the field of Parameterized Complexity: Since its discovery and until this day, {\sc $T$-Cycle} remains the {\em only} natural problem known to admit a polynomial-sized compression but not known to admit a polynomial-sized kernelization! In fact, the resolution of the kernelization complexity of {\sc $T$-Cycle}  (on general graphs) is one of the biggest problems in the field.
We remark that the ideas behind this compression also yield an alternative $O(2^k\cdot n)$-time algorithm for {\sc $T$-Cycle} (in~\cite{DBLP:conf/stacs/Wahlstrom13}), which is, similarly to Bj\"{o}rklund et al.~\cite{DBLP:conf/soda/BjorklundHT12}'s algorithm, also randomized and based on algebraic arguments.

With the above context in mind, the second of the two  main contributions of our manuscript can be summarized in the following theorem.

\begin{restatable}{theorem}{mainKernel}\label{thm:main2}
Restricted to planar graphs, the {\sc $T$- Cycle} problem admits a $k^{O(1)}\cdot n$-time kernelization algorithm of size $k\cdot \log^{O(1)}k$.
\end{restatable}

Again, the bounds in the theorem are almost tight in terms of both time complexity and, more importantly, the size of the reduced instance. First, we cannot expect a time complexity of $o(n)$ as we need to, at least, read the input. (We note that the $k^{O(1)}$ factor is essentially negligible because the time complexity of an exact algorithm to solve this problem afterwards will, anyway, have at least a $2^{\Omega(\sqrt{k})}$ factor under the ETH as argued earlier). Second, we cannot expect the reduced instance to be of size $o(k)$ (or even of size $k-1$), else we can repeatedly reapply the kernelization algorithm until it will yield a constant-sized instance that is solvable in constant time, thereby solving the (NP-hard) problem in polynomial time.

As a side note, we remark that our proof might shed light on the kernelization complexity of {\sc $T$-Cycle} on general graphs as well: We believe that the scheme that we employ---on a high-level, the usage of a polynomial-time treewidth reduction to a polynomial function of $k$ coupled with a polynomial kernel for the combined parameter $k$ plus treewidth, and on a lower level, the specific arguments used in our proofs to implement both procedures---{\em might} be liftable to general graphs.

Combining Theorems~\ref{thm:main1} and~\ref{thm:main2}, we conclude the following corollary, which is, up to a logarithmic factor in the exponent, essentially the best one can hope for.

\begin{corollary}
Restricted to planar graphs, the {\sc $T$-Cycle} problem is solvable in time $2^{O(\sqrt{k}\log k)} + k^{O(1)}\cdot n$.
\end{corollary}

\subsection{Related Problems} We first note that when the {\sc $T$-Cycle} problem is extended to directed graphs, it becomes NP-hard already when $k=2$~\cite{fortune1980directed}.
Various other extensions/generalizations of {\sc $T$-Cycle} were studied in the literature, including from the perspective of parameterized complexity. For example, Kawarabayashi et al.~\cite{kawarabayashi2010recognizing} developed an algorithm for detecting a $T$-loop whose length has a given parity; for fixed $k$, the running time is polynomial in $n$, but the dependency on $k$ is unspecified. For another example, Kobayashi and Kawarabayashi~\cite{kobayashi2009algorithms} developed an algorithm for detecting an induced $T$-loop in a planar graph in time $2^{O(k^{3/2}\log k)}n^2$.  The survey of additional results of this nature is beyond the scope of this paper.

Highly relevant to the {\sc $T$-Cycle} problem is the \dispaths problem. Here, given a graph $G$ and a set $T=\{(s_i,t_i): i\in\{1,2,\ldots,k\}\}$ of pairs of vertices (termed terminals) in $G$, the objective is to determine whether $G$ contains $k$ vertex-disjoint\footnote{With the exception that a vertex can occur in multiple paths if it is an endpoint in all of them.} paths $P_1,P_2,\ldots,P_k$ so that the endpoints of $P_i$ are $s_i$ and $t_i$. The {\sc $T$-Cycle} problem can be reduced to the \dispaths problem with an overhead of $k!=k^{O(k)}$: Given an instance $(G,T)$ of {\sc $T$-Cycle}, for every ordering $v_0,v_1,\ldots,v_{k-1}$ of $T$, create an instance $(G,T')$ of \dispaths where $T'=\{(v_i,v_{(i+1)\mod k}): i\in\{0,1,\ldots,k-1\}\}$; then, $(G,T)$ is a yes-instance of {\sc $T$-Cycle} if and only if at least one of the constructed instances of \dispaths is. Being foundational to the entire graph minor theory~\cite{DBLP:conf/birthday/Lokshtanov0Z20}, the \dispaths problem is of great importance to both algorithm design and graph theory, and has been intensively studied for several decades. Most relevant to this manuscript is to mention that, currently, the fastest algorithms for \dispaths run in $f(k)\cdot n^2$ time for $f(k)>2^{2^{2^{2^{2^k}}}}$ on general graphs~\cite{DBLP:journals/jct/KawarabayashiKR12}, and in $2^{k^{O(1)}}\cdot n$ time on planar graphs~\cite{cho2023parameterized}. It is worth mentioning that  Korhonen, Pilipczuk, and Stamoulis~\cite{korhonen2024minor} recently provided an FPT algorithm for \textsc{Disjoint Paths} with almost-linear running time dependence on the number of vertices and edges, i.e., with running time $f(k)\cdot (m+n)^{1+o(1)}$. Also, the \dispaths problem does not admit a polynomial kernel w.r.t.~$k$ even on planar graphs~\cite{DBLP:conf/focs/0001Z23}, but it admits a polynomial kernel w.r.t.~$k$ plus the treewidth of the input graph on planar graphs~\cite{DBLP:conf/focs/0001Z23} (on general graphs, it is unknown whether \dispaths admits a polynomial kernel w.r.t.~$k$ plus treewidth).

Also quite relevant to the {\sc $T$-Cycle} problem is the {\sc $k$-Cycle} (or {\sc $k$-Path}) problem, which is also, perhaps, the second or third most well-studied problem in Parameterized Complexity (see, e.g., the dozens of papers cited in~\cite{DBLP:conf/soda/LokshtanovSZ21}). Here, given a graph $G$ and a nonnegative integer $k$, the objective is to determine whether $G$ contains a simple cycle (or path) on $k$ vertices. The techniques used to solve {\sc $k$-Cycle} can be trivially lifted to solve {\sc $T$-Cycle} when parameterized by the {\em sought size of a solution $\ell$} (given as an extra parameter in the input), which can be arbitrarily larger than $T$ (e.g., in the extreme case where the entire graph is a single cycle, $\ell=n$ while $k=1$). In particular, this means that, on general graphs, {\sc $T$-Cycle} is solvable in time $O(2^k 1.66^{\ell-k}\cdot (n+m))$ by a randomized algorithm~\cite{DBLP:journals/jcss/BjorklundHKK17} (or time $O(2.56^k 1.66^{\ell-k}\cdot (n+m))$ by a deterministic algorithm~\cite{DBLP:journals/tcs/Tsur19b}), and on planar graphs,  {\sc $T$-Cycle} is solvable in time $2^{O(\sqrt{\ell})}\cdot (n+m)$ by a deterministic algorithm~(e.g., via~\cite{DBLP:journals/jcss/DornFT12}).

\subsection{Techniques} We present a high-level overview of our proofs in the next section. Here, we only briefly highlight (in a very informal manner) the novelties in our proofs, and put them in the context of the known literature.  Essentially, the following theorem reveals some of the core combinatorial arguments that underlie our algorithmic results. Thus, we believe that this result may be of independent interest. 

\begin{restatable}{theorem}{TWReduction}\label{thm:twreduction}
There exists a $k^{O(1)}\cdot n$-time algorithm that, given an instance $(G,T)$ of {\sc $T$-Cycle} on planar graphs, outputs an equivalent instance $(G',T)$ of {\sc $T$-Cycle} on planar graphs along with a set of vertices $U$ such that:
\begin{enumerate}
\item $G'$ is a subgraph of $G$ whose treewidth is bounded by $O(\sqrt{k}\log k)$, and
\item $|U|\in O(k\log k)$ and the treewidth of $G'-U$ is bounded by $O(\log k)$.
\end{enumerate}
\end{restatable}

Indeed, from here, Theorem~\ref{thm:main1} directly follows by using a $2^{O(tw)}\cdot n$-time algorithm for {\sc $T$-Cycle} parameterized by the treewidth $tw$ of the planar graph $G'$ as a black box~(e.g., via~\cite{DBLP:journals/jcss/DornFT12}), while Theorem~\ref{thm:main2} requires additional non-trivial work, described later in this section.

Our proof of Theorem~\ref{thm:twreduction} consists of three main ingredients. Towards the first ingredient, we consider a sequence of concentric cycles, and classify the ``segments'' of an (unknown) solution $T$-loop that go inside this sequence into types, similarly to how it is done in~\cite{JCTB} for the \dispaths problem. Here, a cheap solution is one that uses as few edges as possible that do not belong to the sequence of concentric cycles. Then, we prove the following result, which is the first ingredient.
\begin{lemma}[Informal Statement of \Cref{L:main}]\label{lem:segmentTypes}
For a ``cheap'' solution $T$-loop and a sequence of concentric cycles that does not contain any terminal,  there can be at most constantly many segments of the same type.
\end{lemma}

We note that in~\cite{JCTB}, it is proved that the number of segments of the same type for a \dispaths solution is  $2^{O(k)}$. Our proof of Lemma~\ref{lem:segmentTypes} is based on an elegant (and novel)  re-routing argument   (see Section~\ref{sec:overview}), which has no relation to that in~\cite{JCTB}.  This argument is, in particular, perhaps the cornerstone  of the proof of Theorem~\ref{thm:twreduction}. 

With Lemma~\ref{lem:segmentTypes} at hand, we then turn to prove our second ingredient. 

\begin{lemma}[Informal Statement of \Cref{thm:logarithmic}]\label{lem:segmentNumber}
For a ``cheap'' solution $T$-loop and a sequence of concentric cycles that does not contain any terminal,  no segment goes more than $O(\log k)$ cycles deep into the sequence.
\end{lemma}

We note that in~\cite{JCTB}, it is proved that no segment (for a \dispaths solution) goes more than $2^{O(k)}$ cycles deep into the sequence. For the proof of our lemma, we make use of one intermediate result from~\cite{JCTB} about \dispaths solutions, which states that there can be  $O(k)$ segments of different types. However, besides that, the proof of our lemma is different, based on a simple analysis of the ``segment-tree'' induced by the segments (see Section~\ref{sec:overview}). We note that  rerouting arguments have been used for other terminal-based routing problems~\cite{kobayashi2009algorithms}. But, in our knowledge, our rerouting arguments are first to achieve a sublinear bound on the number of concentric cycles used by a cheap solution. Further, we believe that another strength of our result is that the arguments used to prove Lemma~\ref{lem:segmentNumber} are rather simple, and therefore easy to reuse in future works.

 In particular, Lemma~\ref{lem:segmentNumber} implies that within a $c\log k$-sized (for some constant $c$) sequence of concentric cycles that does not contain any terminal, every vertex that lies inside the innermost cycles is irrelevant, that is, can be deleted without turning a yes-instance into a no-instance. In particular, this implies that within a $c\log k\times c\log k$-grid minor that does not contain any terminal, the ``middle-most'' vertex is irrelevant. Here, it is important to mention that the main component in the proof of Kawarabayashi~\cite{kawarabayashi2008improved} is a lemma that (essentially) states that within a $ck\times ck$-grid minor  that does not contain any terminal, the ``middle-most'' vertex is irrelevant. Our proof (that is, the arguments briefly described so far) is completely different, and, in particular, yields an exponential improvement (from $ck$ to $c\log k$).
 
 From here, it is easy to provide a $k^{O(1)}\cdot n^2$-time algorithm that given an instance $(G,T)$ of {\sc $T$-Cycle} on planar graphs, outputs an equivalent instance $(G',T)$ of {\sc $T$-Cycle} on planar graphs such that $G'$ is a subgraph of $G$ whose treewidth is bounded by $O(\sqrt{k}\log k)$. Indeed, this can be done by, as long as possible, computing a $c\log k\times c\log k$ grid minor that does not contain any terminal, and removing its ``middle-most'' vertex, where each iteration requires time $k^{O(1)}\cdot n$, and where $O(n)$ iterations are performed in total.
 
 To reduce the time complexity to be linear in $n$, we make use of Reed's~\cite{reedLinear} approach of dividing the plane and working on each ``piece'' simultaneously, together with an argument by Cho et al.~\cite{cho2023parameterized} that further simplifies the pieces such that none of the piece contains a sequence of ``too many'' concentric cycles. Here, the proof is mostly based on  a similar approach described in the aforementioned two papers with minor adaptations. However, we need to add one novel and critical step. Briefly,  by applying the  known approach (with the minor adaptations to our case) we directly obtain pieces such that: {\em (i)} the total number of vertices on their boundaries is $O(k\log k)$; {\em (ii)} no piece contains a sequence of more than $c\log k$ many concentric cycles. This already yields a set of vertices $U$ such that  $|U|\in O(k\log k)$ and the treewidth of $G'-U$ (where $G'$ is the current graph) is bounded by $O(\log k)$, by taking the union of boundaries (that is what we need for the second item in Theorem~\ref{thm:twreduction}). However, from this, the aforementioned two papers (i.e., \cite{reedLinear} and~\cite{cho2023parameterized}) directly get their tree decomposition of $G'$ by (implicitly) creating a tree decomposition for each piece, gluing them together, and putting the vertices of $U$ in all bags. For these two previous papers, this was sufficient---they get pieces of treewidth $2^{O(k)}$ and $U$ of size $2^{O(k)}$, and thus by doing this, they do not lose anything. However, for us, this yields a tree decomposition of width $O(k\log k)$, while we need width $O(\sqrt{k}\log k)$. We overcome this by adding a new step, where we attempt to remove (possibly many) boundary vertices, and after performing this step, we are able to directly argue that the graph $G'$ itself (rather than just each one of its pieces individually) does not have a sequence of more than $c\log k$ many concentric cycles that do not contain any terminal  (see Section~\ref{sec:overview}). Thus, we get the bound $O(\sqrt{k}\log k)$.
 
Lastly, having Theorem~\ref{thm:twreduction} at hand, let us address the additional work we need to perform in order to prove Theorem~\ref{thm:main2}. Here, we use two powerful theorems/machineries:
\begin{itemize}
\item Machinery I~\cite{bodlaender2016meta} (see also~\cite{kernelbook}): Suppose we are given a set of vertices $U$ with $|U|\in O(k\log k)$ so that the treewidth of $G'-U$ is bounded by $O(\log k)$. Then, one can construct in $k^{O(1)}\cdot n$-time a so-called protrusion decomposition of $G'$ with ``very good parameters''. That is, roughly speaking, the vertex set of $G'$ can be partitioned into a ``universal'' part $U^\star\supseteq U$ of size $O(k\log^2 k)$ and $O(k\log^2 k)$ additional ``protrusion'' parts such that each protrusion part induces a graph of treewidth $O(\log k)$, has  $O(\log k)$ many vertices that have neighbors in $U^\star$, while having no neighbors in the other protrusions.
\item Machinery II~\cite{DBLP:conf/focs/0001Z23}:  A generalization of the \dispaths problem on planar graphs, where, roughly speaking, we need to preserve the yes/no answer for every possible pairing of the terminals vertices, admits a kernel of size $O((k'\cdot w')^{12})$ where $k'$ is the number of terminals and $w'$ is the treewidth of the graph, and the output graph is a minor of the input graph.
\end{itemize}

Thus, we can compute the protrusion decomposition using the set obtained from Theorem~\ref{thm:main1} (which, in particular, contains all terminals), and then kernelize each of the protrusions individually using Machinery II, where the terminal set is defined as those vertices having neighbors outside the protrusion. However, this will not yield a time complexity of $k^{O(1)}\cdot n$ since Machinery II only specifies a polynomial (rather than linear) dependency on $n$, and, furthermore, this will yield a huge polylogarithmic dependency on $k$ in the size of the reduced instance.

To overcome this, we employ two novel twists in this approach.  First, we create a ``nested'' protrusion decomposition, or, alternatively, one can think about this as ``bootstrapping'' the entire proof for each protrusion. In more details, for each protrusion $P$, we yet again apply all the arguments in our proofs so far to now find a set of vertices $U_P$ with $|U_P|\in O(\log k\log\log k)$ so that the treewidth of $P-U_P$ is bounded by $O(\log\log k)$. Then, for each protrusion, we yet again compute a protrusion decomposition, but now with each ``subprotrusion'' having only $O(\log\log k)$ many vertices with neighbors outside the subprotrusion, and with treewidth $O(\log\log k)$. Then, the second twist is that instead of directly kernelizing the subprotrusion, we think of Machinery II as an existential rather than an algorithmic result---we know that there {\em exists} a $\log^{O(1)}\log k$-sized graph that can replace our subprotrusion. Then, we simply guess this replacement! That is, we generate all possible $\log^{O(1)}\log k$-sized graphs. For each generated graph $P'$, we test whether: 
\begin{enumerate}
\item $P'$ is equivalent to the subprotrusion with the same terminal set. That is, we must preserve the yes/no answers to the instances corresponding to all possible terminal pairings.
\item $P'$ is a minor of $P$.
\end{enumerate}

This test can be done using a standard dynamic programming-based algorithm over tree decomposition, since both the treewidth of the subprotrusion and the size of $P'$ are bounded by $\log^{O(1)}\log k$. Specifically, we have $2^{\log^{O(1)}\log k}<k^{O(1)}$ many guesses for the replacement, and the ``validity'' of each can be checked in time $k^{O(1)}\cdot n_P$ where $n_P$ is the number of vertices in $P$. We believe that this  concept of ``nested'' protrusion decomposition can find applications  in improving  kernelization algorithms, in terms of both  kernel size and running time, for other terminal-based problems.

Additional discussion and related open problems can be found in \Cref{sec:conc}.

\section{Overview of Our Proofs}\label{sec:overview}

In this section, we provide an overview of our proofs. We will only informally define notions  required to explain the overview. These notions are formally defined in later sections. Furthermore, to ease the presentation, we borrow some figures here from later sections, which will appear again in their respective sections. 
\subsection{FPT Algorithm}
Our first contribution is an FPT algorithm for \tcycle in planar graphs that has a running time subexponential in $k$ (which is, $2^{O(\sqrt{k}\log k)}$) and linear in $n$. Since \tcycle can be solved in $2^{O(\tw)} n$ time (e.g., using~\cite{DBLP:journals/jcss/DornFT12}), a natural way to attack our problem is via the technique of {\em treewidth reduction} by removal of  some {\em irrelevant vertices}. In fact, this technique was used by Kawarabayashi~\cite{kawarabayashi2008improved} to reduce the treewidth to be polynomial in $k$. To get the algorithm that we desire, we need to
\begin{enumerate}
    \item reduce the treewidth to $O(\sqrt{k} \log k)$ by removal of some irrelevant vertices, and
    \item  to ensure that this entire removal procedure is performed in $k^{O(1)}\cdot n$ time.
\end{enumerate}

\smallskip
\noindent\textbf{Irrelevant vertices:} We begin with establishing that if there is a vertex $v$ that is ``sufficiently separated'' from each vertex in $T$, then $v$ is irrelevant. 
%Intuitively, we say that a cycle $C$ \emph{separates} \todo{A: not used anywhere} $v$ from $T$  if $x$ is in the interior of $C$ and each vertex of $T$ is in the exterior of $C$. 
A set of cycles $\mathcal{C} = (C_0,\ldots,C_r)$ is \textit{concentric} if the cycles in $\mathcal{C}$ are vertex disjoint and cycle $C_{i-1}$ is contained in the interior of $C_i$ (for $i\in [r]$). See Figure~\ref{fig:OC1} for an illustration. If there exists a sequence of $\ell+1$ concentric cycles $C_0,\ldots,C_{\ell}$ such that $v$ in the interior of $C_0$ and $T$ is in the exterior of $C_{\ell}$, then we say that $v$ is $\ell$-\textit{isolated}. Let $D_i$ denote the disk corresponding to the interior of $C_i$. 

\begin{figure}
\centering
\begin{subfigure}{.5\textwidth}
  \centering
  \captionsetup{justification=centering}
    \includegraphics[width=.92\linewidth]{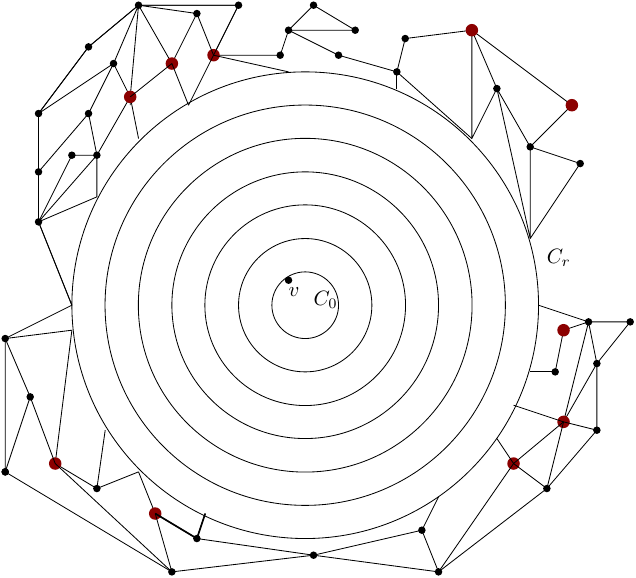}
  \caption{Concentric cycles.}
  \label{fig:OC1}
\end{subfigure}%
\begin{subfigure}{.5\textwidth}
  \centering
  \captionsetup{justification=centering}
  \includegraphics[width=.92\linewidth]{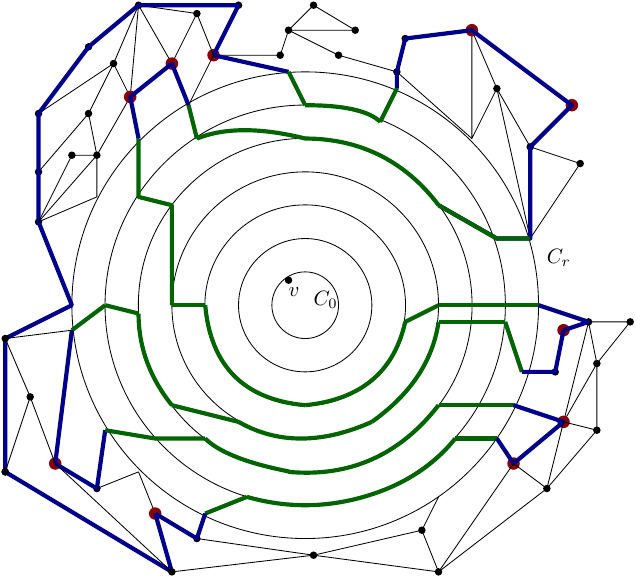}
  \caption{CL-configuration.}
  \label{fig:OC2}
\end{subfigure}
\caption{(a) A set $\mathcal{C}$  of concentric cycles and $v$ is $r$-isolated. (b) A CL-configuration $(\mathcal{C},L)$. The segments of $L$ are represented in  green and the edges of $L$ outside of $D_r$ are shown in blue. Hence, the green edges along with the blue edges combine to give the loop $L$. In both subfigures, the terminal vertices are highlighted in red.} 
\label{fig:test}
\end{figure}

The first component of our algorithm is a proof that establishes that there exists some $g(k)\in O (\log k)$, which we explicitly compute, such that each $g(k)$-isolated vertex is irrelevant, and hence, can be removed from $G$ safely. To this end, we use a notion called {\em CL-configuration}, which is a pair $\mathcal{Q}=(\mathcal{C},L)$, where $\mathcal{C} =(C_0,\ldots,C_r)$ is a set of concentric cycles such that $D_r\cap T =\emptyset$, and $L$ is a $T$-loop. See Figure~\ref{fig:OC2} for an illustration. The connected components of $D_i\cap L$ are said to be $C_i$-\textit{segments}. Intuitively, we say that a $C_j$-segment $S_a$ is in the \textit{zone} of a $C_j$-segment $S_b$ if $S_a$ lies in the region encompassed by $S_b$ and the cycle $C_j$. This can be used to define a partial order $\prec$ on the $C_j$-segments of the set of CL-configuration $\mathcal{Q}$ as follows: we say that $S_a \prec S_b $ if and only if $S_a$ lies in the zone of $S_b$. We need the following notion of segment types (see Figure~\ref{fig:sa} for an illustration.).
\begin{definition}[{\bf Segment types}]
  Let $\mathcal{Q}=(\mathcal{C},L)$ be a depth $r$  CL-configuration of $G$. Moreover, let $S_{1}$ and $S_{2}$ be two $C_j$-segments of ${\cal Q}$ such that $S_i$, for $i\in [2]$, has endpoints $u_i$ and $v_i$. We say that $S_{1}$ and $S_{2}$ have the {\em same $C_j$-type} if:
 \begin{enumerate}
     \item there exist paths $P$ and $P'$  on $C_j$
connecting an endpoint of $S_{1}$ with an endpoint of $S_{2}$ 
such that these paths do not pass through the other endpoints of $S_1$ and $S_2$, 
    \item no  segment of ${\cal Q}$ has both
endpoints on $P$ or on $P'$, and
\item the closed-interior of  the cycle $P\cup S_{1}\cup P'\cup S_{2}$ does not contain
the disk $D_{0}$. 
 \end{enumerate}
A \emph{$C_j$-type of segment} is an equivalence class of segments of ${\cal Q}$.   See Figure~\ref{fig:sb} for an illustration of segment types.
\end{definition}

\begin{figure}[H]
\centering
\begin{subfigure}{.5\textwidth}
  \centering
  \captionsetup{justification=centering}
  \includegraphics[width=0.92\linewidth]{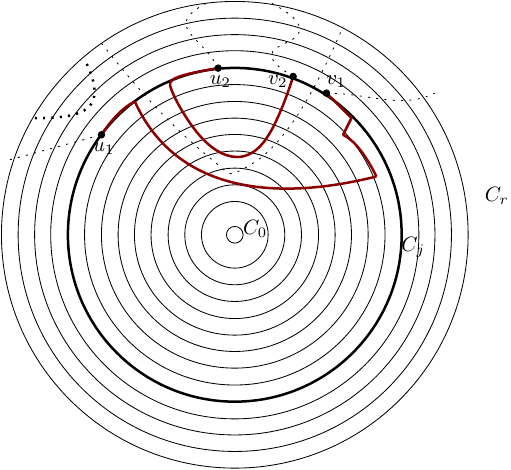}
  \caption{Segments with the same type.}
  \label{fig:sa}
\end{subfigure}%
\begin{subfigure}{.5\textwidth}
  \centering
  \captionsetup{justification=centering}
  \includegraphics[width=0.92\linewidth]{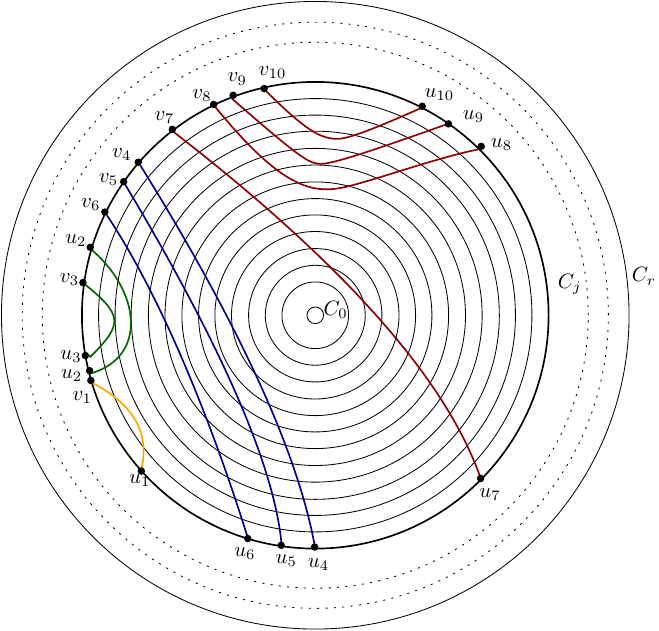}
  \caption{Segment types.}
  \label{fig:sb}
\end{subfigure}
\caption{(a) $S_1$ and $S_2$ are $C_j$-segments with endpoints $u_1,v_1$ and $u_2,v_2$, respectively. Moreover, $P$ and $P'$ are the $(u_1,u_2)$-path and $(v_1,v_2)$-path along $C_j$, respectively. Here, $S_1$ and $S_2$ have the same $C_j$-type. (b) Here, segments $S_1,\ldots,S_{10}$ are $C_j$-segments such that segment $S_i$ has endpoints $u_i$ and $v_i$.  Segments $S_4$ and $S_7$ are not of the same $C_j$-type because of condition~(3) in Definition~\ref{D:ST}, and $S_6$ and $S_2$ are not of the same $C_j$-type because of condition~(2) in Definition~\ref{D:ST}. Further, when $S_6$ and $S_2$ are restricted to $D_{j-2}$, they have the same $C_{j-2}$-type. Finally, $S_{10}\prec S_9 \prec S_8 \prec S_7$, and similarly, $S_6\prec S_5\prec S_4$.}
\label{fig:seg}
\end{figure}

Finally, for any $T$-loop, say $L$, we define a \textit{cost} function that corresponds to the number of edges in $L$ that are not contained in any cycle of $\mathcal{C}$. The main idea behind this is that a $T$-loop has to ``pay'' every time it goes to a ``deeper'' cycle of $\mathcal{C}$. Then, a $T$-loop $L$ is \textit{cheap} if there is no other loop with cost strictly lesser than that of $L$. We would then like to show that any cheap $T$-loop will not go ``very deep'' in $\mathcal{C}$, deeming the deeper cycles of $\mathcal{C}$  irrelevant.

Now, consider a CL-configuration $\mathcal{Q} = (\mathcal{C},L)$ of depth $r$. We are ready to explain our first main lemma, which says that if $L$ is a cheap $T$-loop, then there cannot be more than 3 segments with the same $C_j$-type, for  $j\in \{0,\dots,r\}$. To prove this, we show that if we have more than 3 segments with the same $C_j$-type, then we can reroute $L$ in such a manner that we get a new $T$-loop $L'$ which is cheaper than $L$. See Figure~\ref{fig:exchange} for an illustration of the intuition. This gives us Lemma~\ref{lem:segmentTypes} from Section~\ref{S:intro}.

\begin{figure}
\centering
\begin{subfigure}{.5\textwidth}
  \centering
  \captionsetup{justification=centering}
  \includegraphics[width=.95\linewidth]{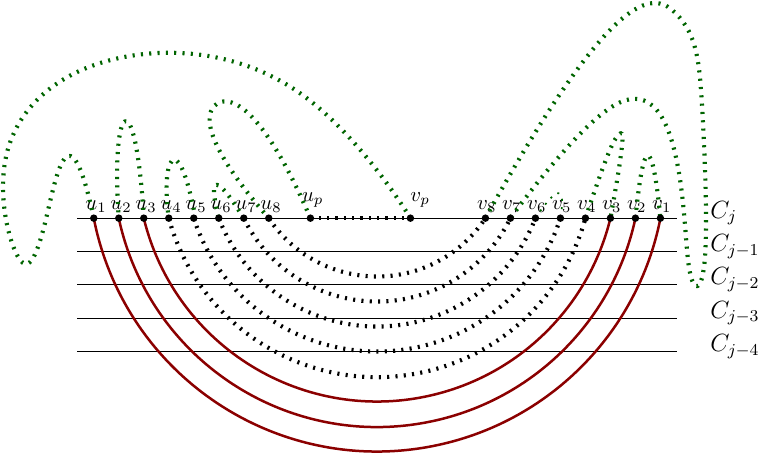}
  \caption{$L$}
  \label{fig:O}
\end{subfigure}%
\begin{subfigure}{.5\textwidth}
  \centering
  \captionsetup{justification=centering}
  \includegraphics[width=.95\linewidth]{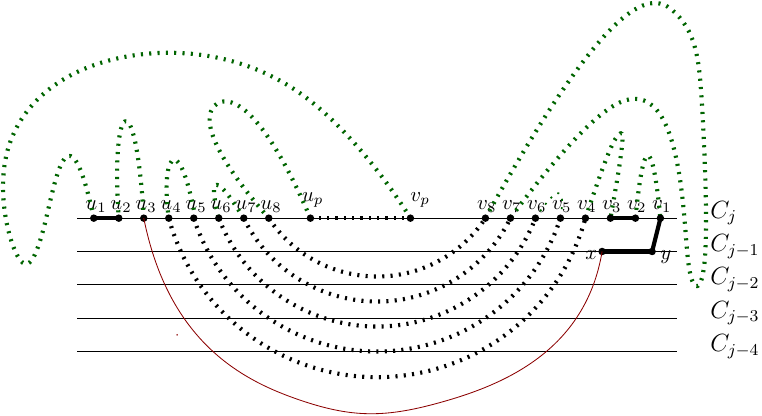}
  \caption{$L'$}
  \label{fig:O2}
\end{subfigure}
\caption{(a) Here, $S_i$ is the $C_j$ segment with endpoints $u_i$ and $v_i$, segments $S_1,S_2,S_3$ are illustrated in red, and segments $S_1,\ldots,S_8$ have the same $C_j$-type. (b) In the rerouting, we remove the segment $S_2$ completely (reducing the cost) and use a part of the segment $S_1$ (not increasing the cost), and additionally use some paths along concentric cycles $C_j$ and $C_{j-1}$ that have cost 0.}
\label{fig:exchange}
\end{figure}

The next component of the proof is to establish, using our above result, that a cheap $T$-loop will not go beyond $O(\log k)$ levels deep, and we compute a $g(k)\in O(\log k)$ such that every vertex that is $g(k)$-isolated, i.e., contained in $D_{r-g(k)-1}$ is irrelevant. To this end, we use a novel idea based on segment forests. In the discussion that follows, $\mathcal{Q}$ is a CL-configutation of depth $r$.

 For every $j \in [r]$, we use the hierarchy induced by $\prec$ to associate a forest structure called a \emph{segment forest}. To begin with, if a $C_j$-segment $S_a$ is in the zone of a $C_j$ segment $S_b$, then we say that $S_b$ is an \emph{ancestor} of $S_a$ and ($S_a$ is a \emph{descendant} of $S_b$). If $S_b$ is an ancestor of $S_a$ such that there is no $C_j$-segment $S_c$ that satisfies $S_a \prec S_c \prec S_b$, then $S_b$ is a parent of $S_a$ (and $S_a$ is a child of $S_b$). Modeling all parent child relationships with edges gives rise to a forest that we call the \emph{$j$-th segment forest} of $\mathcal{Q}$. See \Cref{fig:segmentTreeMain} for an example. 
A $C_j$-segment $S_p$ is an \emph{$m$-th generation} ancestor of $S_q$ if $S_p$ is an ancestor of $S_q$ and the length of the path from $S_p$ to $S_q$ in the $j$-th segment forest is $m$.
The \emph{height} of a segment forest is the length of the longest path from a leaf segment to a segment that has no ancestors.
The \emph{width $m$ $j$-subtree} of a $C_j$-segment $S$ is the tree induced by all $C_j$-segments that are descendants of $S$ at a distance of at most $m$ in the $j$-th segment forest. 

For a segment $S$, its eccentricity $e$ is defined as $e = \min_i V(C_i \cap S)\neq \emptyset$. Let $m$ be any positive integer. We will make in \Cref{lem:eitheror} the observation that in the width $m$ $e$-subtree $T$ rooted at the $k$-th generation ancestor of $S$, either all ancestors of $S$ in $T$ have the same type or there is a segment in $T$  that is not an ancestor of $S$. This observation coupled with \Cref{lem:segmentTypes} will tell us that the width three $e$-subtree rooted at the $3$-rd generation ancestor of a segment $S$ of eccentricity $e$ has a segment that is not an ancestor of $S$. The fact that the height of the $r$-segment forest of $\mathcal{Q}$ is at most logarithmic follows then from an inductive argument. 
In~\cite[Lemma 5]{JCTB}, it is shown that a cheap solution for  \dispaths has  $O(k)$ different types of segments. Using this result, it is easy to check that the same holds true for \tcycle.
Leveraging this fact, we deduce that the height of the $r$-segment forest of $\mathcal{Q}$ is $O(\log k)$. This gives us Lemma~\ref{lem:segmentNumber} from Section~\ref{S:intro}.

\begin{figure}
\centering
\begin{subfigure}{.5\textwidth}
  \centering
  \captionsetup{justification=centering}
  \includegraphics[width=0.9\linewidth]{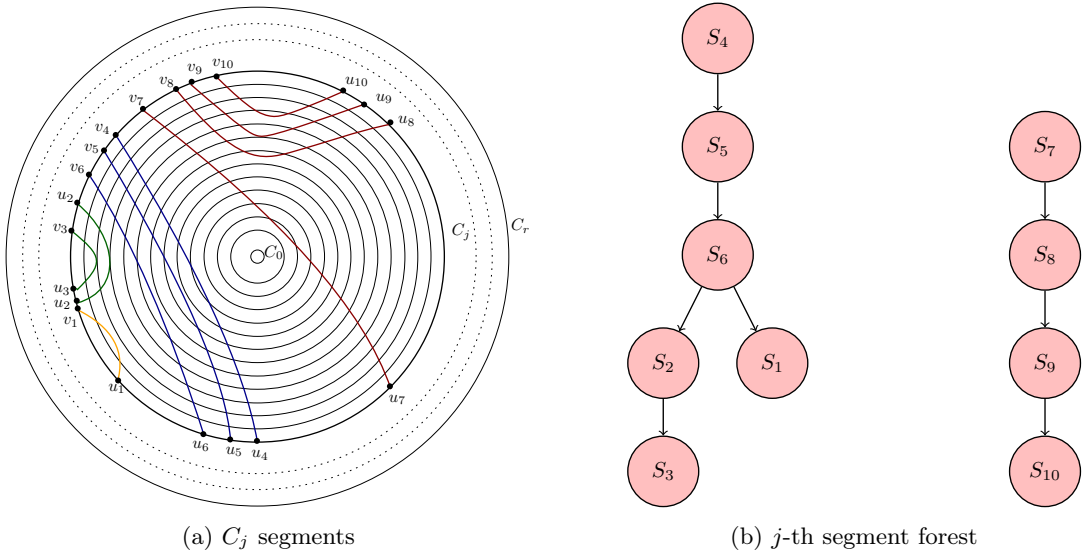}
  \caption{$C_j$ segments}
  \label{fig:sameTypeofSegs}
\end{subfigure}%
\begin{subfigure}{.5\textwidth}
  \centering
  \captionsetup{justification=centering}
\usetikzlibrary{shapes.geometric}
 \resizebox {0.8\textwidth} {!} {
\begin{tikzpicture}
[every node/.style={inner sep=0pt}]
\node (1) [circle, minimum size=32.5pt, fill=pink, line width=0.625pt, draw=black] at (87.5pt, -150.0pt) {\textcolor{black}{$S_4$}};
\node (2) [circle, minimum size=32.5pt, fill=pink, line width=0.625pt, draw=black] at (87.5pt, -200.0pt) {\textcolor{black}{$S_5$}};
\node (3) [circle, minimum size=32.5pt, fill=pink, line width=0.625pt, draw=black] at (87.5pt, -250.0pt) {\textcolor{black}{$S_6$}};
\node (4) [circle, minimum size=32.5pt, fill=pink, line width=0.625pt, draw=black] at (62.5pt, -300.0pt) {\textcolor{black}{$S_2$}};
\node (5) [circle, minimum size=32.5pt, fill=pink, line width=0.625pt, draw=black] at (112.5pt, -300.0pt) {\textcolor{black}{$S_1$}};
\node (7) [circle, minimum size=32.5pt, fill=pink, line width=0.625pt, draw=black] at (237.5pt, -250.0pt) {\textcolor{black}{$S_8$}};
\node (8) [circle, minimum size=32.5pt, fill=pink, line width=0.625pt, draw=black] at (237.5pt, -300.0pt) {\textcolor{black}{$S_9$}};
\node (11) [circle, minimum size=32.5pt, fill=pink, line width=0.625pt, draw=black] at (62.5pt, -350.0pt) {\textcolor{black}{$S_3$}};
\node (6) [circle, minimum size=32.5pt, fill=pink, line width=0.625pt, draw=black] at (237.5pt, -200.0pt) {\textcolor{black}{$S_7$}};
\node (9) [circle, minimum size=32.5pt, fill=pink, line width=0.625pt, draw=black] at (237.5pt, -350.0pt) {\textcolor{black}{$S_{10}$}};
\draw [line width=0.625, ->, color=black] (1) to  (2);
\draw [line width=0.625, ->, color=black] (2) to  (3);
\draw [line width=0.625, ->, color=black] (3) to  (4);
\draw [line width=0.625, ->, color=black] (3) to  (5);
\draw [line width=0.625, ->, color=black] (6) to  (7);
\draw [line width=0.625, ->, color=black] (7) to  (8);
\draw [line width=0.625, ->, color=black] (8) to  (9);
\draw [line width=0.625, ->, color=black] (4) to  (11);
\end{tikzpicture}
}
  \caption{$j$-th segment forest}
  \label{fig:segmentTree}
\end{subfigure}

\caption{Subfigure (a) depicts $C_j$ segments of a CL configuration $\mathcal{Q}$ of depth $r$. Here, segments $S_1,\ldots,S_{10}$ are $C_j$-segments such that segment $S_i$ has endpoints $u_i$ and $v_i$.  All $C_j$-segments of the same type are depicted in a common color. The colors red, blue green and yellow are used to distinguish segment types. Subfigure (b) shows the $j$-th segment forest of $\mathcal{Q}$.}
\label{fig:segmentTreeMain}
\end{figure}

\medskip
\noindent\textbf{Fast Removal of Irrelevant Vertices.} Since we have established that all $g(k)$-isolated vertices are irrelevant, our next task is to remove $g(k)$-isolated vertices in $k^{O(1)}\cdot n$ time.  Note that there is a rather straightforward way to remove such vertices by, as long as possible, computing a $g(k)\times g(k)$ grid minor that does not contain any terminal, and removing its ``middle-most'' vertex. Here, each iteration requires $k^{O(1)}\cdot n$ time, and hence, in total, this algorithm requires $k^{O(1)}\cdot n^2$ time. 

To remove the irrelevant vertices in $k^{O(1)}\cdot n$ time, we use Reed's~\cite{reedLinear} approach of ``cutting'' the plane into more ``manageable'' planes, and then working on these pieces simultaneously. To this end, we consider the ``punctured'' planes, where, intuitively, a $c$-\textit{punctured} plane means a punctured plane with $c$ connected components in its boundary. Note that it is desirable that all terminals lie on the boundary of a punctured plane. This can be achieved, to begin with, by considering the graph to be embedded on a $k$-punctured plane, where each terminal is a trivial hole. Then, the idea of Reed is to recursively cut the plane along some curves of ``bounded'' size into ``nice'' subgraphs, where a subgraph $H$ is \textit{nice} if every vertex of $H$ that is $g(k)$-isolated from the boundary of $H$ is also $g(k)$-isolated in $H$. More specifically, we use the following result of Reed.
\begin{proposition}[Lemma~2 in~\cite{reedLinear}]
    Let $H$ be a nice subgraph of $G$ embedded on a $c$-punctured plane $\boxdot$ such that $c>2$. Then, we can compute both a non-crossing proper closed curve $J$ contained in $
    \boxdot$ and a (possibly empty) set $X$ of vertices that are $g(k)$-isolated from the boundary of $\boxdot$ in $O(|V(H)|)$ time such that 
    \begin{enumerate}
        \item the number of vertices of $H-X$ intersected by $J$ is at most $6g(k)+6$, and
        \item $\boxdot\setminus J$ contains at most 3 components each of which has less than $c$ holes.
    \end{enumerate}
\end{proposition}

A careful analysis shows that by a recursive application of this \textit{cut reduction} procedure, we finally get at most $4k+8$ components, each of which is embedded on either a $2$-punctured plane or a $1$-punctured plane and the total number of vertices on the boundaries of these punctured planes is ${O(k g(k))}$, i.e., $O(k \log k)$. Since we apply this cut reduction  $O(k)$ times, the cutting-plane procedure takes $O(k\cdot n)$ time in total.

The novel component of this part of our algorithm is to decide for each vertex $v$ on the boundary of these punctured planes if $v$ is $g(k)$-isolated in $G$, and remove if it is. We can perform this check in $O(n)$ time for one vertex, and since there are  $O(k \log k)$ vertices altogether on the boundaries of punctured planes, we can remove all $g(k)$-isolated vertices from the boundaries of punctured planes in $O( k\log k\cdot n)$ time. After this step, no vertex on the boundary of a punctured plane is $g(k)$-isolated.  Finally, we use the irrelevant vertex removal algorithm of Cho et al.~\cite{cho2023parameterized} which, given a $c$-punctured plane $\boxdot$ where $c\leq 2$, removes all vertices from $\boxdot$ that are  $4g(k)$-isolated from the boundary of $\boxdot$. Their algorithm in total requires $O(n)$ time for each $c$-punctured plane ($c\in [2]$), and since there are $O(k)$ such planes, this step  requires $O
(k \cdot n)$ time in total.

At this point, we prove that no vertex in $G$ is $5g(k)$-isolated. In particular,we show that if a vertex $v$ in $\boxdot$  is $5g(k)$-isolated, then there exists a vertex on the boundary of $\boxdot$ that is $g(k)$-isolated. An  illustration of the proof is provided in the Figure~\ref{fig:boundaryToIn}.
\begin{figure}
    \centering
    \includegraphics[scale=0.9]{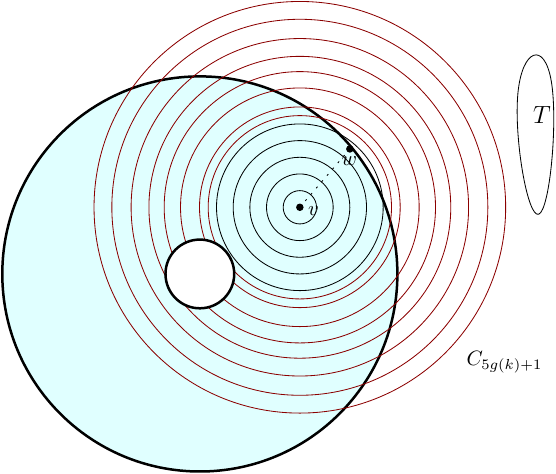}
    \caption{A $2$-punctured plane $\boxdot$ is depicted (in blue) by its two holes highlighted in bold.  $v$ is a vertex in $\boxdot$ which is $5g(k)+1$-isolated in $G$, and hence there is a sequence of $C_0,\ldots,C_{5g(k)+1}$ concentric cycles separating $v$ and $T$. Since $v$ is not $4g(k)$-isolated from boundary of $\boxdot$, there is some boundary vertex $w$ of $\boxdot$ such that $w\in D_{4g(k)}$. In this case, a sequence of concentric cycles $C_{4g(k)+1},\ldots,C_{5g(k)+1}$ separate $w$ from $T$, implying that $w$ is $g(k)$-isolated.}
    \label{fig:boundaryToIn}
\end{figure}

Finally, using a result of Adler et al.~\cite{adler2011tight} (Proposition~\ref{P:jctb}), it follows that if there is no vertex in a planar graph $G$ that is $5g(k)$-isolated from a set of vertices $T$ of size $k$, then the treewidth of $G$ is  $O(\sqrt{k} g(k))$, i.e.,  $O(\sqrt{k} \log k$. This gives us Theorem~\ref{thm:twreduction} from Section~\ref{S:intro}. Indeed, the treewidth of the reduced graph, say $G'$, obtained after removing irrelevant vertices from $G$, is $O(\sqrt{k}\log k)$. Moreover, let  $U$ be the set of all boundary vertices of all $O(k)$ many punctured planes ($|U| \in O(k \log k)$), then observe that in $G'-U$, there is no sequence of $4g(k)+1$ concentric cycles, and hence the treewidth of $G'-U$ is upper bounded by $4g(k)$, i.e., $O(\log k)$.  Now, a dynamic-programming based algorithm can be used to solve our problem in $2^{O(\sqrt{k})\log k}\cdot n$ time. We refer to the procedure described in the above steps as Reed decomposition algorithm.

\subsection{Kernelization Algorithm}

The main tool used in our kernelization algorithm is the so-called \emph{protrusion decomposition}. Please refer to \cite[Chapter 15]{kernelbook} for a detailed discussion on this topic. Below, we provide a brief overview  starting with a few definitions. 

Let $G$ be a graph.
Given a vertex set $U \subseteq V(G)$, its \emph{boundary} $\partial U$, is the set of vertices in $U$ that have at least one neighbor outside of $U$.   A vertex set $W \subseteq V(G)$ is called a \emph{treewidth-$\eta$-modulator} if $\tw(G-W) \leq \eta$. A vertex set $Z\subseteq V(G)$ is a \emph{$q$-protrusion} if $\tw(Z) \leq q$ and $|\partial Z| \leq q$. Finally, given a vertex set $S \subset V(G)$, we use $S^+ = N_G[S]$ to denote the set of vertices that are in the closed neighborhood of $S$ in $G$.

A partition $X_0,X_1,\dots,X_\ell$  of vertex set $V(G)$ of a graph $G$ is called an \emph{$(\alpha,\beta,\gamma)$-protrusion decomposition} of $G$ if $\alpha$, $\beta$ and $\gamma$ are integers such that:
\begin{itemize}
    \item $|X_0| \leq \alpha$,
    \item $\ell \leq \beta$,
    \item $X_i $ for every $i \in [\ell]$ is a $\gamma$-protrusion of $G$,
    \item for $i \in [\ell]$, $N_G(X_i) \subseteq X_0$ (here, $N_G(X_i)$ denotes the open neighborhood of $X_i$ in $G$).
\end{itemize}

For $i \in [\ell]$, the sets $X_i$ are called \emph{protrusions}.
For every $i \in [\ell]$, we set $B_i$ = $X_i^+ \setminus X_i$. 

A central result about protrusions is that if a planar graph $G$ has a treewidth-$\eta$ modulator $S$,  then $G$ has a $(O(|\eta| \cdot |S|),O(|\eta| \cdot |S|),O(|\eta|))$-protrusion decomposition with $X_0 \supseteq S$~\cite[Lemma 15.14]{kernelbook}. Moreover, such a protrusion can be computed efficiently using  Algorithm~\ref{alg:protcomp}. The key motivation behind our use of protrusion decompositions for kernelization is \Cref{thm:twreduction}, which readily gives us a graph $\tilde{G}$ and a vertex set $U \subset V(\tilde{G})$ of size $O(k \log k)$ that is a treewidth-$O(\log k)$-modulator of $\tilde{G}$, where $(G,T)$ and  $(\tilde{G},T)$  are equivalent  as \tcycle instances. 

Apart from protrusions, the other main tool we use is a blackbox from \cite{DBLP:conf/focs/0001Z23}. In \cite{DBLP:conf/focs/0001Z23}, the authors define the notion of $B$-linkage equivalence. Specifically, two graphs $G_1,G_2$ with a common vertex set $B$ are said to be $B$-linkage equivalent w.r.t. \dispaths if for every set of pairs of terminals $M \in B^2$, the disjoint path instances $(G_1,M)$ and $(G_2,M)$ are equivalent. The key result we use from \cite{DBLP:conf/focs/0001Z23} is that  a planar graph $G$ can be replaced by a smaller planar graph $H$ that is $B$-linkage equivalent w.r.t. \dispaths, where the size of $H$ is polynomial in the treewidth of $G$ and $|B|$. This leads us to our first kernelization algorithm outlined in \Cref{fig:kerfigone} that serves as a forerunner to the linear time algorithm summarized in \Cref{fig:kerfigtwo}.

\begin{figure}[H]
\centering
\includegraphics[scale=0.7]{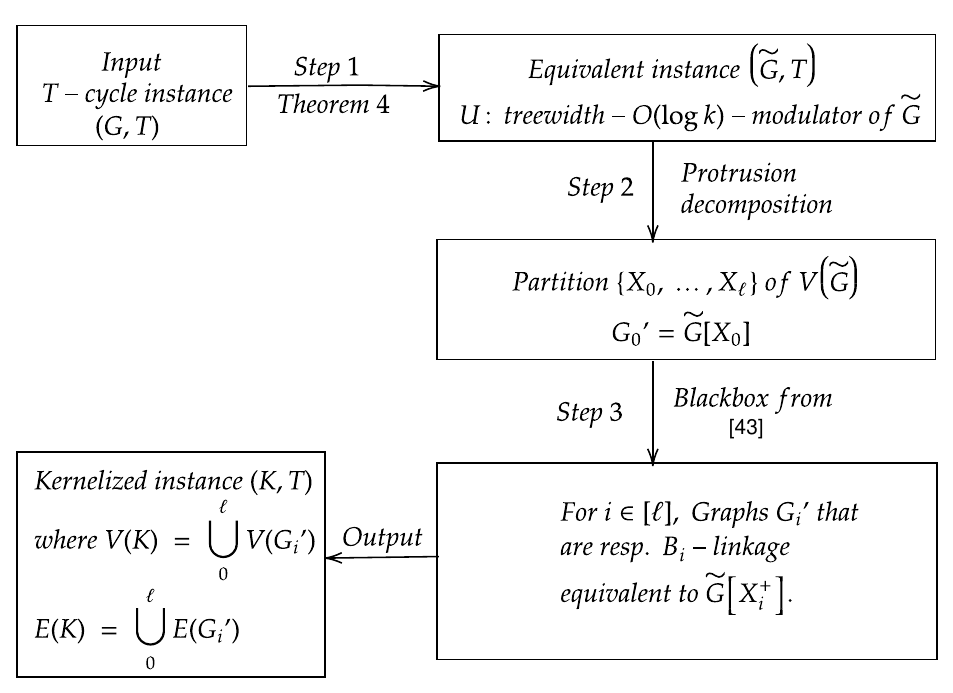}
\caption{Here, we depict the Steps in Kernelization Algorithm-I. The arrows are labelled with the Step$\#$ and the tool used, and every box contains the output data associated to that box.} \label{fig:kerfigone}
\end{figure}

One of the key insights behind the linear time kernelization algorithm is to treat the black box result from   \cite{DBLP:conf/focs/0001Z23} as an existential result rather than an algorithmic one. Thus, instead of using the algorithm from \cite{DBLP:conf/focs/0001Z23} for replacing graphs with smaller $B$-linkage equivalent counterparts, one could guess the replacement graph instead. One obstacle in doing this is that $|B_i|$ where $B_i = X_i^+ \setminus X_i$ and the treewidth of $\tilde{G}[X_i^+]$ for $i\in [\ell]$ are both $O(\log k)$ making the guesses expensive. A potential remedy, which is our other key insight, is to apply a nested protrusion decomposition. However, a roadblock to  this remedy is that we do not have a small treewidth modulator for graphs $\tilde{G}[X_i^+]$ right away. This is fixed by adapting the Reed decomposition algorithm for \textsc{$T$-Cycle}  to  $B_i$-linkage equivalence w.r.t. \tcycle. We stress here that the notion of linkage equivalence w.r.t. \dispaths differs from the notion of linkage equivalence w.r.t. \tcycle. Below, we give an  intuitive definition for linkage equivalence w.r.t. \tcycle. 

 Given a \tcycle instance $(G,T)$, let $G_1$ and $G_1'$ be subgraphs of $G$ with a common set of vertices $B \subset V(G_1) \cap V(G_1')$. Then, $G_1$ is \emph{$B$-linkage equivalent  w.r.t. \tcycle} to $G_1'$ if 
\begin{itemize}
    \item a $T$-loop in $G$ can be written as a union of collections of paths $\mathcal{P}_1\cup \mathcal{P}_2$ with endpoints of paths in  $\mathcal{P}_1\cup \mathcal{P}_2$ lying in $B$ such that paths in $\mathcal{P}_1$ use edges that belong to $G_1$ and paths in $\mathcal{P}_2$ use edges outside of $G_1$ if and only if in the graph $H$ obtained by replacing $G_1$ with $G_1'$ in $G$ there is a $T$-loop that can be written as a  union of collections of paths $\mathcal{P}_1'\cup \mathcal{P}_2$ where paths in $\mathcal{P}_1'$ use edges that belong to $G_1'$.
\end{itemize}

    A vertex $v \in G_1$ is \emph{$B$-linkage irrelevant} if $G_1 - v$ is $B$-linkage equivalent w.r.t. \tcycle to $G_1$. It is easy to check that vertices that are $B$-linkage irrelevant in $G_1$ are also irrelevant for the \tcycle problem in $G$.

Returning to our discussion, an adaptation of the Reed decomposition algorithm allows us to delete $B_i$-linkage irrelevant vertices from $\tilde{G}[X_i^+]$ for every $i\in[\ell]$. This has two important outcomes. First, for every $i\in[\ell]$, we get smaller graphs $G_i'$ that are $B_i$-linkage equivalent to  $\tilde{G}[X_i^+]$. Second, for every $i\in [\ell]$, the Reed decomposition algorithm also computes sets $B_i' \supseteq B_i$ such that the sets $B_i'$ are treewidth-$O(\log \log k)$-modulators of $G_i'$. In fact, the sets $B_i'$ are boundary vertices in the Reed decomposition algorithm.   

\begin{figure}
\centering
\includegraphics[scale=0.7]{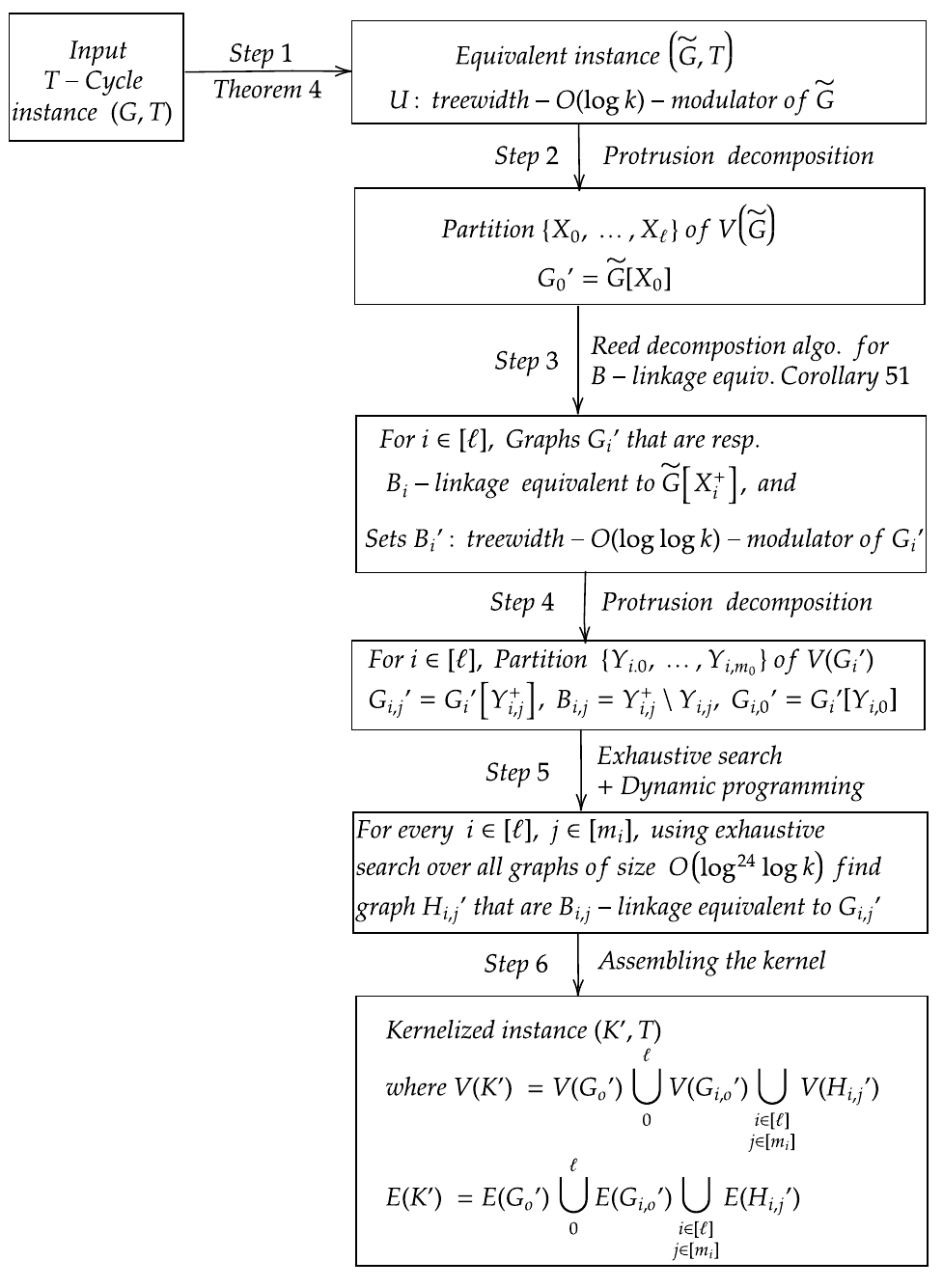}
\caption{Here, we depict the Steps in Kernelization Algorithm-II. The arrows are labelled with the Step$\#$ and the tool used, and every box contains the output data associated to that box.} \label{fig:kerfigtwo}

\end{figure}

With that at hand, we can apply protrusion decomposition to graphs $G_i'$ giving us ``subprotrusions". Specifically,  computing a protrusion decomposition of $G_i'$ for every $i\in[\ell]$, gives a partition $\{Y_{i,0},\dots,Y_{i,m_i}\}$ of $V(G_i')$. For $i \in [\ell]$, $j \in [m_i]$, let $G_{ij}' = G_i'[Y_{i,j}^+] $, where $Y_{i,j}^+ = N_{G_{i}'}[Y_{i,j}]$, and $B_{i,j} = Y_{i,j}^+ \setminus Y_{i,j}$. Also, let  $G_{i,0}' = G_i'[Y_{i,0}]$ for $i\in [\ell]$.  

To now apply \cite{DBLP:conf/focs/0001Z23} as an existential result, we need to guess graphs of size $\log^{O(1)}\log k$. Checking if a guessed graph is $B_{i,j}$ linkage equivalent w.r.t. \dispaths to a subprotrusion $G_{i,j}'$ and replacing it while preserving planarity involves the use of two dynamic programming subroutines: \textsc{Minor Containment} and \dispaths, the details of which can be found in \Cref{sec:ltkernel}. We do these replacements for all $i \in [\ell]$, $j \in [m_i]$. The replacement graphs for $G_{i,j}'$ are denoted by $H_{i,j}'$.

The kernel for the\tcycle problem is given by assembling the graph $K'$ where 
    \[ V(K') = V(G_0') \bigcup_{i\in [\ell]} V(G_{i,0}') \bigcup_{\substack{i\in [\ell] \\ j\in [m_i]}} V(H_{i,j}') \text{ and }  E(K') = E(G_0') \bigcup_{i\in [\ell]} E(G_{i,0}') \bigcup_{\substack{i\in [\ell] \\ j\in [m_i]}} E(H_{i,j}').\]

\Cref{fig:kerfigtwo} outlines the key steps in the linear time kernelization algorithm. We show in \Cref{sec:ltkernel} that $|V(K')|\leq k \log^4 k$, and $K'$ can be computed in $k^{O(1)}\cdot n$ time.

\section{Preliminaries}\label{S:prelim}

For integers $i\leq j$, let $[i,j]=\{i, i+1,\ldots, j-1, j\}$. Moreover, for $\ell\in \mathbb{N}$, let $[\ell] = [1,\ell]$. 
Given a set $A$, the Cartesian product $A\times A$ is denoted by $A^2$.

\subsection{Graph Theory}
For a graph $G$, let $V(G)$ and $E(G)$ denote the vertex set and edge set of $G$, respectively. When $G$ is clear from the context, we denote $|V(G)|$ by $n$. For a vertex set $X \subseteq V(G)$, we use $N_G(X)$ to denote the open neighborhood of $X$ in $G$ and $N_G[X]$ to denote the closed neighborhood of $X$ in $G$. For a graph $G$ and a vertex set $U\subseteq V(G)$, we use $G[U]$ to denote the graph induced by $U$. Moreover, let $G-U = G[V(G)\setminus U]$, and for $x\in V(G)$, let $G-x = G-\{x\}$. 

For two distinct vertices $u,v\in V(G)$, a $(u,v)$-\textit{path} is a path with endpoints $u$ and $v$. Given $S,X,Y\subseteq V(G)$, we say that $X$ \textit{separates} $X$ and $Y$ if $G-S$ has no path with an endpoint in $X$ and an endpoint in $Y$. A \textit{cycle} is a connected subgraph of $G$ where each vertex has degree 2. A cycle $C$ \textit{separates} $X$ and $Y$ if $V(C)$ separates $X$ and $Y$. 
Given a subset $T\subseteq V(G)$, a $T$-\textit{loop} is a simple cycle $C$ such that $V(T)\subseteq V(C)$. We are interested to study the following problem.

% \harmender{for self: Remove the number 1.}

% \begin{problem}[\tcycle]
% {A graph $G$, and a set $T \subseteq V(G)$ of vertices.}
% {$k = |T|$.}
% {Does $G$ have a cycle $C$ such that $V(T)\subseteq V(C)$.
% }
% \end{problem}

\Pb{\tcycle}{A graph $G$, and a set $T \subseteq V(G)$ of vertices.}{Question}{Does $G$ have a cycle $C$ such that $T\subseteq V(C)$?}
A vertex $v$ in $G$ is \textit{irrelevant} if $G-v$ has a $T$-loop whenever $G$ has a $T$-loop. When we say that we delete a path $P$ from a graph $G$, we delete all internal vertices of $P$ from $G$.

A graph $H$ is a \textit{minor} of $G$ if we can obtain $H$ from $G$ by using the following three operations: vertex deletion, edge deletion, and edge contractions. If $H$ is a minor of $G$, we also say that $G$ has a $H$ minor. A $p\times q$-grid is the graph obtained by the Cartesian product of two paths on $p$ and $q$ vertices, respectively. Treewidth, formally defined below, is a measure of how ``tree-like'' a graphs is.

\begin{definition}[{\bf Treewidth}]\label{D:tw}
 A \emph{tree decomposition} of a graph $G$ is a pair $(T,\chi)$, where $T$ is a tree and $\chi$ is a mapping from $V(T)$ (called \emph{bags}) to subsets of $V(G)$, i.e., $\chi: V(T) \rightarrow 2^{V(G)}$, satisfying the following properties. 
\begin{enumerate}
    \item For every $uv \in E(G)$, there exists $t\in V(T)$, such that $\{u,v\} \subseteq \chi(t)$.
    \item For every $v \in V(G)$, the subgraph of $T$ induced by the set $T_v = \{ t \in V(T)~|~v \in \chi(t)\}$ is a non-empty tree.
\end{enumerate}

\noindent The \emph{width} of a tree decomposition $(T,\chi)$ is $\max_{t\in V(T)} |\chi(t)| - 1$. The \emph{treewidth} of $G$ is the minimum possible width of a tree decomposition of $G$. The mapping $\chi$ is extended from vertices of $T$ to subgraphs of $T$. In particular, for a subgraph $U$  of $T$,  $\chi(U) = \bigcup_{v\in V(U)} \chi(v)$.
\end{definition}

\medskip
\noindent\textbf{Planarity.} 
A \textit{planar graph} is a graph that can be embedded in the Euclidean plane. Whenever we consider a planar graph $G$, we consider an embedding of $G$ in the plane as well, and, to simplify notation, we do not distinguish between a vertex of $G$ and the point of the plane used in the drawing to represent the vertex or between an edge and the arc representing it. The \textit{faces}  of a planar graph are the regions bounded by the edges, including the outer infinitely large region. For a planar graph $G$, the \emph{radial distance} between two vertices $u$ and $v$, denoted $\mathsf{d^R}(u,v)$, is one less than the minimum length of a sequence of vertices that starts at $u$ and ends at $v$, such that every two consecutive vertices in the sequence lie on a common face.
%The {\em radial graph}  of a planar graph $G$ is the planar graph $G'$ whose vertex set consists of $V(G)$ and a vertex $v_f$ for each face $f$ of $G$, and whose edge set consists of an edge $uv_f$ for every vertex $u\in V(G)$ and face $f$ of $G$ such that $u$ is incident to (i.e.~lies on the boundary of) $f$. The {\em radial completion} of $G$ is the graph $G'$ obtained by adding the edges of $G$ to the radial graph of $G$. The graph $G'$ is planar, and we draw it on the plane so that its drawing coincides with that of $G$ with respect to $V(G)\cup E(G)$. 
%  This definition extends to subsets of vertices: for $X,Y \subseteq V(G)$, ${\sf rdist}_G(X,Y)$ is the minimum radial distance over all pairs of vertices in $X \times Y$.
% For any $t\in\mathbb{N}$, a sequence ${\cal C}=(C_1,C_2,\ldots,C_t)$ of $t$ cycles in a plane graph $G$ is said to be {\em concentric} if for all $i\in\{1,2,\ldots,t-1\}$, the cycle $C_i$ is drawn in the strict interior of $C_{i+1}$ (excluding the boundary, that is, $V(C_i)\cap V(C_{i+1})=\emptyset$). The {\em length} of ${\cal C}$ is $t$. For a subset of vertices $U\subseteq V(G)$, we say that $\cal C$ is $U$-free if no vertex of $U$ is drawn in the strict interior~of~$C_t$.
As observed in~\cite{JCTB}, combining results from~\cite{gu2012improved} and~\cite{robertson1991graph}, we have the following proposition.

\begin{proposition}[\cite{JCTB}]\label{P:tw}
 Any planar graph with treewidth at least $4.5k$ has a $(k \times k)$-grid minor.   
\end{proposition}

\medskip
\noindent
\textbf{Concentric Cycles.}
Let $G$ be a planar graph. Consider a cycle $C$ of $G$ and let $D$ be the closed disk associated to $C$.  We say that $D$ is \textit{internally chordless} if there is no path in $G$ whose endpoints are vertices  of $C$ and whose edges belong to the open interior of $C$. Let $\mathcal{C} = (C_0,\ldots,C_r)$ be a sequence of cycles in $G$. By $D_i$, we denote the closed interior of $C_i$ (for $i\in [0,r]$). $\mathcal{C}$ is \textit{concentric} if, for $i\in [0,r]$, $C_i$ is contained in the open interior of $D_{i}$. Furthermore, $\mathcal{C}$ is said to be \textit{tight} in $G$, if, additionally,
\begin{enumerate}
    \item $D_0$ is internally chordless.
    \item For $i\in [0,r]$, there is no cycle in $G$ that is contained in $D_{i+1}\setminus D_i$ and whose closed interior $D$  has the property $D_i \subsetneq D \subsetneq D_{i+1}$. 
\end{enumerate}

A vertex $v$ is said to be $\ell$-\textit{isolated} if there exists $\ell+1$ concentric cycles $C_0,\ldots,C_\ell$ such that $v$ is contained in $D_0$ and no vertex of $T$ is contained in $D_{\ell}$. Intuitively, $v$ is separated from each vertex in $T$ by $\ell$ concentric cycles $C_1,\ldots,C_\ell$. See Figure~\ref{fig:OC1} for an illustration. The following result from~\cite{JCTB} shall be useful for us.
\begin{proposition}[\cite{JCTB}]\label{P:jctb}
    There exists an algorithm that, in $2^{(r.\sqrt{|T|})^{O(1)}}\cdot n$ time, given a $n$-vertex planar graph $G$, $T\subseteq V(G)$, and $r\in \mathbb{N}_0$, either outputs a tree decomposition $G$ of width at most $9\cdot(r+1)\cdot \lceil \sqrt{|T|+1}\rceil$ or an internally chordless cycle $C$ of $G$ such that there exists a tight sequence of cycles $C_0,\ldots,C_r$ in $G$ where
    \begin{itemize}
        \item $C_0 = C$, and
        \item all vertices of $T$ are in open exterior of $C_r$.
    \end{itemize}
\end{proposition}

%\todo{A: Maybe write this in an existential form instead of algorithmic?}

\medskip
\noindent\textbf{Segments.}
Let $G$ be a planar graph and let $D$ be the closed interior of a cycle $C$ of $G$. Given a path $P$ in $G$, a subpath $P'$ of $P$ is a {\em $D$-segment} of $P$, if $P'$ is a non-empty (possibly edgeless) path obtained as a connected component of the intersection of $P$ with $D$. Given a $T$-loop $L$, we say that $P'$ is a $D$-segment of $L$ if $P'$ is a $D$-segment of a subpath $P_0$ of $L$. %\harmender{to rewrite segments directly in terms of a cycle.}  

\begin{figure}
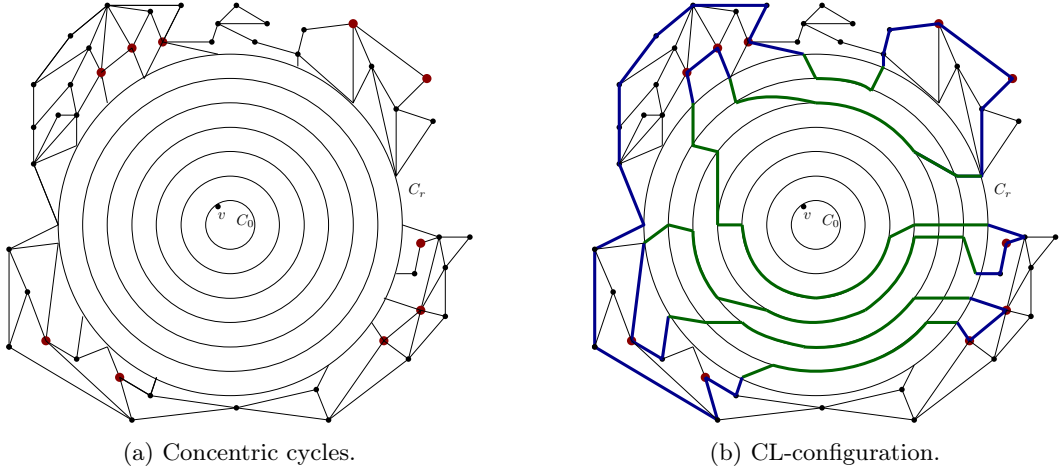

\centering
\begin{subfigure}{.5\textwidth}
  \centering
  \captionsetup{justification=centering}
  \includegraphics[width=.8\linewidth]{OCC1.pdf}
  \caption{Concentric cycles.}
  \label{fig:OC1M}
\end{subfigure}%
\begin{subfigure}{.5\textwidth}
  \centering
  \captionsetup{justification=centering}
  \includegraphics[width=.8\linewidth]{OCC0.pdf}
  \caption{CL-configuration.}
  \label{fig:OC2M}
\end{subfigure}
\caption{(a) A sequence $\mathcal{C}$  of concentric cycles and $v$ is $r$-isolated. (b) A CL-configuration $(\mathcal{C},L)$. The segments of $L$ are represented in  green and the edges of $L$ outside of $D_r$ are shown in blue. Hence, the green edges along with the blue edges combine to give the loop $L$. In both subfigures, the terminal vertices are highlighted in red.} 
\label{fig:testM}
\end{figure}

\medskip
\noindent\textbf{CL-configurations.} A \textit{CL-configuration}, formally defined below, is used to capture the interactions between a sequence of concentric cycles $\mathcal{C}$ and a $T$-loop $L$. See Figure~\ref{fig:OC2M} for an illustration. 

\begin{definition}[{\bf CL-configuration}]\label{D:CL}
Given a planar graph $G$, a pair $\mathcal{Q} = (\mathcal{C},L)$ is a {\em CL-configuration} of $G$ of \textit{depth} $r$ if $\mathcal{C}= (C_0,\ldots,C_r)$ is a sequence of concentric cycles in $G$, $L$ is a $T$-loop in $G$, and $D_r$ does not contain any vertex from $T$. The $D_r$-segments of $L$ are said to be the \textit{segments} of $\mathcal{Q}$. 
\end{definition}

%Given a planar graph $G$, we say that a pair $\mathcal{Q} = (\mathcal{C},L)$ is a CL-configuration of $G$ of \textit{depth} $r$ if $\mathcal{C}= (C_0,\ldots,C_r)$ is a sequence of concentric cycles in $G$, $L$ is a $T$-loop in $G$, and $D_r$ does not contain any vertex from $T$. The $D_r$-segments of $L$ are said to be the \textit{segments} of $\mathcal{Q}$. 
%\abhishek{The \emph{reach} of a segment $P$ is the number of connected components of $P \bigcap (D_r \setminus D_\ell)$, where $\ell = \arg \min_\ell  V(C_\ell\cap P) \neq \emptyset$.\todo{Abhishek: I suggest the definition above for eccentricity as it seems to be closer to what we need.}}
For a segment $P$ of $\mathcal{Q}$, consider the minimum $i$ such that $V(C_i\cap P) \neq \emptyset$. Then, the \textit{eccentricity} of $P$ is $i$. 
%A segment with eccentricity $r$  is said to be \textit{extremal}.  Observe that if $\mathcal{C}$ is tight, then all the extremal segments are subpaths of $C_r$.
Given a cycle $C_i\in \mathcal{C}$ and a segment $P$ of $\mathcal{Q}$,  the $i$-\textit{chords}  of $P$ are the connected components of $P\cap \mathsf{int}(D_i)$. Similarly, for every $i$-chord $X$ of $P$, we define the $i$-\textit{semichords} of $P$ as the connected components of the set $X\setminus D_{i-1}$. Notice that $i$-chords as well as $i$-semichords are open arcs. Given a segment $P$ that does not have any $0$-chord, we define its \textit{zone} as the connected component of $D_r\setminus P$ that does not contain the open-interior of $D_0$. (A zone is an open set.) 
%A CL-configuration $\mathcal{Q} = (\mathcal{C}, L)$ is called \textit{reduced} if the graph $L\cap \cup\mathcal{C}$ is edgeless. \harmender{Observe that we can contract edges of a CL-configuration to make it reduced.}
Following \cite{JCTB}, we define convex segments and convex CL-configurations.
\begin{definition}[{\bf Convex Segments and CL-configurations.}]\label{D:CLConvex}
    A segment $P$ of  ${\cal Q}$ is {\em convex}
if the following three conditions are satisfied:
\begin{itemize}
\item[(i)] $P$ has no $0$-chord. 
\item[(ii)] For every $i\in [r]$,  the following holds:
\begin{itemize}
\item[a.] $P$ has at most one $i$-chord, 
\item[b.] if  $P$ has  an  $i$-chord, then $P\cap C_{i-1}\neq\varnothing$.
\item[c.] Each $i$-chord of $P$ has exactly two $i$-semichords.
\end{itemize}
\item[(iii)] If $P$ has eccentricity $i<r$, there 
is another segment inside the   zone of $P$
%({\tt open closed?})
with eccentricity $i+1$.
\end{itemize}
A CL-configuration ${\cal Q}$ is \emph{convex}  if all its segments are convex.
\end{definition}

Next, we define a notion of ``\textit{cost}'' for a CL-configuration $\mathcal{Q}= (\mathcal{C},L)$, which corresponds to the number of edges in $L$ that do not belong to $\mathcal{C}$.
\begin{definition}[{\bf Cheap Loops}]\label{D:CheapLoop}
Let $G$ be a planar graph and $\mathcal{Q}= (\mathcal{C},L)$ be a CL-configuration of $G$ of depth $r$. For a subgraph $H$ of $G$, let $c(H) = |E(H)\setminus \bigcup_{i\in [0,r]} E(C_i)|$.
$L$ is {\em $\mathcal{C}$-cheap} if there is no other CL-configuration $\mathcal{Q}'=(\mathcal{C},L')$ such that $L'$ is a $T$-loop and $c(L) >c(L')$. 
\end{definition}

\section{Irrelevant Vertex Argument}
Recall that a vertex $v$ of $G$ is irrelevant if $G-v$ has a $T$-loop whenever $G$ has a $T$-loop. In this section, we will establish that, for some $c>1$, if $G$ has a CL-configuration $\mathcal{Q}=(\mathcal{C},L)$ where $\mathcal{C} = (C_0,\ldots,C_r)$ such that $r>c \log k$, then every vertex contained in $D_{r-c\log k}$ is irrelevant, i.e., each $c\log k$-isolated vertex is irrelevant. For the rest of this section, we assume that $G$ is a planar graph, $T\subseteq V(G)$, and $k=|T|$. 
%We first establish in a ``cheap solution'', there can only be a constant number of segments of the same type (in Section~\ref{S:same}), and then using this fact prove our irrelevant vertex argument (in Section~\ref{S:depth}). 

\subsection{Segments of the Same Type}\label{S:same}
Let $\mathcal{Q}=(\mathcal{C},L)$ be a depth $r$ convex CL-configuration of $G$. In this section, we will establish that if $L$ is $\mathcal{C}$-cheap, then the segments of $L$ behave nicely and do not intersect ``deeper'' cycles of $\mathcal{C}$. To this end, we will use the following notion of segment types.

\begin{definition}[{\bf Segment types}]\label{D:ST}
 See Figure~\ref{fig:sameType} for an illustration. Let $\mathcal{Q}=(\mathcal{C},L)$ be a depth $r$ convex CL-configuration of $G$. Moreover, let $S_{1}$ and $S_{2}$ be two $C_j$-segments of ${\cal Q}$ such that $S_i$, for $i\in [2]$, has endpoints $u_i$ and $v_i$. We say that $S_{1}$ and $S_{2}$ are {\em parallel}, and  write $S_1\parallel S_2$ if:
 \begin{enumerate}
     \item there exist paths $P$ and $P'$  on $C_j$
connecting an endpoint of $S_{1}$ with an endpoint of $S_{2}$ 
such that these paths do not pass through the other endpoints of $S_1$ and $S_2$, 
    \item no  segment of ${\cal Q}$ has both
endpoints on $P$ or on $P'$, and
\item the closed-interior of  the cycle $P\cup S_{1}\cup P'\cup S_{2}$ does not contain
the disk $D_{0}$. 
 \end{enumerate}
A \emph{$C_j$-type of segment} is an equivalence class of segments of ${\cal Q}$ under the relation $\parallel$.    
\end{definition}

\begin{figure}
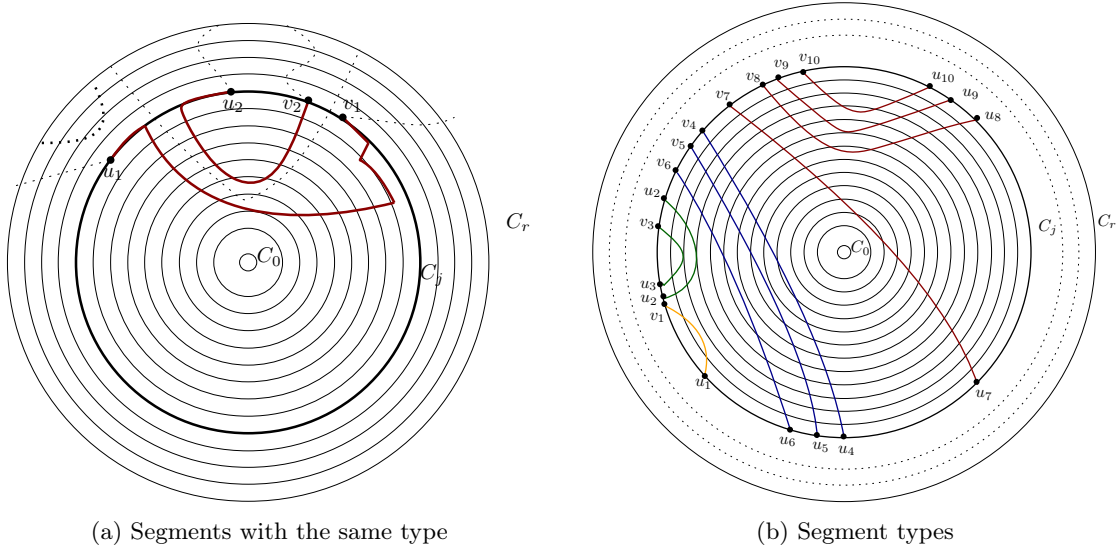

\centering
\begin{subfigure}{.5\textwidth}
  \centering
  \captionsetup{justification=centering}
  \includegraphics[width=0.9\linewidth]{sameType.pdf}
  \caption{Segments with the same type}
  \label{fig:sameType}
\end{subfigure}%
\begin{subfigure}{.5\textwidth}
  \centering
  \captionsetup{justification=centering}
  \includegraphics[width=0.9\linewidth]{sameType2.pdf}
  \caption{Segment types}
  \label{fig:segments}
\end{subfigure}

\caption{(a) Here $S_1$ and $S_2$ are $C_j$-segments with endpoints $u_1,v_1$ and $u_2,v_2$, respectively. Moreover, $P$ and $P'$ are the $(u_1,u_2)$-path and $(v_1,v_2)$-path along $C_j$, respectively. Here $S_1 \parallel S_2$. (b) (i)Here, segments $S_1,\ldots,S_{10}$ are $C_r$-segments such that segment $S_i$ has endpoints $u_i$ and $v_i$.  The segments $S_4$ and $S_7$ are not of the same $C_r$-type because of condition~(3) in Definition~\ref{D:ST} and $S_6$ and $S_2$ are not of the same $C_r$-type because of condition~(2) in Definition~\ref{D:ST}. Further notice that when $S_6$ and $S_2$ are restricted to $D_{r-2}$, they have the same $C_{r-2}$-type. 
    (ii) Here, $S_{10}\prec S_9 \prec S_8 \prec S_7$, and similarly, $S_6\prec S_5\prec S_4$}
\label{fig:segmentFigure}
\end{figure}

% \begin{figure}
%     \centering
%     \includegraphics[scale=0.8]{sameType.pdf}
%     \caption{Here $S_1$ and $S_2$ are $C_j$-segments with endpoints $u_1,v_1$ and $u_2,v_2$, respectively. Moreover, $P$ and $P'$ are the $(u_1,u_2)$-path and $(v_1,v_2)$-path along $C_j$, respectively. Here $S_1 \parallel S_2$.}
%     \label{fig:sameType}
% \end{figure}
 
% \begin{definition}[{\bf Segment types}]\label{D:ST}
%  Let $\mathcal{Q}=(\mathcal{C},L)$ be a depth $r$ convex CL-configuration of $G$. For $j\in [0,r]$, let $S_{1}$, $S_{2}$  be two $C_j$-segments of ${\cal Q}$
% and suppose there exist paths $P$ and $P'$  on $C_j$
% connecting an endpoint of $S_{1}$ with an endpoint of $S_{2}$ 
% such that these paths do not pass through the other endpoints of $S_1$ and $S_2$.
% We say that $S_{1}$ and $S_{2}$ are {\em parallel}, and  write $S_1\parallel S_2$, if (1) no  segment of ${\cal Q}$ has both
% endpoints on $P$ or on $P'$, and  (2) the closed-interior of  the cycle $P\cup S_{1}\cup P'\cup S_{2}$ does not contain
% the disk $D_{0}$.
% A \emph{$C_j$-type of segment} is an equivalence class of segments of ${\cal Q}$ under the relation $\parallel$.    
% \end{definition}

%Sometimes we use abbreviated notation for $C_r$-type segments by simply saying that two $C_r$ segments $P$ and $Q$ have the same \emph{type} instead of saying that they have the same $C_r$-type. 
Let $\mathcal{S}$ be a set of segments of $\mathcal{Q}$ having the same $C_j$-type. We can define a partial order $\prec$ on $\mathcal{S}$ as follows: we say that $S_a \prec S_b $ if and only if $S_a$ lies in the zone of $S_b$. See Figure~\ref{fig:segments} for a reference. 
% We now consider $\mathcal{S}= \{S_1,\ldots,S_p\}$ such that, for $i \in [p-1]$, $S_{i+1} \prec S_i$. 

% \begin{figure}
%     \centering
%     \includegraphics[scale=0.7]{sameType2.pdf}
%     \caption{(i)Here, segments $S_1,\ldots,S_{10}$ are $C_r$-segments such that segment $S_i$ has endpoints $u_i$ and $v_i$.  The segments $S_4$ and $S_7$ are not of the same $C_r$-type because of condition~(3) in Definition~\ref{D:ST} and $S_6$ and $S_2$ are not of the same $C_r$-type because of condition~(2) in Definition~\ref{D:ST}. Further notice that when $S_6$ and $S_2$ are restricted to $D_{r-2}$, they have the same $D_{r-2}$-type. 
%     (ii) Here, $S_{10}\prec S_9 \prec S_8 \prec S_7$, and similarly, $S_6\prec S_5\prec S_4$}
%     \label{fig:segments}
% \end{figure}

In this section, we will establish that a cheap $T$-loop cannot have ``many'' segments of the same type in a convex CL-configuration. Let $P$ (resp., $S$) is a path (resp., segment) of $G$ with endpoints $u$ and $v$. Then, $int(P)$ (resp., $int(S)$ denote the set $V(P)\setminus\{u,v\}$ (resp., $V(S)\setminus \{u,v\}$). We need the following notion of \textit{out-segments} to proceed further.

% \begin{definition}[{\bf Out-segments}]\label{D:OS}
%     Consider a convex CL-configuration $\mathcal{Q}=(\mathcal{C},L)$ of depth $r$ and let $\mathcal{S} = \{S_1,\ldots S_p\}$ be a set of segments of the same $C_j$-type ($j\in [r]$) such that $S_{i+1} \prec S_i$ and $p >3$. Then, the $\mathcal{S}$-out-segments are the connected subgraphs (paths) of $L$ obtained by removing the internal vertices and all edges of paths corresponding to the segments $S_1,S_2,S_3$ from $L$. Observe that since we remove the internal vertices of three subpaths of a cycle, we have three $\mathcal{S}$-out-segments, and let them be denoted by $O_1,O_2,O_3$.  An out-segment $O$ with distinct endpoints $x,y$ is referred as $x,y$-out-segment. See Figure~\ref{fig:basic2} for an illustration. 
% \end{definition}

\begin{definition}[{\bf Out-segments}]\label{D:OS}
    Consider a convex CL-configuration $\mathcal{Q}=(\mathcal{C},L)$ of depth $r$ and let $\mathcal{S} = \{S_1,\ldots S_p\}$ be a set of segments of the same $C_j$-type ($j\in [r]$) such that $S_{i+1} \prec S_i$ and $p >3$.  Then, the $\mathcal{S}$-out-segments are the connected components of $L\setminus(int(S_1)\cup int(S_2) \cup int(S_3))$. Observe that, for $i\in [3]$, $|V(S_i)>2|$ due to the definition of convex segments, and hence $|int(S_i)|>1$. Since we remove three disjoint, connected, and non-empty subparts of a cycle, we have three $\mathcal{S}$-out-segments, and let them be denoted by $O_1,O_2,O_3$. When $\mathcal{S}$ is clear from the context, we refer to a $\mathcal{S}$-out-segments as simply an out-segment. An $\mathcal{S}$-out-segment $O$ with distinct endpoints $x,y$ is referred as $x,y$-out-segment. See Figure~\ref{fig:basic2} for an illustration. 
\end{definition}

% \begin{definition}[{\bf Out-segments}]\label{D:OS}
%     Consider a convex CL-configuration $\mathcal{Q}=(\mathcal{C},L)$ of depth $r$ and let $\mathcal{S} = \{S_1,\ldots S_p\}$ be a set of segments of the same $C_j$-type ($j\in [r]$) such that $S_{i+1} \prec S_i$ and $p >3$. Further, for $i\in [p]$, let $u_i$ and $v_i$ be the endpoints of $S_i$, $Z_i$ be the zone of $S_i$, and let $R$ denote the region $Z_{1}\setminus Z_3$. Since $S_1,S_2,S_3$ are disjoint segments of the $T$-loop $L$, the vertices $u_1,u_2,u_3,v_1,v_2,v_3$ are connected by some paths from outside $R$. Let $O = \{O_1,\ldots, O_q\}$ be the set of connected components of $L\setminus int(R)$\todo{I did not understand the remark of Meirav here}. We say that $O_1,\ldots,O_q$ are the \textit{out-segments}.  Moreover, when an out-segment $O$ connects distinct $x,y\in \{u_1,u_2,u_3,v_1,v_2,v_3\}$, we say that $O$ is a $x,y$-out-segment. See Figure~\ref{fig:basic2} for an illustration. 
% \end{definition}

Now, we present some easy observations concerning segments and out-segments that will shall be helpful for us later. We begin by observing that if some cycle $C$ contains all $\mathcal{S}$-out-segments, then $C$ is a $T$-loop indeed.
\begin{observation}\label{O:three}
    Let $\mathcal{Q}=(\mathcal{C},L)$ be a depth $r$ CL-configuration of $G$ and let $\mathcal{S} = \{S_1,S_2,S_3\ldots,S_p\}$ be a set of the segments of the same $C_j$-type ($j\in [r]$). Further, let   $O_1,\ldots,O_3$ be the $\mathcal{S}$-out-segments. Then, if there exists a cycle $C$ such that $O_1,O_2,O_3$ are connected subgraphs of $C$, then $C$ is a $T$-loop.
\end{observation}
\begin{proof}
    Since $V(L) = V(S_1)\cup V(S_2) \cup V(S_3)\cup V(O_1) \cup V(O_2) \cup V(O_3)$ (by definition of out-segments) and $T\cap (V(S_1)\cup V(S_2) \cup V(S_3)) = \emptyset$ (by definition of CL-configurations), it follows that $T\subseteq V(O_1) \cup V(O_2) \cup V(O_3)$. Hence, if some cycle $C$ contains $O_1,O_2,O_3$ as subgraphs, then $C$ is a $T$-loop.
\end{proof}

Next, we observe a rather straightforward fact.
\begin{observation}\label{O:free}
    Consider a convex CL-configuration $\mathcal{Q} = (\mathcal{C},L)$ of depth $r$ and let $\mathcal{S} = \{S_1,\ldots,S_p\}$ be the segments of the same $C_j$-type such that $p>3$ and $S_p \prec \ldots \prec S_1$. Then, $c(S_2) >0$. 
\end{observation}
\begin{proof}
    Since $S_2$ and $S_3$ are convex segments of the same $C_j$-type and $S_3\prec S_2$, it follows that $S_2$ contains a vertex of the cycle $C_{j-1}$ as well (as otherwise the $\mathcal{C}$ is not tight), and hence $S_2$ contains at least one edge $uv$ such that at most one of $u,v$ is in $C_j$, i.e., $|\{u,v\} \cap V(C_j)| \leq 1$, and hence $c(S_2) \geq c(uv) \geq 1$.
\end{proof}

Now, we prove the main lemma of this section that establishes that in a cheap solution, there cannot be many segments of the same type. 
\begin{lemma}\label{L:main}
   Let $\mathcal{Q}=(\mathcal{C},L)$ be a depth $r$ convex CL-configuration of $G$. Moreover, let $\mathcal{S}$ be a set of segments of the same $C_j$-type, for some $j\in [r]$. If $|\mathcal{S}|>3$, then $L$ is not a $\mathcal{C}$-cheap $T$-loop. 
\end{lemma}

\begin{proof}
    Suppose that $\mathcal{S} = \{S_1,\ldots,S_p\}$ such that for $i \in [p-1]$, $S_{i+1} \prec S_i$. Let $S_i$ intersects $C_j$ at vertices $u_i$ and $v_i$, i.e., the vertices in $\{u_1,\ldots, u_p, v_1,\ldots,v_p\}$ appear on $C_j$ in the order $(u_1,\ldots,u_p,v_p,\ldots,v_1)$ (or the reverse order). We note that possibly $u_p=v_p$, i.e., the segment $S_p$ intersects $C_j$ at a single vertex, say $u$.
    See Figure~\ref{fig:segments} for an illustration.   If $p\leq 3$, then  our claim is vacuously true. So, we assume that $p>3$. Let $O_1,O_2,O_3$ be the $\mathcal{S}$-out-segments. Since $\mathcal{S}$ is clear from the context, we will refer to $\mathcal{S}$-out-segments as simply out-segments to ease the notation. Further, we shall use the  following terminology. For two consecutive vertices $x,y$ in the sequence $(u_1,\ldots,u_p,v_p,\ldots,v_1)$, let $(x,y)$-$C_j$-path denote the $(x,y)$-path along $C_j$ that does not contain any vertex from $\{u_1,\ldots,u_p,v_1,\ldots,v_p\}\setminus \{x,y\}$.We begin with establishing  the following useful claim.
    
    %$i\in [2]$ and $x\in\{u,v\}$, let $(x_i,x_{i+1})$-$C_j$-path denote the $(x_i,x_{i+1})$-path along $C_j$ that does not contain any vertex from $\{u_1,\ldots,u_p,v_1,\ldots,v_p\}\setminus \{x_i,x_{i+1}\}$. 

   \begin{claim}\label{L:T2}
    For $i\in [3]$, no out-segment can have $u_i,v_i$ as its endpoints, no out-segment can have $u_1,v_2$, and symmetrically, $v_1,u_2$ as its endpoints. Further, either $u_1,u_2$  or $v_1,v_2$ are endpoints of an out-segment.
\end{claim}
\begin{proofofclaim}
    First, observe that if some out-segment, say $O$, is a  $(u_i,v_i)$-out-segment, then the segment $S_i$ along with the out-segment $O$ induces a cycle, say $C$, that does not contain any segments other that $S_i$. Since $V(C)\subsetneq V(L)$ and $E(C)\subsetneq E(L)$, and $C$ is a cycle, $L$ cannot be a cycle. Second, note that if there is a $(u_1,v_2)$-out-segment $O$, then there is no way to connect $v_1$ to the cycle due to planarity and the definition of same segment type. The proof for the fact that there is no $(v_1,u_2)$-out-segment is symmetric.
    
     Finally, observe that each vertex from $\{u_1,u_2,u_3,v_1,v_2,v_3\}$ can participate in at most one out-segment since each of these vertices is participating in a segment as well, and the degree of each vertex in a cycle is 2. Hence, there can be at most two out-segments with endpoints from $\{u_3,v_3\}$, and thus, there should be at least one out-segment, say $O$, with both its endpoints from $\{u_1,u_2,v_1,v_2\}$. Since $O$ cannot have $u_1,v_1$ and $u_2,v_2$ as its endpoints, and $O$ cannot have $u_1,v_2$ and $v_1,u_2$ as its endpoints (as proved above), the only two possible endpoint pairs  for $O$ are $u_1,u_2$ and $v_1,v_2$. Hence, $O$ has either  $u_1,u_2$ as endpoints or $v_1,v_2$ as endpoints. This completes our proof.   
\end{proofofclaim}

\begin{figure}
    \centering
    \includegraphics[scale=0.8]{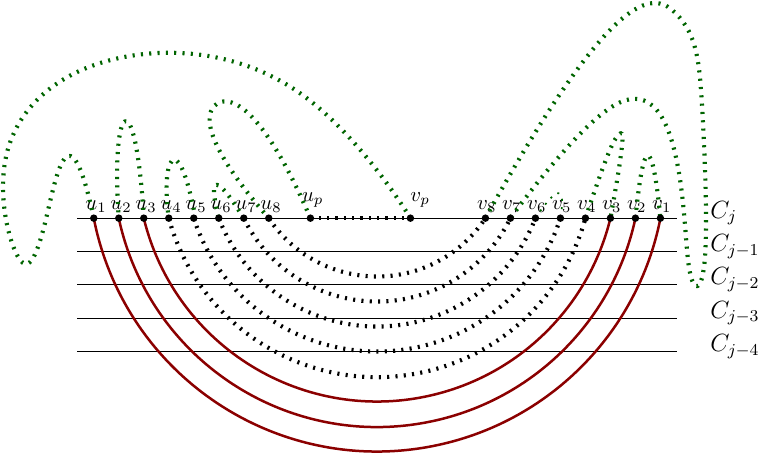}
    \caption{Here $\mathcal{S} = \{S_1,\ldots,S_8\}$ are segments of the same $C_j$-type such that $S_8\prec \cdots \prec S_1$ and $S_i$ has endpoints $u_i$ and $v_i$. The segments $S_1,S_2,$ and $S_3$ are marked in red. Other segments are marked in dotted black lines. All out-segments  are represented as dotted curves and note that all segments other than $S_1,S_2,S_3$ are a subpart of the $u_1,v_3$-out-segment.} 
    \label{fig:basic2}
\end{figure}

Due to Claim~\ref{L:T2}, we know that either $u_1,u_2$  or $v_1,v_2$ are endpoints of the same out-segment, say $O$. Without loss of generality, let us assume that $O$ is a $(v_1,v_2)$-out-segment.  We will construct a $T$-loop $L'$ such that $c(L')<c(L)$. See Figure~\ref{fig:basic3} for an illustration. Let $x$ be a vertex of $S_3$ such that (i) $x$ is on the cycle $C_{j-1}$, (ii) and there exists a $(x,v_3)$-path $P$  along $S_3$ such that each vertex of $P$ has eccentricity at least $j-1$. Similarly, let $y$ be a vertex of $S_1$ such that (i) $y$ is on cycle $C_{j-1}$, (ii) and there exists a $(y,v_1)$-path $P'$  along $S_1$ such that each vertex of $P'$ has eccentricity at least $j-1$. Finally, let $P''$ be the $(x,y)$-path along $C_{j-1}$ such that the interior of $P''$ lies completely in the interior of the zone of the segment $S_1$. We will obtain $L'$ from $L$ by performing the following operations (while performing the operations, the subgraphs may be disconnected). See Figure~\ref{fig:basic4} for an illustration.

\begin{figure}
    \centering
    \includegraphics[scale=0.8]{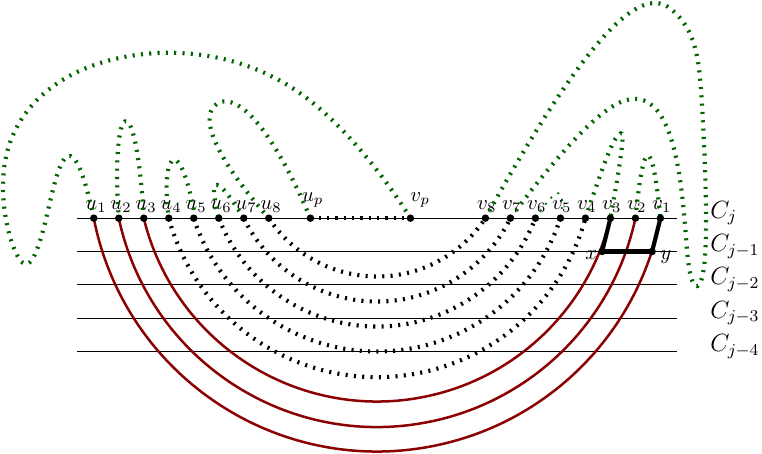}
    \caption{$x,v_3$-path marked in bold black is $P$, $y,v_1$-path in bold is $P'$, and $x,y$-path in bold is $P''$.}
    \label{fig:basic3}
\end{figure}

\begin{figure}[!htb]
    \centering
    \includegraphics[scale=0.9]{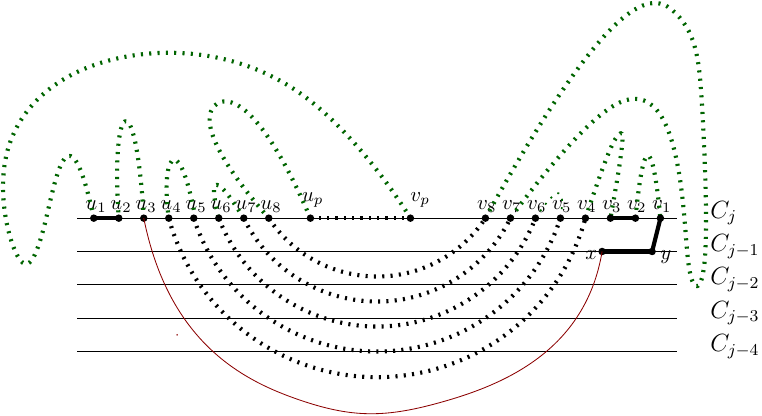}
    \caption{Construction of $L'$ from $L$.}
    \label{fig:basic4}
\end{figure}

    \begin{itemize}
        \item Remove the subpaths corresponding to the segments $S_1$ and $S_2$.
        \item Add the $(u_1,u_2)$-$C_j$-path.
        \item Remove the path $P$ and add the paths $P'$ and $P''$.  
        \item Finally, add the $(v_2,v_3)$-$C_j$-path.
    \end{itemize}

    %To see that $L'$ is indeed a cycle, we prove that: $L'$ is connected, and the degree of each vertex in $L'$ is 2. To see that $L'$ is connected, observe that each vertex in $u_1,v_3$-out-segment and $u_2,u_3$-out-segment are connected by the $u_1,u_2$-$C_r$-path. 

    First, it is not difficult to see that $L'$ is a cycle.
    To see this, first observe that $L'$ has a path from $u_1$ to $u_2$. Also, $L'$ has a path from $v_3$ to $v_2$, from $v_2$ to $v_1$ and from $v_1$
 to $u_3$ (and hence a path from $v_3$ to $u_3$). From \Cref{L:T2}, there is no out-segment from $u_3$ to $v_3$ in $L$. Hence, apart from the $v_1$,$v_2$-out segment, the other two out-segments in $L$ have endpoints either (i) $u_1,u_3$ and (ii) $u_2, v_3$ or (i) $u_1,v_3$ and (ii) $u_2, u_3$, both of which are preserved by $L'$. If follows that every vertex in $L'$ has degree $2$ and $L'$ is a connected graph. 
    
    Next, observe that since we do not remove any out-segment from $L$ to obtain $L'$, $L'$ is a $T$-loop (Observation~\ref{O:free}). Finally, observe that since we remove all the edges of $S_2$ that do not intersect $C_{j-1}$, (in particular, we remove all the edges of $S_2$ that contribute to $c(L)$) (Observation~\ref{O:three}) and only add paths either along the cycle $C_j$ ($(u_1,u_2)$-$C_j$-path and $(v_2,v_3)$-$C_j$-path) or along the cycle $C_{j-1}$ ($(x,y)$-path along $C_{j-1}$), which have cost 0 (by definition), we have that $c(L')< c(L)$. This completes our proof.
\end{proof}

\subsection{Logarithmic depth}\label{S:depth}

The goal of this section is to show that unless there is an irrelevant vertex, the eccentricity of any segment of a CL-configuration is at least $r-O(\log k)$.
To this end, we begin this section by defining ancestor-descendant and parent-child relationship between segments. 

\begin{definition}[\bf{Ancestor and descendant segments}]
Let $\mathcal{Q}=(\mathcal{C},L)$ be a  convex CL-configuration of $G$ of depth $r$, let $j\in[r]$, and let $S$ and $T$ be $C_j$-segments of $\mathcal{Q}$.
    We say that a segment $S$ is an \emph{ancestor} of $T$ (equivalently, $T$ is a \emph{descendant} of $S$) if $T \prec S$, that is, if $T$ is in the zone of $S$.
\end{definition}

Note that if $S$ is the ancestor of $T$, then the eccentricity of $S$ is less than or equal to that of $T$.

\begin{definition}[\bf{Chain of segments}]
Let $\mathcal{Q}=(\mathcal{C},L)$ be a  convex CL-configuration of $G$  of depth $r$, and let $j\in[r]$.  
    We say that there is a \emph{chain of $C_j$-segments of size $i$} from segment $S$ to segment $T$ if there exist segments $S_0 = S, S_1, \dots, S_i = T$ of ${\cal Q}$ such that for any $j, k \in [0,i]$, we have $j<k $ if and only if $S_k$ is in the zone of $S_j$, and for any consecutive segments $S_{k-1}$ and $ S_k$ for $k\in [i]$, there does not exist a segment $S'$ of ${\cal Q}$ such that $S'$ is an ancestor of $S_{k}$ and $S_{k-1}$ is an ancestor of $S'$.
\end{definition}

\begin{definition}[\bf{Generation of ancestor and descendant segments, Parent and child segments}]
Let $\mathcal{Q}=(\mathcal{C},L)$ be a  convex CL-configuration of $G$ of depth $r$, and let $j\in[r]$.
    We say that a $C_j$-segment $S$ is an \emph{($i$-th generation) ancestor} of a $C_j$-segment $T$ if there is a chain of size $i$ from $S$ to $T$. We say that a $C_j$-segment $P$ is a \emph{parent} of a $C_j$-segment $Q$ (equivalently $Q$ is a \emph{child} of a segment $P$) if $P$ is a first generation ancestor of $Q$.
\end{definition}

Next, we define the notions of a segment forest of a CL-configuration, the height of a segment in the  segment forest, and width $k$ subtrees of the segment forest.

\begin{definition}[\bf{Segment forest}]
    Given a CL-configuration $\mathcal{Q}=(\mathcal{C},L)$ of depth $r$, and some $j \in [r]$, the \emph{$j$-th segment forest} is a forest $F = (V,E)$, where $V$ is the set of $C_j$-segments of $\mathcal{Q}$ and the edge set $E$ models parent-child relationsip of $C_j$-segments. 
\end{definition}

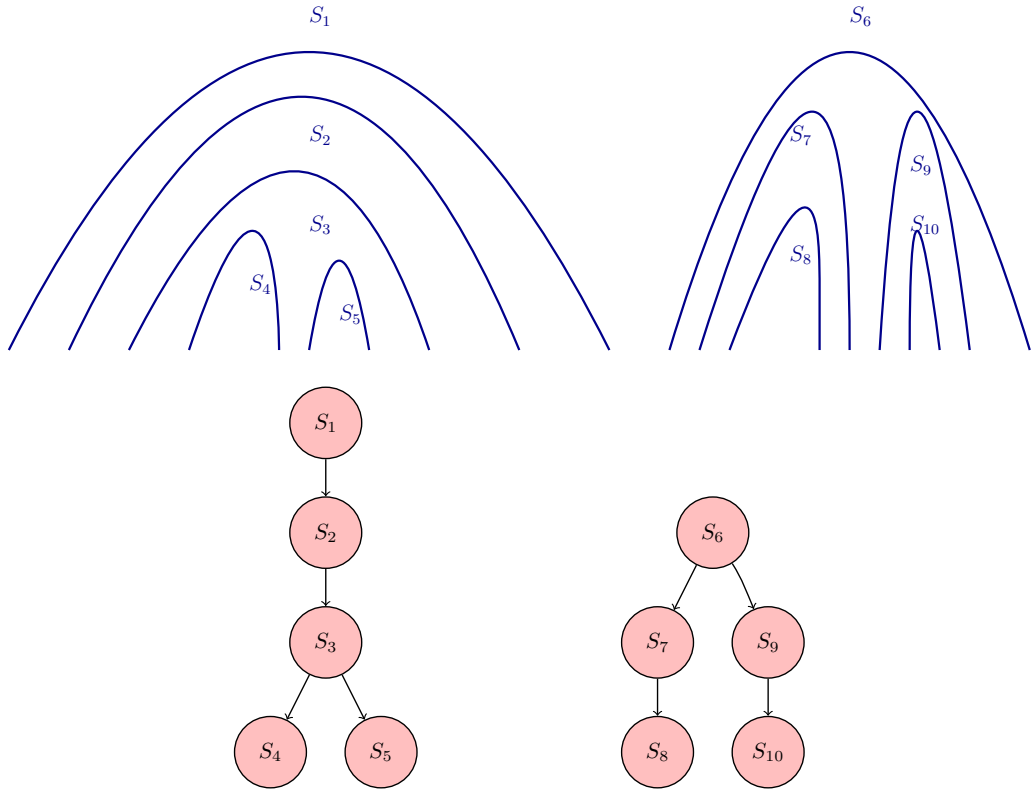
\begin{figure}[htp]
\centering

  \includegraphics[scale=0.7]{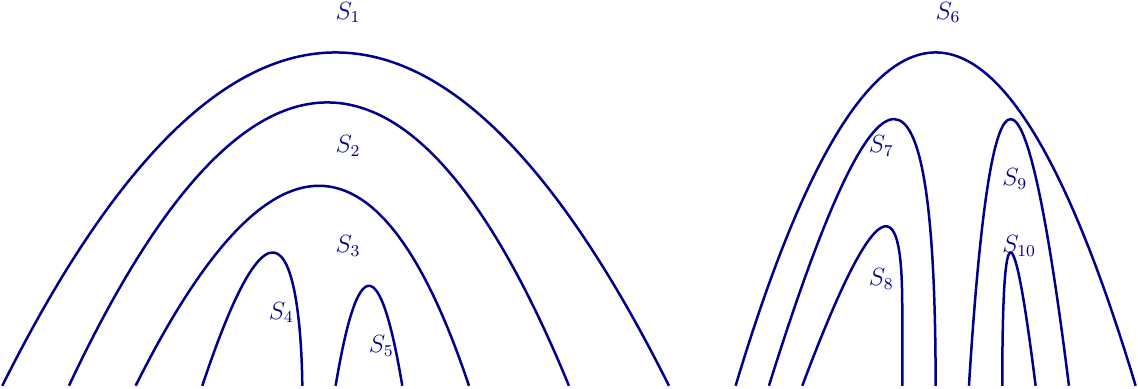}

\bigskip

\usetikzlibrary{shapes.geometric}

\resizebox {0.5\textwidth} {!} {
\begin{tikzpicture}
[every node/.style={inner sep=0pt}]
\node (1) [circle, minimum size=32.5pt, fill=pink, line width=0.625pt, draw=black] at (87.5pt, -150.0pt) {\textcolor{black}{$S_1$}};
\node (2) [circle, minimum size=32.5pt, fill=pink, line width=0.625pt, draw=black] at (87.5pt, -200.0pt) {\textcolor{black}{$S_2$}};
\node (3) [circle, minimum size=32.5pt, fill=pink, line width=0.625pt, draw=black] at (87.5pt, -250.0pt) {\textcolor{black}{$S_3$}};
\node (4) [circle, minimum size=32.5pt, fill=pink, line width=0.625pt, draw=black] at (62.5pt, -300.0pt) {\textcolor{black}{$S_4$}};
\node (5) [circle, minimum size=32.5pt, fill=pink, line width=0.625pt, draw=black] at (112.5pt, -300.0pt) {\textcolor{black}{$S_5$}};
\node (6) [circle, minimum size=32.5pt, fill=pink, line width=0.625pt, draw=black] at (262.5pt, -200.0pt) {\textcolor{black}{$S_6$}};
\node (7) [circle, minimum size=32.5pt, fill=pink, line width=0.625pt, draw=black] at (237.5pt, -250.0pt) {\textcolor{black}{$S_7$}};
\node (8) [circle, minimum size=32.5pt, fill=pink, line width=0.625pt, draw=black] at (237.5pt, -300.0pt) {\textcolor{black}{$S_8$}};
\node (9) [circle, minimum size=32.5pt, fill=pink, line width=0.625pt, draw=black] at (287.5pt, -250.0pt) {\textcolor{black}{$S_9$}};
\node (10) [circle, minimum size=32.5pt, fill=pink, line width=0.625pt, draw=black] at (287.5pt, -300.0pt) {\textcolor{black}{$S_{10}$}};
\draw [line width=0.625, ->, color=black] (1) to  (2);
\draw [line width=0.625, ->, color=black] (2) to  (3);
\draw [line width=0.625, ->, color=black] (3) to  (4);
\draw [line width=0.625, ->, color=black] (3) to  (5);
\draw [line width=0.625, ->, color=black] (6) to  (7);
\draw [line width=0.625, ->, color=black] (6) to  [in=113, out=302] (9);
\draw [line width=0.625, ->, color=black] (7) to  (8);
\draw [line width=0.625, ->, color=black] (9) to  (10);
\end{tikzpicture}
}
    \caption{The top subfigure shows segments of a CL-configuration $\mathcal{Q}=(\mathcal{C},L)$. This gives rise to the segment forest shown in the bottom subfigure. The directed edges depict parent child relationship. The height of leaves $S_4,S_5,S_8$ and $S_{10}$ is $0$. The height of $S_1$ is $3$ and the height of $S_6$ is $2$.}
    \label{fig:segforest}
\end{figure}

\begin{definition}[\bf{Height}]
    Let $\mathcal{Q}=(\mathcal{C},L)$ be a CL-configuration. Let $s$ be the size of the maximum chain in a \emph{segment forest} of $\mathcal{Q}$ from  a leaf node to a segment $S$. 
    Then, the \emph{height of $S$} in that segment forest is equal to $s-1$. The height of a leaf node is zero. The height of a segment forest $F$ is the maximum height of a segment in $F$.
\end{definition}

We refer the reader to \Cref{fig:segforest} for an example of a segment forest.

%Note that the height of a segment plus its eccentricity is equal to $r$. 
Note that a segment has at most one parent, and every segment that is not a leaf in a segment forest has at least one child.

\begin{definition}[\bf{Width $m$ $j$-subtree rooted at a segment}]
    For a $C_j$-segment $P$ of  a  CL-configuration $\mathcal{Q}=(\mathcal{C},L)$, the subtree  of the \emph{segment forest}  of $\mathcal{Q}$  induced by the vertex corresponding to segment $P$ along with all $C_j$-segments that are $i$-th generation descendants of $P$ for all $i \in [m]$ is called the \emph{width $m$ $j$-subtree rooted at $P$}. 
\end{definition}

\begin{figure}
    \centering
\usetikzlibrary{shapes.geometric}
\resizebox {0.5\textwidth} {!} {
\begin{tikzpicture}
[every node/.style={inner sep=0pt}]
\node (1) [circle, minimum size=31.25pt, fill=pink, line width=0.625pt, draw=black] at (87.5pt, -75.0pt) {\textcolor{black}{$S_1$}};
\node (2) [circle, minimum size=31.25pt, fill=pink, line width=0.625pt, draw=black] at (87.5pt, -125.0pt) {\textcolor{black}{$S_2$}};
\node (3) [circle, minimum size=31.25pt, fill=lightgray, line width=0.625pt, draw=black] at (50.0pt, -175.0pt) {\textcolor{black}{$S_3$}};
\node (4) [circle, minimum size=31.25pt, fill=pink, line width=0.625pt, draw=black] at (125.0pt, -175.0pt) {\textcolor{black}{$S_4$}};
\node (5) [circle, minimum size=31.25pt, fill=lightgray, line width=0.625pt, draw=black] at (50.0pt, -237.5pt) {\textcolor{black}{$S_5$}};
\node (6) [circle, minimum size=31.25pt, fill=pink, line width=0.625pt, draw=black] at (125.0pt, -237.5pt) {\textcolor{black}{$S_6$}};
\node (7) [circle, minimum size=31.25pt, fill=lightgray, line width=0.625pt, draw=black] at (25.0pt, -312.5pt) {\textcolor{black}{$S_7$}};
\node (8) [circle, minimum size=31.25pt, fill=lightgray, line width=0.625pt, draw=black] at (75.0pt, -312.5pt) {\textcolor{black}{$S_8$}};
\node (9) [circle, minimum size=31.25pt, fill=pink, line width=0.625pt, draw=black] at (75.0pt, -375.0pt) {\textcolor{black}{$S_9$}};
\node (10) [circle, minimum size=31.25pt, fill=pink, line width=0.625pt, draw=black] at (237.5pt, -162.5pt) {\textcolor{black}{$S_{10}$}};
\node (11) [circle, minimum size=31.25pt, fill=pink, line width=0.625pt, draw=black] at (200.0pt, -237.5pt) {\textcolor{black}{$S_{11}$}};
\node (12) [circle, minimum size=31.25pt, fill=pink, line width=0.625pt, draw=black] at (275.0pt, -237.5pt) {\textcolor{black}{$S_{12}$}};
\node (13) [circle, minimum size=31.25pt, fill=pink, line width=0.625pt, draw=black] at (275.0pt, -325.0pt) {\textcolor{black}{$S_{13}$}};
\draw [line width=0.625, ->, color=black] (1) to  (2);
\draw [line width=0.625, ->, color=black] (2) to  (3);
\draw [line width=0.625, ->, color=black] (2) to  (4);
\draw [line width=0.625, ->, color=black] (3) to  (5);
\draw [line width=0.625, ->, color=black] (4) to  (6);
\draw [line width=0.625, ->, color=black] (5) to  (7);
\draw [line width=0.625, ->, color=black] (5) to  (8);
\draw [line width=0.625, ->, color=black] (8) to  (9);
\draw [line width=0.625, ->, color=black] (10) to  (11);
\draw [line width=0.625, ->, color=black] (10) to  (12);
\draw [line width=0.625, ->, color=black] (12) to  (13);
\end{tikzpicture}
}
    \caption{Here, we show a segment forest. The width $2$ subtree rooted at the second generation ancestor of $S_8$ (which is $S_3$) is shown in grey.}
    \label{fig:depth-m}
\end{figure}
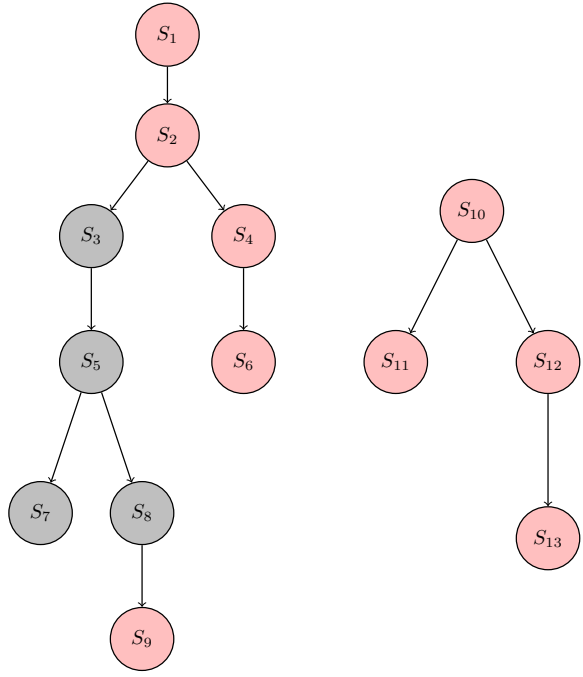

\begin{lemma} \label{lem:eitheror}
Let $\mathcal{Q}=(\mathcal{C},L)$ be a  convex CL-configuration of $G$.
    Let $e$ be the eccentricity of a segment $S$. For any given $m$, suppose that $S$ has a $k$-th generation ancestor that we denote by $U$. Let $T$ be the width $m$ $e$-subtree rooted at $U$. Then, 
    \begin{enumerate}[a.]
        \item  either all ancestors of $S$ in $T$ have the same $C_e$-type
        \item or there exists at least one other segment $R \neq S$ in $T$ that is not an ancestor of $S$.
    \end{enumerate}
\end{lemma}
\begin{proof}
    Let $A$ be an ancestor of $S$ that does not have the same $C_e$-type. Since $A$ is an ancestor, there exist two paths $P$ and $P'$ on $C_e$ connecting an endpoint of $S$ with an endpoint of $A$ that do not pass through the other endpoint of $A$. Since $A$ does not have the same $C_e$ type as $S$ there exists a segment $R$ with both endpoints in $P$ (or both endpoints in $P'$). Such a segment cannot be an ancestor of $S$, proving the claim.
\end{proof}

The two cases of \Cref{lem:eitheror} are depicted in \Cref{fig:either-or}.

\begin{figure}[!htb]
    \centering
    \includegraphics[scale=0.7]{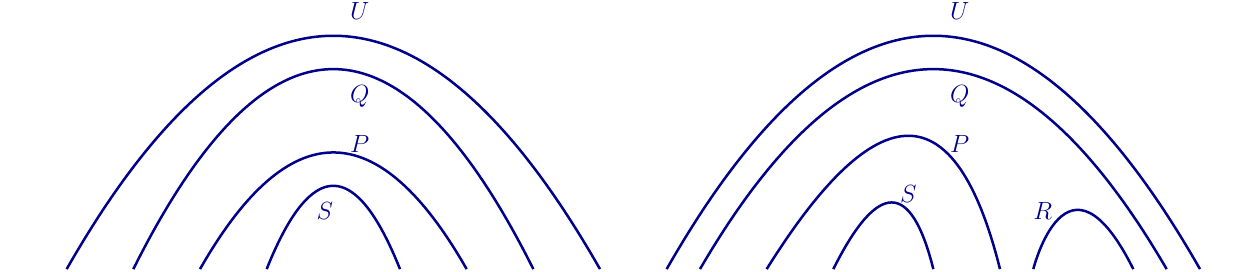}
    \caption{In both subfigures, $S$ is a segment of eccentricity $e$, for some $e$, and $U$ is its $3$-rd generation ancestor. In the left subfigure, all ancestors of $S$ up to $U$ have the same $C_e$-type, whereas in the right subfigure $Q$ does not have the same $C_e$ type, and this is because of the existence of another segment $R$ in the subtree rooted at $U$ that is not an ancestor of $S$ and has eccentricity less than or equal to $e$.}
    \label{fig:either-or}
\end{figure}

\begin{lemma} \label{lem:atleasttwo}
    Consider a CL-configuration $\mathcal{Q}=(\mathcal{C},L)$ of a planar graph $G$. Let $e$ be the eccentricity of a segment $S$. Suppose that $S$ has a $3$-rd  generation ancestor that we denote by $U$. Let $T$ be the width $3$ $e$-subtree rooted at $U$. Then, there exists one other segment  $R \neq S$ that is not an ancestor of $S$ with eccentricity less than or equal to $e$.
\end{lemma}
\begin{proof}
    Towards a proof by contradiction, suppose that there is no such segment $R$ in $T$, then we are in case (a.) of Lemma~\ref{lem:eitheror}. Therefore, $S$ and its three ancestors have the same $C_e$ type. So, there exist four segments of the same $C_e$-type, which contradicts Lemma~\ref{L:main}.  
\end{proof}
    
\begin{theorem} \label{thm:logarithmic}
    Let $h$ be the height of a segment $S$ and let $N$ be the number of segments in the subtree rooted at $S$  in the $r$-th segment forest of a  CL-configuration $\mathcal{Q}=(\mathcal{C},L)$ of depth $r$. Then, $h \leq \log_a N + \hconstant$, where $a = 2^{\frac{1}{\hconstant}}$. 
\end{theorem}
\begin{proof}
        We use induction to prove the claim. For the base case, for every $h \leq \hconstant$, it follows that $h \leq \log_a N + \hconstant$. Now, suppose that the inequality holds true for all segments of height less than $H \geq 4$. We will show that it is also true for segments of height $H$. 
        Let $S$ be a segment of height $H$ and let $N$ be the number of segments in the subtree rooted at $S$. 
        We have $\hconstant = \log_a 2$, which gives 
        \begin{align*}
        H &= \log_a(2) + H - \hconstant \\
          &= \log_a(2 a^{H-6}) + \hconstant 
        \end{align*} 
By induction hypothesis, the subtrees rooted at segments of height $H-\hconstant$ have at least $a^{H-6}$ nodes. Then, by Lemma~\ref{lem:atleasttwo}, there exist at least two segments of height $H-\hconstant$ in the width $\hconstant$ subtree rooted at $S$. Hence, $N \geq 2 a^{H-6}$, proving the claim.
\end{proof}

\begin{corollary}\label{C:irrelevant}
  There exist constants $c_1$ and $c_2$ such that in a  CL-configuration $\mathcal{Q}=(\mathcal{C},L)$ of depth $r$, if the eccentricity of a segment is less than $r - (c_1 \log k + c_2)$, then there exists an irrelevant vertex in $\mathcal{Q}=(\mathcal{C},L)$.
\end{corollary}
\begin{proof}
    From \cite[Lemma 5]{JCTB},  we know that a cheap solution for the \dispaths problem has  $O(k)$ different types of segments, where $k$ is the number of terminal pairs in the \dispaths instance. Note that the above theorem is true irrespective of the choice of the pairing $\{ (s_1,t_1), \ldots,  (s_k,t_k)\}$. In particular, this is true for all pairings which have $t_{i} = s_{i+1}$ for $i\in [k-1]$ and $t_k = s_0$. That is, it is true for all possible pairings that are needed for the \textsc{T-Cycle} problem. Hence, the \textsc{T-Cycle} problem has  $O(k)$ different types of segments.
    
    Since there is a constant number of segments of the same type, \Cref{thm:logarithmic} implies  that the height of the $r$-th segment forest is $O(\log k)$. The claim of the corollary follows.
\end{proof}

As mentioned in Section~\ref{S:intro}, it is easy to get a $2^{O(\sqrt{k}\log k)}\cdot n^2$ time algorithm for \tcycle from here. Indeed,  since all $c\log k$-isolated vertices are irrelevant (Corollary~\ref{C:irrelevant}), we can, as long as possible, compute a $c\log k\times c\log k$ grid minor that does not contain any terminal, and remove its ``middle-most'' vertex. Since each iteration of this procedure requires time $k^{O(1)}\cdot n$, and $O(n)$ iterations are performed in total, this process will take  $k^{O(1)}\cdot n^2$ time in total. Now, since the reduced graph now has treewidth  $O(\sqrt{k}\log k)$ (follows from Proposition~\ref{P:jctb}), a $2^{tw}\cdot n$ time algorithm for \tcycle coupled with above arguments, gives us a $2^{O(\sqrt{k}\log k)}\cdot n^2$ time algorithm for \tcycle. In the following section, we optimize the irrelevant vertices removal and remove ``sufficiently many'' irrelevant vertices in time linear in $n$.

\section{Fast Removal of Irrelevant Vertices} \label{sec:reeddeco}
In this section, we will provide a linear-time algorithm to compute and remove some irrelevant vertices such that the remaining graph has bounded treewidth. Our work builds on techniques developed by Reed~\cite{reedLinear} to solve $k$-\textsc{Realizations} in linear time, which were later used to obtain a linear-time algorithm for \textsc{Planar Disjoint Paths}~\cite{cho2023parameterized}. To this end, we first need the notion of \textit{punctured planes}.

\begin{definition}[{\bf Punctured Plane}]
    A $c${\em-punctured plane} $\boxdot$ is the region obtained by removing $c$ open holes from the plane. The \textit{boundary} of $\boxdot$ is the union of the boundaries of its holes. The vertices on the boundary of $\boxdot$ are called the {\em boundary vertices}.
\end{definition}

Let $(G,T)$ be an instance of \tcycle such that $G$ is a planar graph and $T\subseteq V(G)$ is the set of terminals. Observe that we can get a $k$-punctured plane $\boxdot$ from $G$ by considering each point where a terminal lies as a (trivial) hole.

\medskip
\noindent\textbf{Preliminaries.} A curve on the plane is called \textit{proper} if it intersects the (planar) graph $G$ only at its vertices. For two vertices $u,v \in V(G)$, recall that $\mathsf{d^R}(u,v)$ denotes the radial distance between $u$ and $v$.
Similarly, for $X,Y \subseteq V(G)$, let $\mathsf{d^R}(X,Y) = \min_{x\in X, y\in Y} \mathsf{d^R}(x,y)$, and let $\mathsf{d^R}(x,Y)= \mathsf{d^R}(\{x\},Y)$. 
Recall that there exists some $g(k) \in O(\log k)$ such that each $g(k)$-isolated vertex  is irrelevant (\Cref{C:irrelevant}), which can be explicitly derived from the proof of \Cref{C:irrelevant}. A subgraph $H$ of $G$ embedded on a $c$-punctured plane $\boxdot$ is \textit{nice} if any cycle  separating a vertex $v$ of $H$ from the boundary of $\boxdot$ also separates $v$ from $T$ in $G$.
A vertex $v\in V(H)$ is $\ell$-{\em boundary-isolated in $\boxdot$} if there is a sequence of $\ell$ concentric cycles that separates $v$ from boundary of $\boxdot$.  Observe that if a vertex $v\in V(H)$  is $\ell$-boundary-isolated in $\boxdot$ and $H$ is a nice subgraph, then $v$ is $\ell$-isolated in $G$ (from $T$). 

%\todo{Har: Is the overloading of the definition of ``$\ell$-isolated'' in $G$ and in $\boxdot$ okay or is it confusing?}\todo{A: I think we should use some other  term. I suggest  ``$\ell$-boundary isolated''. Let's use a macro so that we can choose later?}

Reed~\cite{reedLinear} proved the following result, which is an essential component of our algorithm.
\begin{proposition}[Lemma~2 in~\cite{reedLinear}]\label{P:reed}
%\todo{A: Statement of the lemma: Where in Reed?}
    Let $H$ be a nice subgraph of $G$ embedded on a $c$-punctured plane $\boxdot$, and let $Y$ be the set of all $g(k)$-boundary-isolated vertices of $H$ in $\boxdot$. If $c\geq3$, then we can compute both a non-crossing proper closed curve $J$ contained in $\boxdot$ and a (possibly empty) subset $X\subseteq Y$ in $O(|V(H)|)$ time such that
    \begin{enumerate}
        \item the number of vertices of $H-X$ intersected by $J$ is at most $6g(k)+6$, and
        \item $\boxdot \setminus J$ contains at most three connected components each of which has less than $c$ holes.  
    \end{enumerate}
\end{proposition}

% \todo{A: Add a remark.}
% \begin{remark}\labesl{R:Reed}
%     We note that in an application of Proposition~\ref{P:reed},  $X$ is a (possibly empty) subset of all $g(k)$-boundary-isolated vertices of $\boxdot$.
% \end{remark}
%\todo{H: In the above proposition it should be $6g(k)+6$ instead of $6g(k)+5$. I will change this and any subsequent related thing.}
\begin{figure}
    \centering
    \includegraphics[scale=1.2]{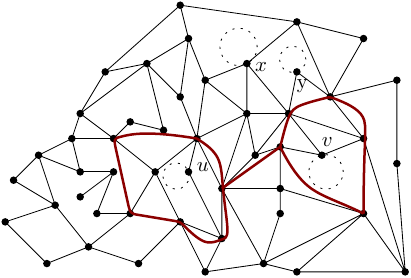}
    \caption{Here, originally $G$ is embedded on a 4-punctured plane where each puncture is illustrated with dotted circle. The cut is illustrated in red, and after the cut reduction, we have three punctured planes: (a.) a $2$-punctured plane with a puncture incident to vertex $u$; (b.) a $2$-punctured plane with a puncture incident to vertex $v$; and (c.) a $3$-puncture plane with punctures incident to vertices $x,y$. For all three punctured planes one of the  punctures is given by the cut.}
    \label{fig:exampleCut}
\end{figure}

In each application of Proposition~\ref{P:reed}, we get at most three components, each of which contains a reduced number of punctures than the input. We call each such application of Proposition~\ref{P:reed} a {\em cut reduction}.  See Figure~\ref{fig:exampleCut} for an illustration of cut reduction. A straightforward corollary of Proposition~\ref{P:reed} is that we can decompose $G$ into $O(k)$  many subgraphs, each embeddable on a $2$-punctured plane with a boundary of size  $O(k \log k)$ in time linear in $n$.
\begin{corollary}\label{C:1imp}
    Let $G$ be a planar graph and $T\subseteq V(G)$ be a set of terminals. Then, in $O(k\cdot n)$ time, after removing some $g(k)$-isolated vertices from $G$, we can decompose $G$ into $O (k)$ nice subgraphs, each embedded on either a $2$-punctured or a $1$-punctured plane, and the total number of boundary vertices in all of these punctured planes, combined, is $O(k \log k)$.
`\end{corollary}
\begin{proof}
     We begin with a $k$-punctured planar graph $G$, and in each application of the \textit{cut reduction} procedure of Proposition~\ref{P:reed} on a $c$-punctured plane (for $c>3$), we obtain at most three instances, each  on a $c-1$-punctured plane. As pointed out by Reed in Section~3.1 in~\cite{reedLinear}, a careful analysis shows that we consider at most $4k+8$ subproblems, and hence we apply cut reduction at most $4k+8$ times. Hence, we decompose $G$ into $4k+8 \in O(k)$ nice subgraphs. Next, observe that in each application of the cut reduction procedure of Proposition~\ref{P:reed}, we obtain a proper curve intersecting $H$ in at most $6g(k)+6 \in O(\log k)$ vertices, and hence, in each cut reduction,  $O(\log k)$  vertices of $G$ are marked as boundary vertices. Since we apply cut reduction only $4k+8$ times, the number of total newly introduced boundary vertices are $O(k \log k)$. Finally, since each application of the cut-reduction from Proposition~\ref{P:reed} takes $O(n)$ time and we apply this procedure  $O(k)$ times, the total running time of our procedure takes $O(k\cdot n)$ time.  Our claim follows.
\end{proof}

\medskip
\noindent \textbf{Deleting Irrelevant Vertices on a 1-punctured plane.}
The next component of our algorithm is a procedure to remove all $g(k)$-boundary-isolated vertices from a graph $H$ embedded on a 1-punctured plane in linear time. We note that this approach was also considered in~\cite{cho2023parameterized} to obtain a linear time algorithm for \textsc{Planar Disjoint Paths}.\footnote{In~\cite{cho2023parameterized} and~\cite{reedLinear}, $g(k) \in 2^{O(k)}$. We note that although we fix $g(k)\in O(\log k)$ to ease the presentation, all of our results also generalize for any $g:\mathbb{N}\rightarrow\mathbb{N}$.}
Let $H$ be a nice subgraph of $G$ embedded on a 1-punctured plane $\boxdot$, and let $V_0\subseteq V(H)$ be the set of vertices that lie on the boundary of $\boxdot$. Further, for $i>0$, let $V_i = \{v~|~ v\in V(H) ~\& ~{\mathsf{d^R}}(v,V_0) = i\}$. Note that we can partition $V(H)$ into  $V_0, \ldots, V_{\ell}$ in $O(|V(H)|)$ time  using standard data structures representing an embedding for planar graphs (\textit{doubly connected edge list})~\cite{CGBook}. The following result from~\cite{cho2023parameterized} establishes that we can remove all vertices of $V_{j}$, where  $j > g(k)$.
\begin{proposition}[~\cite{cho2023parameterized}]\label{P:oneFace}
    Let $H$ be a nice subgraph of $G$ embedded on a 1-punctured plane $\boxdot$, and let $V_0$ be the set of boundary vertices of $H$. Moreover, for $i>0$, let $V_i= \{x~|~ x\in V(H)~\&~ \mathsf{d^R}(x,V_0) =i\}$. Then, for $v\in V_i$:
    \begin{enumerate}
        \item $v$ is $(i-1)$-boundary-isolated in $\boxdot$,
        \item and no sequence of $i$ concentric cycles exists in $G$ that separates $v$ and $V_0$.
    \end{enumerate}
\end{proposition}

Hence, we can use Proposition~\ref{P:oneFace} to remove all $g(k)$-boundary-isolated vertices from a 1-punctured plane $\boxdot$ in linear time by computing a partition of $V(H)$ based on radial distance. 

\medskip
\noindent\textbf{Handling 2-punctured planes.} Since Corollary~\ref{C:1imp} helps us to obtain only 2-punctured planes and Proposition~\ref{P:oneFace} is designed to handle 1-punctured planes, we need  either a way to further apply some cut reduction on 2-punctured planes to get 1-punctured planes or directly get a way to remove irrelevant vertices from 2-punctured planes. To this end, we will use the following lemma, the proof of which was provided to us by Cho, Oh, and Oh~\cite{cho2023parameterized} in a personal communication. We provide a proof for the sake of completeness, but we want to stress that the credit of the proof is to Cho, Oh, and Oh.
\begin{proposition}\label{L:2Cycle}
    Let $H$ be a nice subgraph of $G$ embedded on a $2$-punctured plane $\boxdot$. Then, we can remove all $4g(k)$-boundary-isolated vertices from $\boxdot$ in $O(|V(H)|)$ time.  
\end{proposition}
\begin{proof}
Let $H$ be a nice subgraph of $G$ embedded on a 2-punctured plane $\boxdot$. Let the two boundaries of $\boxdot$ be $H_1$ and $H_2$. We note that we can assume that $H_1$ is contained inside $H_2$ without loss of generality. 

Let $A$ be a shortest simple curve joining a vertex $v_1$ of $H_1$ and a vertex $v_2$ of $H_2$. If $A$ intersects $H$ at most $6g(k)$ vertices, then we can cut along $A$ to get a new 1-punctured plane with  $O(k \log k)$ boundary vertices, and then we can apply Proposition~\ref{P:oneFace} to remove the $g(k)$-isolated vertices. Hence, now we assume that $|A|>6g(k)$.

Now, let us cut along $A$ to get a new punctured plane $\boxdot_A$. Note that the size of boundary of $\boxdot_A$ can be unbounded by any function of $k$. Now, in $\boxdot_A$, curve $A$ corresponds to two curves, say $A_1$ and $A_2$. See Figure~\ref{fig:2face1} for a reference. 
For a vertex $v\in V(H)$, let $\mathsf{d^R_i}(v) = \mathsf{d^R}(v,V(A_i))$, for $i\in [2]$. Let $R$ be the set of vertices in $H$ such that for each vertex $v\in R$, $|\mathsf{d^R_1}(v)-\mathsf{d^R_2}(v)| \leq 1$. It is easy to see that there is a simple curve $B$ in $\boxdot$ connecting $H_1$ and $H_2$ such that $V(B) \subseteq R$.

\begin{figure}
    \centering
    \includegraphics{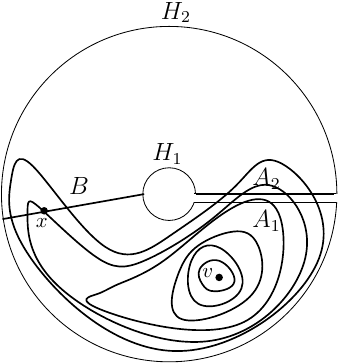}
    \caption{Illustration for the Proof of Proposition~\ref{L:2Cycle}}
    \label{fig:2face1}
\end{figure}
Let $\boxdot_A$ and $\boxdot_B$ be the 1-punctured planes obtained by cutting $\boxdot$ along the simple curves $A$ and $B$, respectively. Now, we apply the following procedure to remove some irrelevant vertices from $H$. 
\begin{enumerate}
    \item Use Proposition~\ref{P:oneFace} on $\boxdot_A$ to remove all $g(k)$-boundary-isolated vertices from $\boxdot_A$ in $H$. 
    \item Use Proposition~\ref{P:oneFace} on $\boxdot_B$ to remove all $g(k)$-boundary-isolated vertices from $\boxdot_B$ in $H$. 
\end{enumerate}

Notice that each vertex removed in the above procedure is a $g(k)$-isolated vertex in $\boxdot$ in $H$ as well (since each vertex on the boundary of $\boxdot$ is also a boundary vertex in $\boxdot_A$ as well as in $\boxdot_B$). But it might so happen that there are still many $g(k)$-boundary-isolated vertices in $\boxdot$ (i.e., the above procedure have not removed all $g(k)$-boundary-isolated vertices from $\boxdot$). We have the following crucial claim which proves that the above procedure removes all $4g(k)$-boundary-isolated vertices from $\boxdot$. 

\begin{claim}
    Let $H'$ be the graph embedded in $\boxdot$ obtained after removing $g(k)$-boundary-isolated vertices from $\boxdot_A$ and then removing $g(k)$-boundary-isolated vertices from $\boxdot_B$. Then, no vertex of $H'$ is $4g(k)$-boundary-isolated in $\boxdot$. 
\end{claim}
\begin{proofofclaim}
    Targeting a contradiction, let $v\in V(H')$ be a vertex that is $4g(k)$-boundary-isolated in $\boxdot$. Then, there exists a vertex $v$ and a sequence of $4g(k)+1$ concentric cycles, say $\mathcal{C} = \{C_0,\ldots C_{4g(k)}\}$, with the following properties:
    \begin{enumerate}
        \item $v \in D_0$ (recall that $D_i = int(C_i)$),
        \item $V(A) \cap D_{g(k)} \neq \emptyset$ and $V(B) \cap D_{g(k)} \neq \emptyset$ (since no vertex in $H'$ is $g(k)$-boundary-isolated in $\boxdot_A$ and in $\boxdot_B$), and 
        \item $(V(H_1) \cup V(H_2)) \cap D_{4g(k)} = \emptyset$.
    \end{enumerate}

    First, we note that each cycle in $\mathcal{C}$ intersects $A$ at most twice, as otherwise, it contradicts the fact that $A$ is a shortest simple curve joining $H_1$ and $H_2$.
    Now, due to (2), the radial distance between $v$ and $A$ is less than $g(k)$, i.e., $\mathsf{d^R}(v,V(A)) < g(k)$. Hence, WLOG, we can assume that $\mathsf{d^R_1}(v) < g(k)$. Therefore, $\mathsf{d^R_2}(v) > 3g(k)$ since to go from $v$ to $A_2$, each radial path needs to cross all the concentric cycles in $\mathcal{C}$ (as no cycle in $\mathcal{C}$ contains $H_1$, no cycle intersects $A$ more than twice). Again, due to (2), we have that the radial distance between $v$ and some vertex, say $x$, of the curve $B$ is at most $g(k)-1$, i.e., $\mathsf{d^R}(v,x) \leq g(k)-1$. Therefore, the radial distance between $x$ and $A_1$ is at most $2g(k)-2$, i.e., $\mathsf{d^R_1}(x) \leq 2g(k)-2$. Thus, since $x$ is a vertex on $B$, $\mathsf{d^R_2}(x) \leq 2g(k)-1$. Hence, $\mathsf{d^R_2}(v)$ is at most the radial distance from $v$ to $x$, and from $x$ to $A_2$, i.e., $\mathsf{d^R_2}(v) \leq 2g(k)-1+g(k)-1 \leq 3g(k)-2$. This contradicts the fact that $\mathsf{d^R_2}(v)>3 g(k)$, completing the proof that no vertex in $H'$ in $\boxdot$ is $4g(k)$-isolated.   
\end{proofofclaim}
Hence, we have a linear-time procedure to remove all $4g(k)$-isolated vertices from $\boxdot$. 
\end{proof}

Due to Corollary~\ref{C:1imp}, we have a decomposition of $G$ into  $O(k)$ subgraphs, each embedded on either a 2-punctured plane or a 1-punctured plane. Moreover, due to Proposition~\ref{P:oneFace} and \Cref{L:2Cycle} none of the $O(k)$ subgraphs contains a vertex inside them that is $4g(k)$-boundary-isolated. Moreover, the total number of vertices on the boundary of these punctured planes is $O (k \log k)$. The final step of our algorithm is to remove some irrelevant vertices from the boundaries of the punctured planes to reach a state where none of the vertices in $G$ is ``too much'' isolated. For this purpose, we begin by proving the following lemma. 
\begin{lemma}\label{L:boundary}
    For $v\in V(G)$, we can decide in $O(n)$ time whether $v$ is $\ell$-isolated (from $T$ in $G$).
\end{lemma}
\begin{proof}
    Start a radial breadth-first search (BFS) from $v$. Observe that $v$ is $\ell$-isolated if and only if for each vertex $x\in T$, $\mathsf{d^R}(v,x) > \ell$ (i.e., the radial distance between $x$ and $v$ is greater than $\ell$). Since radial BFS from a vertex can be performed in $O(n)$ time for a planar graph~\cite{CGBook}, our proof is completed. 
\end{proof}

Finally, we have the following lemma.
\begin{lemma}\label{L:fin}
    Given a planar graph $G$ and $T\subseteq V(G)$, we can remove all $5g(k)+1$-isolated vertices in $O(k\log k \cdot n)$ time. 
\end{lemma}
\begin{proof}
    First, we use Reed's algorithm (Corollary~\ref{C:1imp}) to decompose $G$ into $O(k)$ nice subgraphs in $O(k \cdot n)$ time such that each of these subgraphs is embedded either on a 2-punctured plane or a 1-punctured plane. Moreover, recall from Corollary~\ref{C:1imp} that the total number of vertices of $G$ are on the boundary of these punctured planes is $O(k\log k)$. Then, we apply Lemma~\ref{L:boundary} on each of these vertices to see if they are $g(k)$-isolated (from $T$) in $G$, and remove  if they are $g(k)$-isolated in $O ( k \log k\cdot n)$ time in total. (Recall that if a vertex is $g(k)$-isolated, then it is irrelevant). Note that if all vertices of the boundary of a punctured plane $\boxdot$ that corresponds to a nice subgraph $H$ (of $G$) are $g(k)$-isolated, then, by definition of nice subgraphs, all vertices of $H$ are irrelevant, and we can remove them. Hence, after this step (removing irrelevant vertices from the boundary), we are left with $O(k)$ nice subgraphs, each embedded on either a 1-punctured or a 2-punctured plane such that none of the vertices on the boundary of these planes is $g(k)$-isolated from $T$ in $G$.

    Next, from each 2-punctured plane (resp., 1-punctured plane)  $\boxdot$, we remove all $4g(k)$-isolated vertices in $\boxdot$ using Lemma~\ref{L:2Cycle} (resp., Proposition~\ref{P:oneFace}). Finally, we have the following crucial claim.

\begin{figure}
    \centering
    \includegraphics[scale=1.0]{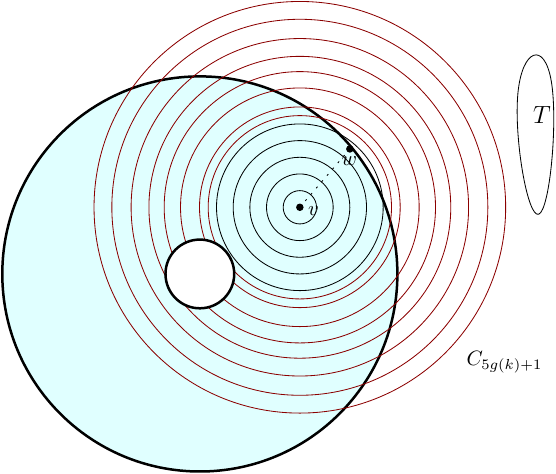}
    \caption{A $2$-punctured plane $\boxdot$ is depicted (in blue) by its two holes illustrated in bold black, and $v$ is a vertex in $\boxdot$ which is $5g(k)+1$-isolated in $G$.  A sequence $C_0,\ldots,C_{5g(k)+1}$ of concentric cycles separating $v$ and $T$ is illustrated such that cycles $C_0,\ldots,C_{4k}$ are illustrated in black and cycles $C_{4g(k)+1},\ldots,C_{5g(k)+1}$ are illustrated in red.}
    \label{fig:C38}
\end{figure}
    
    \begin{claim}\label{C:bd}
        Let $H$ be a nice subgraph of $G$ embedded on a $c$-punctured plane $\boxdot$. Moreover, suppose that none of the vertices on the boundary of $\boxdot$ is $g(k)$-isolated in $G$ and none of the vertices in $H$ is $4g(k)$-boundary-isolated in $\boxdot$. Then, none of the vertices in $H$ is $5g(k)+1$-isolated from $T$ in $G$.  
    \end{claim}
    \begin{proofofclaim}
        Targeting a contradiction, let us assume that there exists a vertex $v\in V(H)$ such that $v$ is $5g(k)+1$-isolated in $G$. See Figure~\ref{fig:C38} for an illustration. Then, there exists a sequence of $5g(k)+2$ concentric cycles, say $C_0,C_1,\ldots, C_{5g(k)+1}$ such that $v\in D_0$ and $T\cap D_{5g(k)+1}= \emptyset$. Further, since $v$ is not $4g(k)$-boundary-isolated in $\boxdot$, there exists a vertex $w$ on the boundary of $\boxdot$ such that $w\in D_{4k}$. But in this case, observe that the concentric cycles $C_{4g(k)+1}, \ldots, C_{5g(k)+1}$ separate $w$ from $T$, contradicting the fact that $w$ is $g(k)$-isolated. Hence, each vertex in $H$ is $5g(k)+1$-isolated in $G$ (from $T$). 
    \end{proofofclaim}

% \todo{The two different notions of isolated seem to have been used in the statement of the claim. This is quite confusing.}   

    Due to Claim~\ref{C:bd}, and the removal of all $g(k)$-isolated vertices from the boundaries of punctured planes, and all $4g(k)$-boundary-isolated vertices from the interior of  $2$-punctured planes and $g(k)$-boundary-isolated (and hence, $4g(k)$-boundary-isolated) vertices from the interior of $1$-punctured planes, it follows that none of the remaining vertices in $G$ is $5g(k)+1$-isolated. 
\end{proof}

Now, we are ready to prove the following result.
\TWReduction*
\begin{proof}
    First, we use Lemma~\ref{L:fin} to remove all $5g(k)+1$-isolated vertices from $G$ to obtain an equivalent instance $(G',T)$ such that $G'$ is a subgraph of $G$ (and $T\subseteq V(G')$) in $O(k\log k \cdot n)$ time. Recall that $g(k) \in O(\log k)$. Since there is no vertex in $G'$ that is $c\log k$-isolated from $T$,  for some constant $c$, $G'$ has treewidth at most $9(c\log k+1)\lceil \sqrt{k +1}\rceil$ (Proposition~\ref{P:jctb}), i.e.,  $tw(G') = O (\sqrt{k} \log k)$. 
    
    Now, let  $U$ be the set of all boundary vertices of all $O(k)$ many punctured planes of $G'$ ($|U| \in O(k \log k)$). Then, observe that  there is no sequence of  $4g(k)+1$ concentric cycles in $G'-U$, as otherwise, there is a vertex in some some punctured plane $\boxdot$ of $G'$ that is  $4g(k)$-isolated in $\boxdot$.  Finally, since there is no sequence of $4g(k)$ concentric cycles in $G'-U$, it follows from Proposition~\ref{P:tw} that the treewidth of $G'-U$ is upper bounded by $O(g(k))$, i.e., $O(\log k)$.
\end{proof}

Finally, we present the main result of this section.
\mainFPT*
\begin{proof}
    First, given an instance of \tcycle $(G,T)$, where $G$ is a planar graph, we use \Cref{thm:twreduction} to get an equivalent instance $(G',T)$ of \tcycle such that $G'$ is planar and has treewidth bounded by $O( \sqrt{k} \log k)$ in $k^{O(1)}\cdot n$ time. We can compute a tree decomposition of width $O (\sqrt{k} \log k)$ in $2^{O (\sqrt{k} \log k)}\cdot n$ time using~\cite{treewidthSingle}.  Since we can solve \tcycle in $2^{O(tw)}\cdot n$ time using a dynamic-programming based algorithm on tree decomposition, for example, using~\cite{DBLP:journals/jcss/DornFT12}, we can solve our problem in $2^{O(\sqrt{k}\log k)}\cdot n$ time.  
\end{proof}

We call the algorithm described in this section \textsc{Reed Decomposition Algorithm}.
Below, we summarize the steps of \textsc{Reed Decomposition Algorithm} in pseudocode form for improved readability.

\noindent
{\bf Algorithm} {\sc Reed decomposition algorithm} ($G,T$)
\begin{itemize}
\item Step 1. Given a planar graph $G$ and a set of terminals $T\subseteq V(G)$ with $|T|=k$, decompose $G$ into $O(k)$ nice subgraphs each embedded into either a $2$-punctured  plane or a $1$-punctured plane using \Cref{C:1imp} so that the total number of boundary vertices is $O(k \log k)$. 
\item Step 2. Using \Cref{L:boundary}, determine which of the boundary vertices of $G$ are  $g(k)$-isolated vertices, and remove those that are. If all the boundary vertices of some $\boxdot$, corresponding to $H$, are $g(k)$-isolated, then remove $V(H)$. 

\item Step 3. Using \Cref{P:oneFace}, Remove all the $g(k)$-boundary-isolated vertices from nice subgraphs of $G$ that are embedded in $1$-punctured planes.
\item Step 4.  Using \Cref{L:2Cycle}, remove all the $4 g(k)$-boundary-isolated vertices from nice subgraphs of $G$ that are embedded in $2$-punctured planes.

\item Step 5. Solve the \tcycle problem using \Cref{thm:main1} 
\end{itemize}

\section{Kernelization} \label{sec:kernel}

In this section and the next section, we will prove  {\sc $T$-Cycle} problem admits a kernelization algorithm of size $k \log^{O(1)}k$ that runs in $k^{O(1)}\cdot n$ time. Specifically, we will prove \Cref{thm:main2}.  As a stepping stone to the linear time kernelization algorithm, in this section, we first design a simpler kernelization algorithm that runs in polynomial (but not linear) time. Both algorithms make use of a standard  tool in parameterized complexity called \emph{protrusion decompositions}. Towards the description of this tool, we start with a few standard definitions from \cite{kernelbook}.

\begin{definition}[\textbf{Boundary of a vertex set}]
Given a graph $G$, we define the \emph{boundary} of a vertex set $U \subseteq V(G)$, denoted by $\partial U$, as the set of vertices in $U$ that have at least one neighbor outside of $U$.    
\end{definition}

\begin{definition}[\textbf{Treewidth-$\eta$-modulator}]
    Given a graph $G$, a vertex set $U \subseteq V(G)$ is called a treewidth-$\eta$-modulator if $\tw(G-U) \leq \eta$.
\end{definition}

\begin{definition}[\textbf{Protrusion, protrusion decomposition}]
For a graph $G$ and an integer $q>0$, we say that a vertex subset $U\subseteq V(G)$ is a \emph{$q$-protrusion} if $\tw(U) \leq q$ and $|\partial U| \leq q$.    

A partition $X_0,X_1,\dots,X_\ell$  of vertex set $V(G)$ of a graph $G$ is called an \emph{$(\alpha,\beta,\gamma)$-protrusion decomposition} of a graph $G$ if $\alpha$, $\beta$ and $\gamma$ are integers such that:

\begin{itemize}
    \item $|X_0| \leq \alpha$,
    \item $\ell \leq \beta$,
    \item $X_i$, for every $i \in [\ell]$, is a $\gamma$-protrusion of $G$,
    \item for every $i \in [\ell]$, $N_G(X_i) \subseteq X_0 $.
\end{itemize}
\end{definition}

For every $i \in [\ell]$, we set $B_i = X_i^+ \setminus X_i$, where $X_i^+ = N_G[X_i]$.

Given a treewidth-$\eta$-modulator $S$ of a planar graph $G$, there exists an $(\alpha,\beta,\gamma)$-protrusion decomposition such that $\alpha,\beta$ and $\gamma$ depend only on $\eta$ and $|S|$, as the following theorem shows.

\begin{theorem}[Lemmas 15.13 and 15.14 in \cite{kernelbook}]
\label{thm:protrusioncomp}
    If a planar graph $G$ has a treewidth-$\eta$-modulator $S$, then
$G$ has a set $X_0 \supseteq S$, such that 
\begin{itemize}
%\item \enskip{} $\left|X_0\right|\leq\enskip4(\eta+1)\left|S\right|+\left|S\right|$, \enskip{} 
\item each connected component of $G-X_0$ has at most $2$ neighbors in $S$
and at most $2\eta$ neighbors in $X_0\setminus S$,
\item $G$ has a $\left(\left(4\left(\eta+1\right)+1\right)\left|S\right|,\left(20(\eta+1)+5\right)\left|S\right|,3\eta+2\right)$-protrusion
decomposition.
\end{itemize}
\end{theorem}

\begin{remark} \label{rem:prottime}
    Suppose that we are given a planar graph $G$ with a treewidth-$\eta$-modulator $S$. Based on the proof of \cite[Lemmas 15.13 and 15.14]{kernelbook}, given a tree decomposition of width $O(\eta)$, an $(O(\eta \cdot |S|),$ $ O(\eta \cdot |S|), O(\eta))$-protrusion decomposition can be computed using the pseudocode in Algorithm~\ref{alg:protcomp} in \Cref{sec:ltkernel}.
    
    Instead of computing the treewidth of $G$ exactly, we use an approximation algorithm such as the one described in \cite{treewidthSingle} for computing the tree decomposition in Step 2 of Algorithm~\ref{alg:protcomp}. Thus, Step 2 of  Algorithm~\ref{alg:protcomp} can be executed in $2^{O(\eta)}\cdot n$ time. The remaining Steps of Algorithm~\ref{alg:protcomp}  can be executed in $O({(\eta \cdot |S|)}^{O(1)}\cdot n)$ time.
    Thus, for our purposes, Algorithm~\ref{alg:protcomp} runs in $O(({(\eta \cdot |S|)}^{O(1)}+ 2^{O(\eta)})\cdot n)$ time.
    \end{remark}

Next, we describe some machinery we use from \cite{DBLP:conf/focs/0001Z23}. Following \cite{DBLP:conf/focs/0001Z23}, we say that two graphs $G_1, G_2$ sharing a set of vertices $B$ are {\em $B$-linkage equivalent w.r.t. \dispaths} if for every set of pairs $\mathcal{M} \subseteq B^2$, the instances $(G_1,\mathcal{M})$ and $(G_2,\mathcal{M})$ of \dispaths are equivalent. We now recall a key kernelization result from \cite{DBLP:conf/focs/0001Z23} that will serve as a blackbox for our algorithm.

\begin{theorem}[Theorem 8 in \cite{DBLP:conf/focs/0001Z23}]
\label{thm:outline:polyKer}
Let $G$ be a planar graph of treewidth $\tw$ and $B \subseteq V(G)$ be of size $c$. 
Then, there exists a polynomial time algorithm that constructs a planar graph $G'$ with $B \subseteq V(G')$ such that $|V(G')| = O(c^{12}\tw^{12})$ and $G'$ is $B$-linkage equivalent w.r.t. \dispaths to $G$.
\end{theorem}

\begin{remark}
    From the proof of Theorem 8 in \cite{DBLP:conf/focs/0001Z23}, it can further be seen that we can construct a graph $G'$ such that $G'$ is a minor of $G$.
\end{remark}

\begin{figure} 
    \centering
\usetikzlibrary{shapes.geometric}
\resizebox {0.6\textwidth} {!} {
\begin{tikzpicture}
[every node/.style={inner sep=0pt}]
\node (2) [circle, minimum size=50.0pt, fill=pink, line width=0.625pt, draw=black] at (75.0pt, -200.0pt) {\textcolor{black}{$\tilde{G}[X_1^+]$}};
\node (3) [circle, minimum size=50.0pt, fill=pink, line width=0.625pt, draw=black] at (150.0pt, -200.0pt) {\textcolor{black}{$\tilde{G}[X_2^+]$}};
\node (4) [circle, minimum size=50.0pt, fill=pink, line width=0.625pt, draw=black] at (250.0pt, -200.0pt) {\textcolor{black}{$\tilde{G}[X_{\ell-1}^+]$}};
\node (5) [circle, minimum size=50.0pt, fill=pink, line width=0.625pt, draw=black] at (325.0pt, -200.0pt) {\textcolor{black}{$\tilde{G}[X_\ell^+]$}};
\node (1) [circle, minimum size=50.0pt, fill=pink, line width=0.625pt, draw=black] at (200.0pt, -100.0pt) {\textcolor{black}{$\tilde{G}[X_0]$}};
\node (6)  at (200.0pt,-200.0pt) {\textcolor{black}{\LARGE $\dots$}};
\draw [line width=0.625, color=black] (1) --   (2);
\draw [line width=0.625, color=black] (1) --  (3);
\draw [line width=0.625, color=black] (1) --  (4);
\draw [line width=0.625, color=black] (1) -- (5);
\node at (120.625pt, -151.875pt) [rotate=38] {\textcolor{black}{ $B_1$}};
\node at (165.625pt, -148.75pt) [rotate=63] {\textcolor{black}{ $B_2$}};
\node at (215.625pt, -150.625pt) [rotate=297] {\textcolor{black}{ $B_{\ell -1 }$}};
\node at (241.875pt, -145.0pt) [rotate=321] {\textcolor{black}{ $B_{\ell}$}};
\end{tikzpicture}
}
\caption{{\sc Kernelization algorithm-I} first finds a graph $\tilde{G}$ such that $(G,T)$ and $(\tilde{G},T)$ are equivalent as T-cycle instances. It then finds a protrusion decomposition of $\tilde{G}$ which is shown in the figure. For each $i \in [\ell]$, $B_i$ is a vertex set that is shared by graphs $G[X_0]$ and $G[X_i^+]$.} \label{fig:kernelone}
\end{figure}
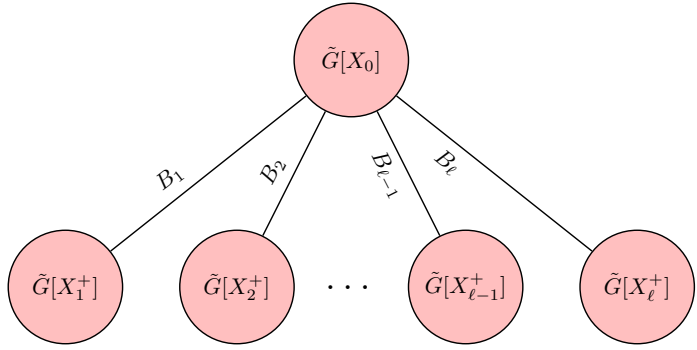

We now provide the pseudocode for a polynomial time kernelization algorithm, which we  use as a springboard for our linear time kernelization algorithm in the next section.

\noindent
{\bf Algorithm} {\sc Kernelization algorithm-I} ($G,T$)
\begin{itemize}
\item  Step 1. From \Cref{thm:twreduction}, we know that there exists a $k^{O(1)}\cdot n$-time algorithm that, given an instance $(G,T)$ of {\sc $T$-Cycle} on planar graphs, outputs an equivalent instance $(\tilde{G},T)$ of {\sc $T$-Cycle} on planar graphs and  a set $U$ such that $|U|\in O(k\log k)$ and the treewidth of $\tilde{G}-U$ is bounded by $O(\log k)$. In other words, $U$ is a treewidth-$O(\log k)$-modulator of graph $\tilde{G}$. Our kernelization algorithm first computes such a graph $\tilde{G}$ and set $U$.

\item Step 2. Next, using Algorithm~\ref{alg:protcomp}, we compute an $(O(k \log^2 k), O(k \log^2 k), O(\log k))$-protrusion decomposition of $\tilde{G}$, which is a partition $\{X_0,\dots,X_\ell\}$ of $V(\tilde{G})$. Let $G_0' = \tilde{G}[X_0]$.

\item Step 3. Finally, using \Cref{thm:outline:polyKer} and \Cref{cor:outline:polyKer}, for every $i \in [\ell]$, we  compute a graph $G_i' $ that is $B_i$-linkage equivalent w.r.t. \dispaths to $\tilde{G}[X_i^+]$, where $X_i^+ = N_{\tilde{G}}[X_i]$. The kernel for the  {\sc $T$-Cycle} instance $(G,T)$ is given by the graph $K$ whose vertex set is given by $V(K) =  \bigcup_{i=0}^\ell V(G_i')$ and whose edge set is given by $E(K) = \bigcup_{i=0}^\ell E(G_i')$.

\end{itemize}  

We refer the reader to \Cref{fig:kernelone} for a schematic depiction of the protrusion decomposition algorithm used in Step 2 of {\sc Kernelization algorithm-I}.

 Towards the goal of showing that {\sc Kernelization algorithm-I} is indeed correct,  we first prove the following corollary.

\begin{corollary} \label{cor:outline:polyKer} 
In {\sc Kernelization algorithm-I}, for every $i \in [\ell]$, we can construct, in polynomial time, a planar graph $G_i'$ with $B_i \subseteq V(G_i')$ such that $|V(G_i')| = O(\log^{24}k)$ and $G_i'$ is $B_i$-linkage equivalent w.r.t. \dispaths to $\tilde{G}[X_i^+]$.
\end{corollary}
\begin{proof}
    Note that $G[X_i^+]$ is a planar graph of treewidth $O(\log k)$ and $B_i \subseteq V(G[X_i^+])$ is also of size $O(\log k)$.
    Applying  \Cref{thm:outline:polyKer}, the claim follows.
\end{proof}

\begin{theorem}
    {\sc Kernelization algorithm-I} correctly computes a kernel for the  {\sc $T$-Cycle} problem on planar graphs in polynomial time.
\end{theorem}
\begin{proof} 
We only need to justify Step 3 of {\sc Kernelization algorithm-I}.
Let $\mathsf{SOL}$ be the set of all solutions to an instance $(\tilde{G},T)$ of the {\sc $T$-Cycle} problem. 
Note that for any $i\in [\ell]$ and for some $t,q\geq 0$, a solution $\mathcal{S}^j \in \mathsf{SOL}$ to the {\sc $T$-Cycle} problem   can be written as a union  of paths in $\mathcal{P}_1 = \{S_1, \dots, S_t\}$ and $\mathcal{P}_2 = \{O_1, \dots, O_q\}$ where the endpoints of the paths in  $\mathcal{P}_1 \bigcup \mathcal{P}_2$ are in $B_i$,  the paths in $\mathcal{P}_1$ use edges that belong to $\tilde{G}[X_i]$ and the paths in $\mathcal{P}_2$ use edges outside of $\tilde{G}[X_i]$. Let $M_i^j = \{ \{u_1,v_1\},\dots,\{u_q,v_q\} \}$ be the respective endpoints of the paths in $\mathcal{P}_1$, and let $\mathcal{M}_i = \{M_i^j \mid \mathcal{S}^j \in  \mathsf{SOL}\}$. Then, $\mathcal{M}_i \subset B_i^2$. By construction, $G_i'$ is $B_i$-linkage equivalent w.r.t. \dispaths to $\tilde{G}[X_i^+]$. So, there is also a way to connect all the pairs in $\mathcal{M}_i$ in $G_i'$.
Thus, after replacing each graph $\tilde{G}[X_i^+]$ by an equivalent subgraph $G_i'$, we obtain a new instance of the  {\sc $T$-Cycle} problem whose solution set is non-empty if and only if $|\mathsf{SOL}| \neq 0$. Using \Cref{cor:outline:polyKer}, the claim follows.
\end{proof}

The size of the vertex set $V(K)$ for the kernel $K$ computed by {\sc Kernelization algorithm-I} is $k \log^{O(1)}k$. This follows from \Cref{cor:outline:polyKer} along with the fact that $\ell$ and $X_0$ determined in Step 2 of {\sc Kernelization algorithm-I} are both bounded by $O(k \log^2 k)$.

\section{Linear time Kernelization Algorithm} \label{sec:ltkernel}

Building on the results from the previous section, in this section, we provide a linear time procedure for providing a $k\log^{O(1)} k$ sized kernel. Before we get to the algorithm, we need to adapt the techniques of \Cref{sec:reeddeco} to  $B$-linkage equivalence for \tcycle. Indeed, this is the content of \Cref{C:twrednsubprot}. With an eye towards this goal, we provide some definitions.

\begin{definition}[$\mathcal{M}$-cycle]
     Given a planar graph $G$,  a vertex set $B\subseteq V(G)$, and a set of pairs $\mathcal{M} \subseteq B^2 $, the graph $G_\mathcal{M}$  obtained from $G$ by inserting for every pair $\{ u_i,v_i\} \in \mathcal{M}$ a pair of edges $\{u_iw_i,  w_iv_i\}$ in $G_{\mathcal{M}}$ is called the \emph{$\mathcal{M}$-subdivided graph of $G$}. We say that there is an  $\mathcal{M}$-cycle in $G$ if $G_\mathcal{M}$ has a $T_\mathcal{M}$-loop, where the set of vertices $T_{\mathcal{M}} =  \{w_i | i\in [|\mathcal{M}|]\}$ are called \emph{subdivision vertices}. % T \bigcup
\end{definition}

\begin{definition}[$B$-linkage equivalence w.r.t. \tcycle]
    Let $G_1$ and $G_2$ be two graphs with a common set of vertices $B$. That is, $B\subseteq V(G_1)$ and $B\subseteq V(G_2)$. Then, we say that $G_1$ and $G_2$ are \emph{$B$-linkage equivalent w.r.t. \tcycle} if for every $\mathcal{M} \subseteq B^2 $, there is an $\mathcal{M}$-cycle in $G_1$ implies that there is also an $\mathcal{M}$-cycle in $G_2$ and vice versa.   
\end{definition}

\begin{definition}[$B$-linkage irrelevant vertices]
 Given a planar graph $G$, and a vertex set $B\subseteq V(G)$, we say that a vertex $v \in V(G) \setminus B$ of a graph $G$ is \emph{B-linkage irrelevant} if $G$ and $G - \{v\}$ are $B$-linkage equivalent w.r.t. \tcycle.    
\end{definition}

We now prove a lemma for $B$-linkage irrelevant vertices that is analogous to \Cref{C:irrelevant}  in \Cref{S:depth}.

\begin{lemma} \label{L:linkageirrelevant}
    Given a planar graph $G$, and a vertex set $B\subseteq V(G)$, let $\ell = |B|$. Then, for some $g(\ell) \in O(\log \ell)$, all $g(\ell)$-isolated vertices are $B$-linkage irrelevant. 
\end{lemma}
\begin{proof}
 Let $c_1$ and $c_2$ be the constants determined by \Cref{C:irrelevant}. For $ g(\ell ) = c_1 \log \ell + c_2$, we want to show that  all   $g(\ell)$-isolated vertices are $B$-linkage irrelevant. Towards a proof by contradiction, assume that this is not the case. Then, there exists an $\mathcal{M} \subseteq B^2$ such that in the $\mathcal{M}$-subdivided graph $G_\mathcal{M}$, there exists a $g(\ell)$-isolated vertex that is not $B$-linkage irrelevant. Applying \Cref{C:irrelevant} directly to the \tcycle instance ($G_\mathcal{M},T_\mathcal{M})$, where $T_\mathcal{M}$ is the set of subdivision vertices, we get a contradiction.
\end{proof}

\begin{remark}
    One way to interpret \Cref{L:linkageirrelevant} is that while the set of irrelevant vertices for  T-cycle instances ($G_{\mathcal{M}_1},T_{\mathcal{M}_1})$  ($G_{\mathcal{M}_2},T_{\mathcal{M}_2})$ for distinct $\mathcal{M}_1,\mathcal{M}_2 \subseteq B^2$  may not be the same, all the  $g(\ell)$-isolated vertices are  $B$-linkage irrelevant.
\end{remark}

Note that \Cref{P:reed,P:oneFace,L:2Cycle} and \Cref{C:1imp} are not specific to the \tcycle problem as such. In particular,  given a planar graph $G$,  and a set of boundary vertices $B\subseteq V(G)$, we can now apply the first three steps of the {\sc Reed Decomposition Algorithm} from \Cref{sec:reeddeco}.

\noindent
{\bf Algorithm} {\sc Reed Decomposition Algorithm for $B$-linkage equivalence} ($G,B$)
\begin{itemize}
\item Step 1. Given a planar graph $G$ and a set of boundary vertices $B\subseteq V(G)$ with $|B|=m $, decompose $G$ into $O(m)$ nice subgraphs each embedded into either a $2$-punctured  plane or a $1$-punctured plane using \Cref{C:1imp} so that the total number of  vertices in the (new) boundary $B' \supset B$ is $O(m \log m)$. 
\item Step 2. Using \Cref{P:oneFace}, remove all the $g(m)$-boundary-isolated vertices from nice subgraphs of G that are embedded in $1$-punctured planes.
\item Step 3. Using \Cref{L:2Cycle}, remove all the $4 g(m)$-boundary-isolated vertices from nice subgraphs of $G$ that are embedded in $2$-punctured planes.
\end{itemize}

As a result of applying {\sc Reed decomposition algorithm for $B$-linkage equivalence}, we have the following corollary.

\begin{corollary} \label{C:twrednsubprot}
    Given a planar graph $G$ and a set of
    vertices $B\subseteq V(G)$ with $|B|=m $, one can compute a set of  vertices $B' \supset B$ such that $|B'| \in O(m \log m)$ and remove all vertices that are (at least) $4 g(m)$-boundary-isolated  in linear time to obtain a graph $\tilde{G}$ that is $B$-linkage equivalent w.r.t. \tcycle to $G$. Moreover, $B'$ is a treewidth-$O(g(m))$-modulator of $\tilde{G}$. 
\end{corollary}

Next, we provide a detailed description of {\sc Kernelization algorithm-II}.

\noindent
{\bf Algorithm} {\sc Kernelization algorithm-II} ($G,T$)
\begin{itemize}
\item  Step 1. As in the case of {\sc Kernelization algorithm-I}, use the $k^{O(1)}\cdot n$-time algorithm from \Cref{thm:twreduction}, that given an instance $(G,T)$ of {\sc $T$-Cycle} on planar graphs, outputs an equivalent instance $(\tilde{G},T)$  and  a set $U$ such that $|U|\leq O(k\log k)$ and the treewidth of $\tilde{G}-U$ is bounded by $O(\log k)$. That is, the kernelization algorithm first computes such a graph $\tilde{G}$ and a treewidth-$O(\log k)$-modulator $U$ of graph $\tilde{G}$. 

\item Step 2. Using Algorithm~\ref{alg:protcomp}, we compute an $(O(k \log^2 k), O(k \log^2 k), O(\log k))$-protrusion decomposition of $\tilde{G}$, which is a partition $\{X_0,\dots,X_\ell\}$ of $V(\tilde{G})$. Let $G_0' = \tilde{G}[X_0]$. 
%For $i \in \{0,\dots,\ell\}$, let $G_i'$ be the subgraph in $\tilde{G}$ induced  by vertices in $X_i$.

\item Step 3.  For every $i\in [\ell]$, use the linear-time algorithm from \Cref{C:twrednsubprot} to obtain a graph $G_i'$  that is  $B_i$-linkage equivalent w.r.t. \tcycle to $\tilde{G}[X_i^+]$  and  a set $B_i' \subseteq V(G_i')$ such that $|B_i'|$ is  $ O(\log k\log \log k)$ and the treewidth of $G_i'-B_i'$ is bounded by $O(\log \log k)$. That is, the kernelization algorithm first computes such a graph $G_i'$ and a treewidth-$O(\log \log k)$-modulator $B_i'$ of graph $G_i'$. 

\item Step 4. For every $i \in [\ell]$, using Algorithm~\ref{alg:protcomp},  compute an $(O(\log k \, \log^2 \log k), O(\log k \, \log^2 \log k),$ $ O(\log \log k))$-protrusion decomposition of $G_i'$, which gives a partition $\{Y_{i,0},\dots,Y_{i,m_i}\}$ of $V(G_i')$. For $i \in [\ell]$, $j \in [m_i]$,  let $G_{ij}'$ be the subgraph of $G_i'$   induced  by vertices in $Y_{i,j}^+ = N_{G_{i}'}[Y_{i,j}]$. That is, let $G_{ij}' = G_i'[Y_{i,j}^+] $. Also, let  $G_{i,0}' = G_i'[Y_{i,0}]$ for $i \in [\ell]$.
Finally, let $B_{i,j} = Y_{i,j}^+ \setminus Y_{i,j}$, for $i \in [\ell]$, $j \in [m_i]$.

\item Step 5. For every $i \in [\ell], j\in [m_i]$, do the following:
\begin{itemize}
    \item Step 5.1 Generate all possible distinct graphs $H$ with $|V(H)| = O(\log^{24} \log k)$ and $B_{i,j} \subseteq V(H)$ in a set $\mathcal{G}_{i,j}$.
    \item Step 5.2 For each graph $H \in \mathcal{G}_{i,j}$,
    \begin{itemize}
        \item Use the dynamic programming in \cite{adlerone,adlertwo} to check if $H$ is a minor of $G_{i,j}'$.
        \item For every pairing of terminals $\mathcal{M}\subseteq B_{i,j}^2$, check if $(H,\mathcal{M})$ and $G_{i,j}',\mathcal{M})$ are equivalent as disjoint path instances, that is, check if the yes / no answers to these instances are the same using the algorithm in \cite{cho2023parameterized}. 
    \end{itemize}

\end{itemize}

    \item Step 6 For every $i \in [\ell], j\in [m_i]$, let $H_{i,j}' \in \mathcal{G}_{i,j}$ denote a graph that is found to be $B_{i,j}$-linkage equivalent w.r.t. \dispaths to $G_{i,j}'$ in Step 5.2.
    The kernel for the \tcycle problem is given by the graph $K'$ where 
    \[ V(K') = V(G_0') \bigcup_{i\in [\ell]} V(G_{i,0}') \bigcup_{\substack{i\in [\ell] \\ j\in [m_i]}} V(H_{i,j}') \text{ and }  E(K') = E(G_0') \bigcup_{i\in [\ell]} E(G_{i,0}') \bigcup_{\substack{i\in [\ell] \\ j\in [m_i]}} E(H_{i,j}'). \]
\end{itemize}  

Please see \Cref{fig:kerneltwo} for a schematic depiction of the nested protrusion decomposition algorithm used in {\sc Kernelization algorithm-II}.  

\begin{remark}
    Note that we use two different notions of $B$-linkage equivalence in Step 3 and Step 6 of {\sc Kernelization algorithm-II} respectively. First, it is each to check that given two  subgraphs $G_1$ and $G_2$ of $G$ with $B \in V(G_1) \cap V(G_2)$, if they  are $B$-linkage equivalent w.r.t. \dispaths, then they are also $B$-linkage equivalent w.r.t. \tcycle. For our purposes, checking only $B$-linkage equivalent w.r.t. \tcycle would suffice in all cases. In Step 6, we check linkage equivalent w.r.t. \dispaths, simply because of the simplicity of formulation.

    Likewise, in Step 3 of {\sc Kernelization algorithm-I}, we replace a graph $\tilde{G}[X_i^+]$ with a $B_i$-linkage equivalent w.r.t. \dispaths graph $G_i'$ because we are using a blackbox. 
\end{remark}

\begin{figure}[!htb] 
    \centering
  \usetikzlibrary{shapes.geometric}
  \resizebox {0.8\textwidth} {!} {
\begin{tikzpicture}
[every node/.style={inner sep=0pt}]
\node (1) [circle, minimum size=60.0pt, fill=pink, line width=0.625pt, draw=black] at (200.0pt, -100.0pt) {\textcolor{black}{$\tilde{G}[X_0]$}};
\node (3) [circle, minimum size=60.0pt, fill=pink, line width=0.625pt, draw=black] at (137.5pt, -200.0pt) {\textcolor{black}{$\tilde{G}[X_2^+] \to G_{2}'$}};
\node (15) [circle, minimum size=40pt, fill=lightgray, line width=0.625pt, draw=black] at (137.5pt, -237.5pt) {\textcolor{black}{$G_2'[Y_{2,0}]$}};
\node (2) [circle, minimum size=60.0pt, fill=pink, line width=0.625pt, draw=black] at (37.5pt, -200.0pt) {\textcolor{black}{$\tilde{G}[X_1^+]$}};
\node (5) [circle, minimum size=60.0pt, fill=pink, line width=0.625pt, draw=black] at (375.0pt, -200.0pt) {\textcolor{black}{$\tilde{G}[X_{\ell}^+]$}};
\node (4) [circle, minimum size=60.0pt, fill=pink, line width=0.625pt, draw=black] at (275.0pt, -200.0pt) {\textcolor{black}{$\tilde{G}[X_{\ell-1}^+]$}};
\node (8) [circle, minimum size=40pt, fill=lightgray, line width=0.625pt, draw=black] at (50.0pt, -325.0pt) {\textcolor{black}{$G_2'[Y_{2,1}^+]$}};
\node (9) [circle, minimum size=40pt, fill=lightgray, line width=0.625pt, draw=black] at (225.0pt, -325.0pt) {\textcolor{black}{$G_2'[Y_{2,m_2}^+]$}};
\node (6) [circle, minimum size=40pt, fill=lightgray, line width=0.625pt, draw=black] at (100.0pt, -325.0pt) {\textcolor{black}{$G_2'[Y_{2,2}^+]$}};
\node (10)  at (200.0pt,-200.0pt) {\textcolor{black}{\LARGE $\dots$}};
\node (11)  at (160pt,-325.0pt) {\textcolor{black}{\LARGE $\dots$}};
\draw [line width=0.625, color=black] (1) to  (2);
\draw [line width=0.625, color=black] (1) to  (3);
\draw [line width=0.625, color=black] (1) to  (4);
\draw [line width=0.625, color=black] (1) to  (5);
\draw [line width=0.625, color=black] (15) to  (8);
\draw [line width=0.625, color=black] (15) to  (9);
\draw [line width=0.625, color=black] (15) to  (6);
\node at (97.5pt, -152.5pt) [rotate=32] {\textcolor{black}{ $B_1$}};
\node at (159.375pt, -148.125pt) [rotate=59] {\textcolor{black}{ $B_2$}};
\node at (249.375pt, -150.625pt) [rotate=306] {\textcolor{black}{ $B_{\ell-1}$}};
\node at (287.5pt, -140.0pt) [rotate=331] {\textcolor{black}{ $B_{\ell}$}};
\node at (86.25pt, -276.25pt) [rotate=45] {\textcolor{black}{ $B_{2,1}$}};
\node at (189.375pt, -276.875pt) [rotate=315] {\textcolor{black}{ $B_{2,m_\ell}$}};
\node at (106.25pt, -286.875pt) [rotate=66] {\textcolor{black}{ $B_{2,2}$}};

\end{tikzpicture}
}
    \caption{{\sc Kernelization algorithm-II} first finds a graph $\tilde{G}$ such that $(G,T)$ and $(\tilde{G},T)$ are equivalent as T-cycle instances. It then finds a protrusion decomposition of $\tilde{G}$ which is shown by the nodes in red. For each $i \in [\ell]$, $B_i$ is a vertex set that is shared by graphs $\tilde{G}[X_0]$ and $\tilde{G}[X_i^+]$. For every $i\in [\ell]$, we  obtain a graph $G_i'$  that is  $B_i$-linkage equivalent w.r.t. \tcycle to $\tilde{G}[X_i^+]$. For clarity, we show this transition only in the node of $\tilde{G}[X_2^+]$. Next, for every $i\in [\ell]$, we find the protrusion decomposition of $G_i'$. In the figure we only show the protrusion decomposition of $G_2'$ with nodes in grey. For each $j \in [m_2]$, $B_{2,j}$ is a vertex set that is shared by graphs $G_2'[Y_{2,0}]$ and $G_2'[Y_{2,j}^+]$.} \label{fig:kerneltwo}
\end{figure}
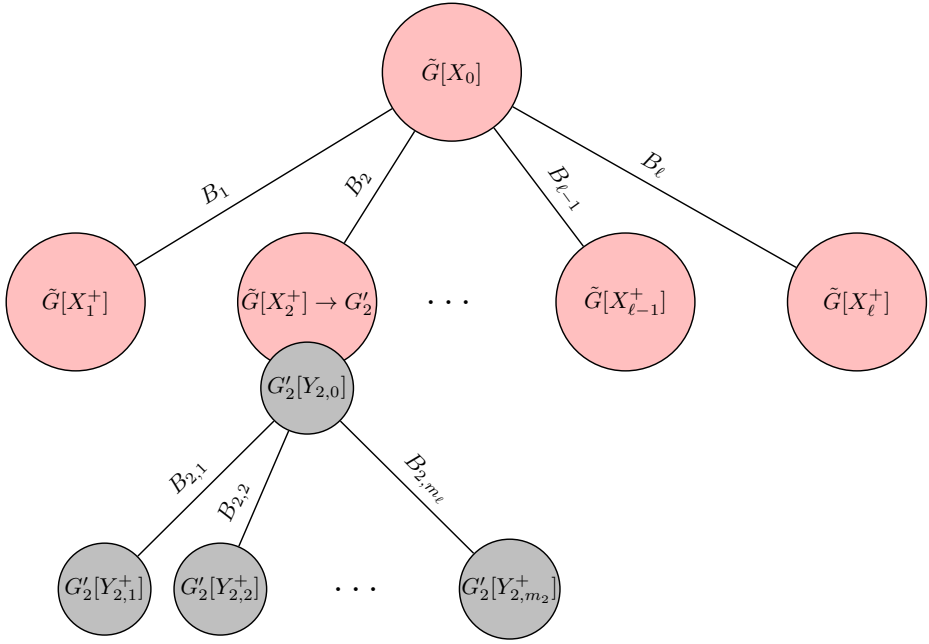

Steps 1 and 2 of {\sc Kernelization algorithm-II} are in common with Steps 1 and 2 of  {\sc Kernelization algorithm-I}. We now prove that Step 3 of {\sc Kernelization algorithm-II} is also correct.

\begin{lemma} \label{lem:linkirrelevant}
       The vertices in $\tilde{G}[X_i^+]$ that are $B_i$-linkage-irrelevant in Step 3 of {\sc Kernelization algorithm-II} are also irrelevant for the \tcycle problem in $G$. 
\end{lemma}
\begin{proof}
    This follows from the fact that  a solution to the {\sc $T$-Cycle} problem in $G$  can be written as a union  of paths in $\mathcal{P}_1 = \{S_1, \dots, S_t\}$ and $\mathcal{P}_2 = \{O_1, \dots, O_q\}$ where the endpoints of the paths in  $\mathcal{P}_1 \bigcup \mathcal{P}_2$ are in $B_i$,  the paths in $\mathcal{P}_1$ use edges that belong to $\tilde{G}[X_i]$ and the paths in $\mathcal{P}_2$ use edges outside of $\tilde{G}[X_i]$. The paths in $\mathcal{P}_2$ induce a matching $\mathcal{M} \subset B_i^2$. If a vertex $v$  is $B_i$-linkage-irrelevant, then there exists a way to connect the pairs in $\mathcal{M} $ without using the vertex $v$. Hence, there is a solution to the  {\sc $T$-Cycle} problem that does not use $v$, proving the claim.
\end{proof}

In the following lemma, we show that the instance $K'$ computed by {\sc Kernelization algorithm-II} is indeed a kernel. We end the section with \Cref{thm:main2} where we provide bounds on the size of the kernel and show that it runs in linear time. 

\begin{lemma} \label{lem:intercorrect}
       After the conclusion of Step 6 of {\sc Kernelization algorithm-II}, the {\sc $T$-Cycle} instance $(G,T)$ is equivalent to the {\sc $T$-Cycle} on the graph $K'$  where 
    \[ V(K') = V(G_0') \bigcup_{i\in [\ell]} V(G_{i,0}') \bigcup_{\substack{i\in [\ell] \\ j\in [m_i]}} V(H_{i,j}') \text{ and }  E(K') = E(G_0') \bigcup_{i\in [\ell]} E(G_{i,0}') \bigcup_{\substack{i\in [\ell] \\ j\in [m_i]}} E(H_{i,j}'). \] 
\end{lemma}
\begin{proof}
    It follows from \Cref{lem:linkirrelevant} that after the conclusion of Step 3 of {\sc Kernelization algorithm-II}, the {\sc $T$-Cycle} instance $(G,T)$ is equivalent to the {\sc $T$-Cycle} instance $(K,T)$ on the graph $K$  whose vertex set is given by $V(K) = \bigcup_{i=0}^\ell V(G_i')$ and whose edge set is given by $E(K) = \bigcup_{i=0}^\ell E(G_i')$.

    Moreover, we know from \Cref{C:twrednsubprot}  that $B_i'$ is a treewidth-$O(\log \log k)$-modulator  of graph $G_i'$.
    For every $i \in [\ell]$, this allows us to use Algorithm~\ref{alg:protcomp}, to compute an $O(\log k \, \log^2 \log k),$ $O(\log k \, \log^2 \log k),$ $ O(\log \log k))$-protrusion decomposition of $G_i'$, giving us a partition $\{Y_{i,0},\dots,Y_{i,m_i}\}$ of $V(G_i')$. As in {\sc Kernelization algorithm-II}, we use the notation $G_{i,j}' = G_i'[Y_{i,j}^+] $ and $B_{i,j} = Y_{i,j}^+ \setminus Y_{i,j}$ for $i\in[\ell], j \in [m_i]$, whereas $G_{i,0}' = G_i'[Y_{i,0}]$ for $i\in [\ell]$. Now, we use \Cref{thm:outline:polyKer} as an existential rather than an algorithmic result. That is, instead of kernelizing each of the graph $G_{ij}'$ using the  algorithm from \cite{DBLP:conf/focs/0001Z23}, we simply generate all possible graphs $H$ such that $|V(H)| = O(\log^{24} \log k)$ and $B_{i,j} \subseteq V(H)$. Using \Cref{thm:outline:polyKer}, we know that one of the generated graphs (denoted by $H_{i,j}'$) will be found to be  $B_{i,j}$-linkage equivalent w.r.t. \dispaths to $G_{i,j}'$ since    $G_{i,j}'$ is a planar graph of treewidth $O(\log \log k)$ and $B_{i,j} \subseteq V(G_{i,j}')$ is also of size $O(\log \log k)$.

    A solution to the {\sc $T$-Cycle} instance of $(K,T)$  can be written as a union  of paths in $\mathcal{P}_1 = \{S_1, \dots, S_t\}$ and $\mathcal{P}_2 = \{O_1, \dots, O_q\}$ where the endpoints of the paths in  $\mathcal{P}_1 \bigcup \mathcal{P}_2$ are in $B_{i,j}$,  the paths in $\mathcal{P}_1$ use edges that belong to $G[Y_{i,j}]$ and the paths in $\mathcal{P}_2$ use edges outside of $G[Y_{i,j}]$. Let $M = \{ \{u_1,v_1\},\dots,\{u_q,v_q\} \}$ be the respective endpoints of the paths in $\mathcal{P}_1$, and let $\mathcal{M}$ be the collection of all such sets $M$ found by varying the solutions to the \tcycle instance $(K,T)$. Then, $\mathcal{M} \subset B_{i,j}^2$. By the preceding exhaustive algorithm, $H_{i,j}'$ is $B_{i,j}$-linkage equivalent w.r.t. \dispaths to $G_{i,j}'$. So, there is also a way to connect all the pairs in $\mathcal{M}$ in $H_{i,j}'$.
Inductively, after replacing each graph $G_{i,j}'$ by an equivalent subgraph $H_{i,j}'$, for every $i\in [\ell]$ and $j \in [m_i]$, we obtain a new instance of the  {\sc $T$-Cycle} problem whose solution set is non-empty if and only if the solution set of the {\sc $T$-Cycle} instance $(K,T)$ is non-empty. The claim of the lemma follows.
\end{proof}

In order to establish runtime bounds for the algorithm, we  use two  results from prior works.

\begin{theorem}[Theorem 4 in \cite{adlertwo}] \label{thm:adlerres}
    Given a host graph $G$ with $|V(G)|=n$ embedded in a fixed surface,
a pattern $H$ with $|V(H)|=h$, and treewidth at most $k$, we can
decide whether $G$ contains a minor isomorphic to $H$ in $2^{O(k+h)+2k \cdot \log h} \cdot n$ time.
\end{theorem}

\begin{theorem}[Theorem 36 in \cite{cho2023parameterized}] \label{thm:chores}
    The \textsc{Planar Disjoint Paths} problem can be solved in $2^{O(k^{2})}n$
time, where $n$ is the size of the graph and $k$ is the number of
terminals. 
\end{theorem}

%\begin{theorem} \label{thm:mainkernel}
% {\sc Kernelization algorithm-II} correctly computes a polynomial kernel in linear time.
%\end{theorem}

\begin{algorithm}[h]
\caption{Computation of an $(\alpha,\beta,\gamma)$-protrusion decomposition of a planar graph $G$}
\label{alg:protcomp}
\SetKwInOut{Input}{Input}\SetKwInOut{Output}{Output}
\Input{Treewidth-$\eta$-modulator $S$ of planar graph $G$}
\Output{An $(O(\eta\cdot|S|)),O(\eta\cdot|S|)),O(\eta))$-protrusion decomposition of $G$}
\SetKwFunction{ProtDec}{Protrusion-Decomoposition}
\SetKwProg{myproc}{Procedure}{}{}
\myproc{\ProtDec{$S,G$}}{
{Compute a nice tree decomposition $(T,\chi)$ of $G-S$.}\\
{Choose an arbitrary root node $r$ for $T$.}\\
{Let $M$ denote the set of marked nodes. Initialize $M \gets \emptyset$.} \\
\Repeat{every component $C$ of $G[\chi(T-M)]$ has at most two neighbors in $S$}{
\Comment{For a node $v \in V(T)$, let $T_v$ denote the subtree of $T$ rooted at $v$.} \\
{Let $v$ be a lowermost node in $T$ such that some component $C$ of $G[\chi(T_v-M)]$ has at least three neighbors in $S$.} \\
{$M \gets M \bigcup \{v\}$.}
}
{$L \gets \textsc{LCA-closure}(M)$.}\\
{$X_0 \gets S \bigcup \chi(L)$.} \\
{Find the connected components $C_1,C_2,\dots,C_t$ of $G - X_0$.}\\
{Find a partition $X_1,\dots,X_\ell$ of $G - X_0$ by grouping components $C_1,C_2,\dots,C_t$ with the same neighborhood in $S$. That is, $C_i \subseteq X_k$ and $C_j \subseteq X_k$ if and only if $N(C_i)=N(C_j)$.} \\
{\textsc{return}  $X_0,\dots,X_\ell$.} \\
} 
\end{algorithm}

Finally, we are ready to prove the main result of this section. The pseudocode for the Protrusion decomposition is provided in Algorithm~\ref{alg:protcomp}.

\mainKernel*
\begin{proof}
    From \Cref{lem:intercorrect}, we know that $(K',T)$ is equivalent as a \textsc{T-Cycle} to the pair $(G,T)$. That is, $K'$ is indeed a kernel.

    First, we provide a bound on the size of the kernel.
    Note that $\tilde{G}[X_0]$  has  $O(k \log^2 k)$ vertices, and $\ell = O(k \log^2 k)$ by using \Cref{thm:protrusioncomp} in Step 2. 
    By using \Cref{thm:protrusioncomp} in Step 4, $G_{i,0}' = G_i'[Y_{i,0}]$ has $O(\log k \log^2 \log k)$ vertices. So, $\bigcup_{i\in [\ell]}V(G_{i,0}')$ has $O(k \log^3 \log^2 \log k)$ vertices. 
    Moreover, using  \Cref{thm:protrusioncomp} once again in Step 4, each $m_i$ for each $i\in [\ell]$ is  $O(\log k \log^2 \log k)$ and by construction $V(H_{i,j}')$ has $O(\log^{24} \log k)$ vertices for each $i\in [\ell]$ and $j\in [m_i]$. Therefore, $\bigcup_{\substack{i\in [\ell] \\ j\in [m_i]}} V(H_{i,j}')$ has  at most $k \log^3 k \log^{O(1)}\log k$ vertices. 
    Therefore, the kernel $K'$, where $V(K') = V(G_0') \bigcup_{i\in [\ell]} V(G_{i,0}') \bigcup_{\substack{i\in [\ell] \\ j\in [m_i]}} V(H_{i,j}') $, has  $O(k \log^4 k)$ vertices.

    We now prove that the algorithm runs in linear time. 
    \begin{enumerate}[(1.)]
        \item Using \Cref{thm:twreduction}, Step 1 runs in $k^{O(1)} \cdot n $ time.
        \item From \Cref{rem:prottime}, it follows that Step 2 runs in $k^{O(1)}\cdot n$ time.
        \item The algorithm from \Cref{C:twrednsubprot} runs in $\log^{O(1)}k  \cdot n $ time since $B_i'$ is  $ O(\log k\log \log k)$ for each $i \in [\ell]$. As noted before $\ell\in O(k \log^2 k)$. Hence, Step 3 runs in $\log^{O(1)}k  \cdot n $ time. 
        \item For Step 5.1 guessing a graph in $H \in \mathcal{G}_{i,j}$ involves guessing the edges of $H$. This means that there are  $O(2^{\log^{48}\log k})$ (which is $O(k)$) number of graphs in $\mathcal{G}_{i,j}$  for every $i \in [\ell]$ and $j \in [m_i]$. Hence, the total cost of Step 5.1 is at most $k^2\log^{O(1)}k $.
        \item Using Theorem 4 in \cite{adlertwo}, the dynamic programming algorithm for \textsc{Minor Containment} in Step 5.2 takes $2^{O(\log\log k+\log^{24}\log k)+2\log\log k\cdot\log\log^{24}\log k}\cdot n$ which is $O(k) \cdot  n$ time. 

        Using Theorem 36 in \cite{cho2023parameterized}, a single run of the dynamic programming algorithm for \dispaths costs $2^{O(\log^{2}\log k)}\cdot n$ which is $O(k)\cdot n$ time.
        In Step 5.2, $B_{i,j}$ is in $O(\log^2\log k)$. Thus, there are  $2^{O(\log^4\log k)} = O(k)$ different pairings of terminals for which the \dispaths algorithm needs to run. Running the over all pairs, the cost is $O(k^2) \cdot n$ time.

        Finally, the two substeps above are executed for every $i \in [\ell]$ and $j \in [m_i]$ and for each of the subgraphs $H$ found in Step 5.2. Nonetheless, the total cost of Step 5.2 over the course of the algorithm is bounded by $k^{O(1)} \cdot n$ time.
    \end{enumerate}
    Combining the cost of all the steps, the algorithm runs in $k^{O(1)} \cdot n$ time. 
\end{proof}

\section{Conclusion} \label{sec:conc}
In this paper, we considered the \tcycle problem on planar graphs. 
We begin by providing an FPT algorithm with running time $2^{O(\sqrt{k} \log k)}\cdot n$. Most subexponential time FPT algorithms on planar graphs use Bidimensionality in a naive fashion: If the treewidth of the input graph is large, then we trivially know whether our instance is a yes or no instance, and when the treewidth is small (sublinear), we use dynamic programming using the tree decomposition. Unfortunately, this fails on \tcycle (and related terminal-based routing problems) as large treewidth neither blocks a $T$-loop nor guarantees it. We bypass this roadblock using a clever rerouting argument and new variations of several techniques in the literature to deal with planar graphs. Further, to the best of our knowledge, {\sc $T$-Cycle} is the first natural terminal-based routing problem that admits a subexponential algorithm for planar graphs. We believe that this may pave the road for new subexponential algorithms for planar graphs for problems that cannot be solved by naive application of bidimensionality. Further, we perform non-trivial work to reduce the running time dependence on $n$ to linear.

Next, we provided a linear time kernelization algorithm for {\sc $T$-Cycle} on planar graphs of size $k\cdot \log^{O(1)}k$. As pointed in the Introduction section, the kernelization complexity of \tcycle on general graphs is one of the biggest open problems in the field. Our result shows that for planar graphs, we  have an almost linear kernel. This raises the natural and important question as to whether our kernel for {\sc $T$-Cycle} can be lifted to more general non-planar graphs, or even to general graphs. To obtain our kernel, we use  Reed's technique of plane cutting. This is perhaps the first application of this technique in kernelization. Further, we obtain the kernelization to be almost-optimal combining Theorem~\ref{thm:twreduction} with a ``nested'' protrusion decomposition technique, which is  novel to this paper. Perhaps these techniques can be combined to obtain new/improved kernelization results for other terminal-based problems. 

As discussed in~\cite{DBLP:conf/birthday/Lokshtanov0Z20}, algorithms for terminal-based routing problems on planar graphs serve as a crucial step in designing algorithms for general case. Indeed, this was the case for \textsc{Disjoint Paths}. All known algorithms for \textsc{Disjoint Paths} are based on distinguishing the cases when the graph contains a large clique as a minor and when it does not. Furthermore, when the input graph does not contain a large clique as a minor, the graph either has  low treewidth (leading towards a DP) or contains a large \textit{flat wall} that necessitates the study of these problems on planar and ``almost-planar'' graph classes.  As mentioned in the Introduction, it would be interesting  to see if it is possible to extend our results for ``almost-planar graphs''.

Finally, we wish to point out that our arguments may be relevant for the resolution of the problem on general graphs as well. In this direction, the first step should be to generalize our results to flat walls towards the resolution of \tcycle on minor-free graphs. %Indeed, for the \textsc{Disjoint Paths} problem, the resolution of the planar case was a cornerstone towards the resolution of the case of general graphs (this is explained in~\cite{DBLP:conf/birthday/Lokshtanov0Z20}). Indeed, to design combinatorial algorithms for routing problems on general graphs, one often uses arguments based on case analysis, where either we have a large clique as a minor, or a large ``almost-planar piece'', and in the latter case, arguments used to solve the problem on planar graphs become very useful.

% We have the following procedure depending on the case we are in.
% \begin{enumerate}
%     \item Case 1: In this case, we apply the simple ``\textit{unfolding procedure}'' as follows to provide a $T$-loop $L'$ that do not use segments $S_1$ and $S_2$. We connect the vertex $u_1$ to $u_2$ via the $u_1,u_2$-path along $C_r$ that excludes $v_1$, and similarly, connect $v_1$ to $v_2$ via the $v_1,v_2$-path along $C_r$ that excludes $u_1$. See Figure~\ref{fig:T1Mod} for an illustration. 
%     \begin{figure}
%         \centering
%         \includegraphics{T1Mod.pdf}
%         \caption{Case1}
%         \label{fig:T1Mod}
%     \end{figure}
% \end{enumerate}

\section*{Acknowledgments}
Abhishek Rathod and Meirav Zehavi are supported by the European Research Council (ERC) project titled PARAPATH (101039913) and by the ISF grant with number ISF--1470/24. Harmender Gahlawat was a postdoc at BGU when this project started and was supported by ERC grant titled PARAPATH (101039913), and was a postdoc at G-SCOP, Grenoble-INP when most of this work was carried out.

\bibliographystyle{plainurl} 
% \bibliography{sample}

\bibliography{main}

\newpage
\appendix

\end{document}